\documentclass[a4paper,11pt,english]{article}
\pdfoutput=1

\usepackage[utf8]{inputenc}
\usepackage[english]{babel}
\usepackage[math]{kurier}
\usepackage[T1]{fontenc}
\usepackage{fullpage}

\usepackage{verbatim}
\usepackage{ifdraft}
\ifdraft{\newcommand{\authnote}[3]{{\color{#3} {\bf  #1:} #2}}}{\newcommand{\authnote}[3]{}}
\newcommand{\anote}[1]{\authnote{ András}{#1}{green}}
\newcommand{\snote}[1]{\authnote{ Stacey}{#1}{red}}
\newcommand{\shnote}[1]{\authnote{ Shantanav}{#1}{blue}}

\usepackage{tikz}
\usetikzlibrary{calc,patterns}
\usepackage{ifthen}
\usepackage{amsmath,amssymb,amsthm,mathtools}
\usepackage{dsfont}
\usepackage{bbm}
\usepackage{mleftright}\mleftright
\usepackage{thmtools} 
\usepackage{thm-restate}

\usepackage{algorithm}
\usepackage{algcompatible}
\usepackage{enumitem}
\usepackage{float}

\usepackage{caption}
\usepackage{subcaption}

\usepackage{titling}

\newcommand{\eps}{\varepsilon}

\newcommand{\ketbra}[2]{|#1\rangle\! \langle #2|}

\newcommand{\nrm}[1]{\left\lVert #1 \right\rVert}
\newcommand{\bigO}[1]{\mathcal{O}\left( #1 \right)}
\newcommand{\bigOt}[1]{\widetilde{\mathcal{O}}\left( #1 \right)}
\newcommand{\littleo}[1]{o\left( #1 \right)}

\newcommand{\pvp}{\vec{p}{\kern 0.45mm}'}
\let\oldnabla\nabla
\renewcommand{\nabla}{\oldnabla\!}

\DeclarePairedDelimiter\ceil{\lceil}{\rceil}
\DeclarePairedDelimiter\bra{\langle}{\rvert}
\DeclarePairedDelimiter\ket{\lvert}{\rangle}
\DeclarePairedDelimiterX\braket[2]{\langle}{\rangle}{#1 \delimsize\vert #2}
\newcommand{\underflow}[2]{\underset{\kern-60mm \overbrace{#1} \kern-60mm}{#2}}

\def\polylog{\mathrm{polylog}}

\long\def\ignore#1{}

\newtheorem{theorem}{Theorem}
\newtheorem{corollary}[theorem]{Corollary}
\newtheorem{lemma}[theorem]{Lemma}

\newtheorem{definition}[theorem]{Definition}
\newtheorem{claim}[theorem]{Claim}

\newcommand{\Hi}{\ensuremath{\mathcal{H}}}

\newcommand{\A}{\ensuremath{\mathcal{A}}}
\newcommand{\B}{\ensuremath{\mathcal{B}}}
\newcommand{\C}{\ensuremath{\mathbb{C}}}
\newcommand{\N}{\ensuremath{\mathbb{N}}}
\newcommand{\Oo}{\ensuremath{\mathcal{O}}}

\newcommand{\R}{\ensuremath{\mathbb{R}}}

\usepackage{tikz}

\usetikzlibrary{backgrounds,fit,decorations.pathreplacing,calc}

\usepackage{qcircuit}

\usepackage[final]{hyperref}
\hypersetup{
	colorlinks = true,
	allcolors = {blue},
	breaklinks=true,
}

\title{
  The power of block-encoded matrix powers: improved regression techniques via faster Hamiltonian simulation
}
\author{
Shantanav Chakraborty\thanks{QuIC, Université Libre de Bruxelles. Supported by the Belgian Fonds de la Recherche Scientifique - FNRS under grants no F.4515.16 (QUICTIME) and R.50.05.18.F (QuantAlgo). \texttt{shchakra@ulb.ac.be}}
	\and	
András Gilyén\thanks{QuSoft/CWI, Science Park 123, 1098 XG Amsterdam, Netherlands. Supported by ERC Consolidator Grant QPROGRESS.  \texttt{gilyen@cwi.nl} } 
	\and	
Stacey Jeffery\thanks{QuSoft/CWI, Science Park 123, 1098 XG Amsterdam, Netherlands. Supported by an NWO WISE Grant and an NWO Veni Grant. \texttt{jeffery@cwi.nl}}
}
\date{\today\vspace{-5mm}}
\begin{document}
	
	\maketitle

\begin{abstract}
We apply the framework of block-encodings, introduced by Low and Chuang (under the name standard-form), to the study of quantum machine learning algorithms and derive general results that are applicable to a variety of input models, including sparse matrix oracles and matrices stored in a data structure. We develop several tools within the block-encoding framework, such as singular value estimation of a block-encoded matrix, and quantum linear system solvers using block-encodings. The presented results give new techniques for Hamiltonian simulation of non-sparse matrices, which could be relevant for certain quantum chemistry applications, and which in turn imply an exponential improvement in the dependence on precision in quantum linear systems solvers for non-sparse matrices. 

In addition, we develop a technique of variable-time amplitude \emph{estimation}, based on Ambainis' variable-time amplitude amplification technique, which we are also able to apply within the framework. 

As applications, we design the following algorithms: (1) a quantum algorithm for the quantum weighted least squares problem, exhibiting a 6-th power improvement in the dependence on the condition number and an exponential improvement in the dependence on the precision over the previous best algorithm of Kerenidis and Prakash; (2) the first quantum algorithm for the quantum generalized least squares problem; and (3) quantum algorithms for estimating electrical-network quantities, including effective resistance and dissipated power, improving upon previous work. 
	\end{abstract}

\newpage
%\ifdraft
{\tableofcontents
\newpage}

	\section{Introduction}
%	We build on ideas presented in \cite{LowChuangQubitization2016}. Our generalization is motivated by \cite{kerenidis:quantumgraddescent}.

A rapidly growing and important class of quantum algorithms are those that use Hamiltonian simulation subroutines to solve linear algebraic problems, many with potential applications to machine learning. This subfield began with the HHL algorithm, due to Harrow, Hassidim and Lloyd \cite{harrow2009quantum}, which solves the \emph{quantum linear system problem} (QLS problem). In this problem, the input consists of a matrix $A\in\mathbb{R}^{N\times N}$ and a vector $\vec{b}\in\mathbb{R}^N$, in some specified format, and the algorithm should output a quantum state proportional to $\sum_{i=1}^Nx_i\ket{i}$, where $\vec{x}=A^{-1}\vec{b}$. 

The format in which the input is presented is of crucial importance. For a sparse $A$, given an efficient algorithm to query the $i$-th non-zero entry of the $j$-th row of $A$, the HHL algorithm and its subsequent improvements \cite{AmbainisVariableTime12,childs2015quantum} can solve the QLS problem in complexity that depends poly-logarithmically on $N$. Here, if $A$ were given naively as a list of all its entries, it would generally take time proportionally to $N^2$ just to read the input. We will refer to this model of accessing $A$, in which we can query the $i$-th non-zero entry of the $j$-th row, as the \emph{sparse-access input model}.\footnote{If the matrix is not symmetric (or Hermitian) we also assume access to its transpose in a similar fashion.}

In \cite{kerenidis2016quantum} and \cite{kerenidis:quantumgraddescent}, Kerenidis and Prakash consider several linear algebraic problems in a different input model. They assume that data has been collected and stored in some carefully chosen data structure in advance. If the data is described by an arbitrary $N\times N$ matrix, then of course, this collection will take time at least $N^2$ (or, if the matrix is sparse, at least the number of non-zero entries). However, processing the data, given such a data structure, is significantly cheaper, depending only poly-logarithmically on $N$. Kerenidis and Prakash describe 
a data structure that, when stored in quantum-random-access read-only memory (QROM)\footnote{This refers to memory that is only required to store classical (non-superposition) data, but can be addressed in superposition.}, allows for the preparation of a superposition over $N$ data points in complexity poly-logarithmic in $N$. We call this the \emph{quantum data structure input model} and discuss it more in Section~\ref{sec:data-structure}. Although in some applications it might be too much to ask for the data to be presented in such a structure, one advantage of this input model is that it is not restricted to sparse matrices. This result can potentially also be useful for some quantum chemistry applications, since a recent proposal of Babbush et al.~\cite{BabbushSecondQuantizedChem15} uses a database of all Hamiltonian terms in order to simulate the electronic structure.

The HHL algorithm and its variants and several other applications are based on techniques from Hamiltonian simulation. Given a Hermitian matrix $H$ and an input state $\ket{\psi}$, the Hamiltonian simulation problem is to simulate the unitary $e^{iH}$ on $\ket{\psi}$ for some time $t$. Most work in this area has considered the sparse-access input model \cite{Llo96,AT03,BACS05,BC09,berry:exponentialprecision,berry:simHamTaylor,Chi04,Chi10,CW12,PQSV11,WBHS11,BerryChilds:hamsimFOCS,NB16,BN16}, but recent work of Low and Chuang \cite{LowChuangQubitization2016} has considered a different model, which we call \emph{the block-encoding framework}\footnote{Low and Chuang call this input model \emph{standard form}.}. 

\paragraph{The block-encoding framework.} A block-encoding of a matrix $A\in\mathbb{C}^{N\times N}$ is a unitary $U$ such that the top left block of $U$ is equal to $A/\alpha$ for some normalizing constant $\alpha\geq \nrm{A}$:
$$ U=\left(\begin{array}{cc} A/\alpha & . \\ . & .\end{array}\right).$$
In other words, for some $a$, for any state $\ket{\psi}$ of appropriate dimension, $\alpha (\bra{0}^{\otimes a}\otimes I)U(\ket{0}^{\otimes a}\otimes \ket{\psi}) = A\ket{\psi}$.

Such an encoding is useful if $U$ can be implemented efficiently. In that case, $U$, combined with amplitude amplification, can be used to generate the state $A\ket{\psi}/\nrm{A\ket{\psi}}$ given a circuit for generating $\ket{\psi}$. The main motivation for using block-encodings is that Low and Chuang showed~\cite{LowChuangQubitization2016} how to perform optimal Hamiltonian simulation given a block-encoded Hamiltonian $A$. 

In Ref.~\cite{kerenidis2016quantum}, Kerenidis and Prakash implicitly prove that if an $N\times N$ matrix $A$ is given as a quantum data structure, then there is an $\eps$-approximate block-encoding of $A$ that can be implemented in complexity polylog$(N/\eps)$. This implies that all results about block-encodings ---
 including Low and Chuang's Hamiltonian simulation when the input is given as a block-encoding \cite{LowChuangQubitization2016}, and other techniques we develop in this paper --- 
also apply to input presented in the quantum data structure model. This observation is the essential idea behind our applications. Implicit in work by Childs \cite{Chi10} is the fact that, given $A$ in the sparse-access input model, there is an $\eps$-approximate block-encoding of $A$ that can be implemented in complexity polylog$(N/\eps)$, so our results also apply to the sparse-access input model. In fact, the block-encoding framework unifies a number of possible input models, and also enables one to work with hybrid input models, where some matrices may come from purifications of density operators, whereas other input matrices may be accessed through sparse oracles or a quantum data structure. For a very recent overview of these general techniques see e.g.~\cite{gilyenBlockMatrices}.

We demonstrate the elegance of the block-encoding framework by showing how to combine and modify block-encodings to build up new block-encodings, similar to building new algorithms from existing subroutines. For example, given block-encodings of $A$ and $B$, their product yields a block-encoding of $AB$. Given a block-encoding of a Hermitian $A$, it is possible to construct a block-encoding of $e^{iA}$, using which one can implement a block-encoding of $A^{-1}$. We summarize these techniques in Section \ref{sec:intro-tools}, and present them formally in Section \ref{sec:tools}.

To illustrate the elegance of the block-encoding framework, consider one of our applications: generalized least squares. This problem, defined in Section \ref{sec:intro-app}, requires that given inputs $X\in\mathbb{R}^{M\times N}$, $\Omega\in \mathbb{R}^{M\times M}$ and $\vec{y}\in\mathbb{R}^M$, we output a quantum state proportional to 
$$\vec{\beta} = (X^T\Omega^{-1}X)^{-1}X^T\Omega^{-1}\vec{y}.$$
Given block-encodings of $X$ and $\Omega$, it is simple to combine them to get a block-encoding of $(X^T\Omega^{-1}X)^{-1}X^T\Omega^{-1}$, which can then be applied to a quantum state proportional to $\vec{y}$. 

\paragraph{Variable-time amplitude estimation.} A variable-stopping-time quantum algorithm is a quantum algorithm $\cal A$ consisting of $m$ stages ${\cal A}={\cal A}_m\dots {\cal A}_1$, where ${\cal A}_j{\cal A}_{j-1}\dots {\cal A}_1$ has complexity $t_j$, for $t_m> \dots > t_1>0$. At each stage, a certain flag register, which we can think of as being initialized to a neutral symbol, may be marked as ``good'' in some branches of the superposition, or ``bad'' in some branches of the superposition, or left neutral. Each subsequent stage only acts non-trivially on those branches of the superposition in which the flag is not yet set to ``good'' or ``bad''. 

At the end of the algorithm, we would like to project onto that part of the final state in which the flag register is set to ``good''. This is straightforward using amplitude amplification, however this approach may be vastly sub-optimal. If the algorithm terminates with amplitude $\sqrt{p_{succ}}$ on the ``good'' part of the state, then standard amplitude amplification requires that we run $1/\sqrt{p_{succ}}$ rounds, each of which requires us to run the full algorithm $\cal A$ to generate its final state, costing $t_m/\sqrt{p_{succ}}$.

To see why this might be sub-optimal, suppose that after ${\cal A}_1$, the amplitude on the part of the state in which the flag register is set to ``bad'' is already very high. Using amplitude amplification at this stage is very cheap, because we only have to incur the cost $t_1$ of ${\cal A}_1$ at each round, rather than running all of $\cal A$. In \cite{AmbainisVariableTime12}, Ambainis showed that given a variable-stopping-time quantum algorithm, there exists an algorithm that approximates the ``good'' part of the algorithm's final state in cost $\bigOt{t_m+\sqrt{\sum_{j=1}^m\frac{p_j}{p_{succ}}t_j^2 }}$\footnote{We use the notation $\bigOt{f(x)}$ to indicate $\bigO{f(x)\polylog(f(x))}$.}, where $p_j$ is the amplitude on the part of the state that is moved from neutral to ``good'' or ``bad'' during application of ${\cal A}_j$ (intuitively, the probability that the algorithm stops at stage $j$). 

%To see why this might be sub-optimal, suppose that after ${\cal A}_1$, the amplitude on the part of the state in which the flag register is set to ``bad'' is already very high. Using amplitude amplification at this stage is very cheap, because we only have to incur the cost $t_1$ of ${\cal A}_1$ at each round, rather than running all of $\cal A$. In \cite{AmbainisVariableTime12}, Ambainis showed that given a variable-stopping-time quantum algorithm there exists an algorithm that approximates the ``good'' part of the algorithm's final state in cost $\bigOt{t_m+\sqrt{\sum_{j=1}^m\frac{p_j}{p_{succ}}t_j^2 }}$\footnote{We use the notation $\bigOt{f(x)}$ to indicate $\bigO{f(x)\polylog(f(x))}$.}, where $p_j$ is the amplitude on the part of the state that is moved from neutral to ``good'' or ``bad'' during application of ${\cal A}_j$ (intuitively, the probability that the algorithm stops at stage $j$). 

While amplitude amplification can easily be modified to not only project a state onto its ``good'' part, but also return an estimate of $p_{succ}$ (i.e. the probability of measuring ``good'' given the output of ${\cal A}$), this is not immediate in variable-time amplitude amplification. The main difficulty is that a variable-time amplification algorithm applies a lot of subsequent amplification phases, where in each amplification phase the precise amount of amplification is a priori unknown. We overcome this difficulty by separately estimating the amount of amplification in each phase with some additional precision and finally combining the separate estimates in order to get a multiplicative estimate of $p_{succ}$.

In Section \ref{sec:vtae}, we show in detail how to estimate the success probability of a variable-stopping-time quantum algorithm to within a multiplicative error of $\eps$ in complexity 
$$\bigOt{\frac{1}{\eps}\left(t_m+\sqrt{\sum_{j=1}^m\frac{p_j}{p_{succ}}t_j^2 }\right)}.$$
Meanwhile we also derive some logarithmic improvements to the complexity of variable-time amplitude amplification. 
%To overcome the aforementioned difficulties, we develop a technique that we call \emph{mindful-amplification}, which allows us to amplify the amplitude on the ``good'' part of a quantum state to constant, while at the same time, computing an estimate of how much the state has been amplified.

\paragraph{Applications.} We give several applications of the block-encoding framework and variable-time amplitude estimation. 

We first present a quantum weighted least squares solver (WLS solver), which outputs a quantum state proportional to the optimal solution to a weighted least squares problem, when the input is given either in the quantum data structure model of Kerenidis and Prakash, or the sparse-access input model. We remark that the sparse-access input model is perhaps less appropriate to the setting of data analysis, where we cannot usually assume any special structure on the input data, however, since our algorithm is designed in the block-encoding framework, it works for either input model. Our quantum WLS solver improves the dependence on the condition number from $\kappa^6$ in \cite{kerenidis:quantumgraddescent}\footnote{In the paper of Kerenidis and Prakash their $\kappa$ corresponds to our $\kappa^2$.} to $\kappa$, and the dependence on $\eps$ from $1/\eps$ to $\polylog(1/\eps)$. 

We next present the first quantum generalized least squares solver (GLS solver), which outputs a quantum state proportional to the optimal solution to a generalized least squares problem. We again assume that the input is given in either the quantum data structure model or the sparse-access model. The complexity is again polynomial in $\log(1/\eps)$ and in the condition numbers of the input matrices. 

Finally, we also build on the algorithms of Wang \cite{wang2017efficient} to estimate effective resistance between two nodes of an electrical network and the power dissipated across a network when the input is given as a quantum data structure or in the sparse-access model. We estimate the norm of the output state of a certain linear system by applying the variable-time amplitude estimation algorithm. In the sparse-access model, we find that our algorithm outperforms Wang's linear-system-based algorithm. In the quantum data structure model, our algorithms offer a speedup whenever the maximum degree of an electrical network of $n$ nodes is $\Omega(n^{1/3})$. Our algorithms also have a speedup over the quantum walk based algorithm by Wang in certain regimes.

We describe these applications in more detail in Section \ref{sec:intro-app}, and formally in Section \ref{sec:app}.

\paragraph{Related Work.} Independently of this work, recently, Wang and Wossnig \cite{WW18} have also considered Hamiltonian simulation of a Hamiltonian given in the quantum data structure model, using quantum-walk based techniques from earlier work on Hamiltonian simulation \cite{BerryChilds:hamsimFOCS}. Their algorithm's complexity scales as $\nrm{A}_1$ (which they upper bound by $\sqrt{N}$); 
whereas 
our Hamiltonian simulation results (Theorem \ref{thm:dataHamSim}), which follow from Low and Chuang's block-Hamiltonian simulation result, have a complexity that depends poly-logarithmically on the dimension, $N$. Instead, our complexity depends on the parameter $\mu$, described below, which is also at most $\sqrt{N}$. In principle, the Hamiltonian simulation result of \cite{WW18} can also be used to implement a quantum linear systems solver, however the details are not worked out in \cite{WW18}. 

\subsection{Techniques for block-encodings}\label{sec:intro-tools}

We develop several tools within the block-encoding framework that are crucial to our applications, but also likely of independent interest. Since an input given either in the sparse-access model or as a quantum data structure can be made into a block-encoding, our block-encoding results imply analogous results in each of the sparse-access and quantum data structure models. 

In the following, let $\mu(A)$ be one of: (1) $\mu(A)=\nrm{A}_F$, the Frobenius norm of $A$, in which case the quantum data structures should encode $A$;
or (2) for some $p\in [0,1]$, $\mu(A)=\sqrt{s_{2p}(A)s_{2(1-p)}(A)}$, where $s_{p}(A)=\max_j\nrm{A_{j,\cdot}}_p^p$, in which case the quantum data structures should encode both $A^{(p)}$ and $(A^{(1-p)})^T$, defined by $A_{i,j}^{(q)}:=(A_{i,j})^q$. 

\paragraph{Hamiltonian simulation from quantum data structure.} In Theorem \ref{thm:dataHamSim}, we prove the following. 
Given a quantum data structure for a Hermitian $A\in\mathbb{C}^{N\times N}$ such that $\nrm{A}\leq 1$, we can implement $e^{it A}$ in complexity $\bigOt{t\mu(A)\polylog(N/\eps)}$. This follows from the quantum Hamiltonian simulation algorithm of Low and Chuang that expects the input as a block-encoding. 
 Independently, Wang and Wossnig have proven a similar result, with $\nrm{A}_1\leq \sqrt{N}$ in place of~$\mu(A)$~\cite{WW18}.

\paragraph{Quantum singular value estimation.} In Ref.~\cite{kerenidis2016quantum}, Kerenidis and Prakash give a quantum algorithm for estimating the singular values of a matrix stored in a quantum data structure. In Theorem~\ref{thm:sve}, we present an algorithm for singular value estimation of a matrix given as a block-encoding. In the special case when the block-encoding is implemented by a quantum data structure, we recover the result of \cite{kerenidis2016quantum}. 

\paragraph{Quantum linear system solver.} Given a block-encoding of $A$, which is a unitary $U$ with $A/\alpha$ in the top left corner, for some $\alpha$, we can implement a block-encoding of $A^{-1}$ (Lemma \ref{lem:negative-power-restated}), which is a unitary $V$ with $A^{-1}/(2\kappa)$ in the top left corner, where $\kappa$ is an upper bound on\footnote{In the special case when $\nrm{A}=1$, $\kappa$ is an upper bound on the condition number of $A$, justifying the notation.} $\nrm{A^{-1}}$. Such a block-encoding can be applied to a state $\ket{b}$ to get $\frac{1}{2\kappa}\ket{0}^{\otimes a}(A^{-1}\ket{b})+\ket{0^\bot}$ for some unnormalized state $\ket{0^\bot}$ orthogonal to every state with $\ket{0}$ in the first $a$ registers. Performing amplitude amplification on this procedure, we can approximate the state $A^{-1}\ket{b}/\nrm{A^{-1}\ket{b}}$. However, this gives quadratic dependence on the condition number of $A$, whereas only linear dependence is needed for quantum linear systems solvers in the sparse-access input model, thanks to the technique of variable-time amplitude amplification. Using this technique, we are able to show, in Theorem~\ref{thm:qls-vtaa}, that given a block-encoding of $A$, we can approximate the state $A^{-1}\ket{b}/\nrm{A^{-1}\ket{b}}$ in time 
$$\bigOt{\alpha\kappa T_U\log^3(1/\eps)+\kappa T_b\log(1/\eps)},$$
where $T_U$ is the complexity of implementing the block-encoding of $A$, and $T_b$ is the cost of a subroutine that generates $\ket{b}$.
From this, it follows that, given $A$ in a quantum data structure, we can implement a QLS solver in complexity 
$$\bigOt{\kappa\mu(A)\polylog(N/\eps)},$$
which we prove in Theorem \ref{thm:QLS-data-structure}. 

Using our new technique of variable-time amplitude estimation, in Corollary \ref{cor:qls-vtae}, we also show how to compute a $(1\pm\eps)$-multiplicative estimate of $\nrm{A^{-1}\ket{b}}$ in complexity 
$$\bigOt{\frac{\kappa}{\eps}(\alpha T_U+T_b)}.$$

\paragraph{Negative powers of Hamiltonians.} Finally, we generalize our QLS solver to apply $A^{-c}$ for any $c\in(0,\infty)$. Using variable-time amplification techniques we show in Theorem~\ref{thm:negative-power-vtaa} that, if we have access to a unitary $U$ that block-encodes of $A$, such that $A/\alpha$ is the top left block of $U$, and $U$ can be implemented in cost $T_U$, then we can generate $A^{-c}\ket{b}/\nrm{A^{-c}\ket{b}}$ in complexity
$$\bigO{\alpha q\kappa^q T_U\log^3(1/\eps)+\kappa^qT_b\log(1/\eps)}$$
where $q=\max\{1,c\}$ and $T_b$ is the complexity of prepare $\ket{b}$.

\subsection{Application to least squares}\label{sec:intro-app}

\paragraph{Problem statements.} The problem of \emph{ordinary least squares (OLS)} is the following. 
Given data points $\{(\vec{x}^{(i)},y^{(i)})\}_{i=1}^M$ for $\vec{x}^{(1)},\dots,\vec{x}^{(M)}\in \mathbb{R}^N$ and $y^{(1)},\dots,y^{(M)}\in\mathbb{R}$, find $\vec{\beta}\in \mathbb{R}^N$ that minimizes:
\begin{equation}
\sum_{i=1}^M(y^{(i)}-\vec{\beta}^T\vec{x}^{(i)})^2.\label{eq:olsIntro}
\end{equation}
The motivation for this task is the assumption that the samples are obtained from some process such that at every sample $i$, $y^{(i)}$ depends linearly on $\vec{x}^{(i)}$, up to some random noise, so $y^{(i)}$ is drawn from a random variable $\vec{\beta}^T\vec{x}^{(i)}+E_i$, where $E_i$ is a random variable with mean 0, for example, a Gaussian. The vector $\vec{\beta}$ that minimizes \eqref{eq:olsIntro} represents a good estimate of the underlying linear function. We assume $M\geq N$ so that it is feasible to recover this linear function. 

%In particular, if $X\in \mathbb{R}^{M\times N}$ is the matrix with $\vec{x}^{(i)}$ as its $i$-th row, for each $i$, and $\vec{y}\in\mathbb{R}^M$ has $y^{(i)}$ as its $i$-th entry, assuming $X^TX$ is invertible, the optimal $\vec{\beta}$ satisfies:
%$$\vec{\beta}=(X^TX)^{-1}X^T\vec{y}.$$
%The assumption that $X^TX$ is invertible, or equivalently, that $X$ has rank $N$, is very reasonable, and is generally used in least squares algorithms. This is because $X^TX \in \mathbb{R}^{N\times N}$ is a sum of $M\geq N$ terms, and so it is unlikely to have rank less than $N$. 

%The problem of quantum OLS is the following. 
%\begin{problem}[quantum OLS]
%Given quantum-accessible data structures for 
%\end{problem}
%We will consider variants of this problem with slightly different data structures. 

We can generalize this task to settings in which certain samples are thought to be of higher quality than others, for example, because the random variables $E_i$ are not identical. We express this belief by assigning a positive weight $w_i$ to each sample, and minimizing
\begin{equation}
\sum_{i=1}^Mw_i(y^{(i)}-\vec{\beta}^T\vec{x}^{(i)})^2.\label{eq:wls}
\end{equation}
Finding $\vec{\beta}$ given $X$, $\vec{w}$ and $\vec{y}$ is the problem of \emph{weighted least squares (WLS)}.

We can further generalize to settings in which the random variables $E_i$ for sample $i$ are correlated. In the problem of \emph{generalized least squares (GLS)}, the presumed correlations in error between pairs of samples are given in a symmetric non-singular covariance matrix $\Omega$. We then want to find the vector $\vec{\beta}$ that minimizes
\begin{equation}
\sum_{i,j=1}^M\Omega_{i,j}^{-1}(y^{(i)}-\vec{\beta}^T\vec{x}^{(i)})(y^{(j)}-\vec{\beta}^T\vec{x}^{(j)}).\label{eq:gls}
\end{equation}

We will consider solving \emph{quantum} versions of these problems. Specifically, a \emph{quantum WLS solver} (resp. \emph{quantum GLS solver}) is given access to $\vec{y}\in\mathbb{R}^M$, $X\in \mathbb{R}^{M\times N}$, and positive weights $w_1,\dots,w_M$ (resp. $\Omega$), in some specified manner, and outputs an $\eps$-approximation of a quantum state $\sum_i\beta_i\ket{i}/\nrm{\vec{\beta}}$, where $\vec\beta$ minimizes the expression in \eqref{eq:wls} (resp. \eqref{eq:gls}).

\paragraph{Prior work.} Quantum algorithms for least squares fitting were first considered in \cite{wiebe2012quantum}. They considered query access to $X$, and a procedure for outputting $\ket{y}=\sum_iy_i\ket{i}/\nrm{\vec{y}}$, which we refer to as the sparse-access input model. They present a \emph{quantum} OLS solver, outputting a state proportional to a solution $\vec\beta$, %and estimate the quality of the fit 
that runs in time $\bigOt{\min\{\log(M)s^3\kappa^6/\eps,\log(M)s\kappa^6/\eps^2\}}$, where $s$ is the sparsity of $X$, and $\kappa$ the condition number. To compute a state proportional to $\vec\beta$, they first apply $X^T$ to $\ket{y}$ to get a state proportional to $X^T\vec y$, using techniques similar to \cite{harrow2009quantum}. They then apply $(X^TX)^{-1}$ using the quantum linear system solving algorithm of \cite{harrow2009quantum}, giving a final state proportional to $(X^TX)^{-1}X^T\vec y = X^+\vec y$. %The SWAP test is then used to estimate the quality of the fit. 

The approach of \cite{wiebe2012quantum} was later improved upon by \cite{LZ15}, who also give a quantum OLS solver in the sparse-access input model. Unlike \cite{wiebe2012quantum}, they apply $X^+$ directly, by using Hamiltonian simulation of $X$ and phase estimation to estimate the singular values of $X$, and then apply a rotation depending on the inverse singular value if it's larger than 0, and using amplitude amplification to de-amplify the singular-value-zero parts of the state. This results in an algorithm with complexity $\bigOt{s\kappa^3\log(M+N)/\eps^2}$. 

Several works have also considered quantum algorithms for least squares problems with a classical output. The first, due to Wang \cite{Wan17}, outputs the vector $\vec\beta$ in a classical form. 
The input model should be compared with the sparse-access model --- although $\vec y$ is given in classical random access memory, an assumption about the regularity of $\vec y$ means the quantum state $\ket{y}$ can be efficiently prepared. The algorithm also requires a regularity condition on the matrix $X$. The algorithm's complexity is $\mathrm{poly}(\log M,N,\kappa,\frac{1}{\eps})$. Like \cite{LZ15}, Wang's algorithm uses techniques from quantum linear system solving to apply $X^+$ directly to $\ket{y}$. To do this, Hamiltonian simulation of $X$ is accomplished via what we would call a block-encoding of $X$. This outputs a state proportional to $X^+\vec y$, whose amplitudes can be estimated one-by-one to recover $\vec\beta$. 

A second algorithm to consider least squares with a classical output is \cite{SSP16}, which does not output $\vec\beta$, but rather, given an input $\vec x$, outputs $\vec x^T\vec\beta$, thus predicting a new data point. This algorithm requires that $\vec x$, $\vec y$, and even $X$ be given as quantum states, and assumes that $X$ has low approximate rank. The algorithm uses techniques from quantum principal component analysis \cite{LMR13}, and runs in time $\bigO{\log(N)\kappa^2/\eps^3}$. %\snote{Could we also output a new data point by using the swap test to estimate $x^T X^+y$? We would also need to be able to encode $x$ as a quantum state.}

Recently, Kerenidis and Prakash introduced the quantum data structure input model \cite{kerenidis2016quantum}.  This input model fits data analysis tasks, because unlike in more abstract problems such as Hamiltonian simulation, where the input matrix may be assumed to be sparse and well-structured so that we can hope to have implemented  efficient subroutines to find the non-zero entries of the rows and columns, the input to least squares is generally noisy data for which we may not assume any such structure. 
In Ref.~\cite{kerenidis:quantumgraddescent}, 
%they used a new quantum gradient descent technique, 
utilizing this data structure, they solve the quantum version of the \emph{weighted} least squares problem. Their algorithm assumes access to quantum data structures storing $X$, or some closely related matrix (see Section \ref{sec:data-structure}), $W=\mathrm{diag}(\vec{w})$, and $\vec y$, and have running time $\bigOt{\frac{\kappa^6\mu}{\eps}\polylog(MN)}$, where $\kappa$ is the condition number of $X^T\sqrt{W}$, and $\mu$ is some prior choice of $\nrm{X^T\sqrt{W}}_F$ or $\sqrt{s_{2p}(X^T\sqrt{W})s_{2(1-p)}(X^T\sqrt{W})}$ for some $p\in [0,1]$\footnote{We stress that our algorithms do not achieve the minimum possible $\mu$, but rather, we need to store the input in QROM with a particular $\mu$ in mind. We might more accurately describe the quantum data structure input model as a \emph{family} of input models, parametrized by $\mu$.}.
%\footnote{$s_q(A)^{1/q}$ is defined as the maximum $q$-norm of any row of $A$.\anote{Should we remove this duplicate comment?}}
Note that the choice of $\mu$ impacts the way $X$ must be encoded, leading to a family of algorithms requiring slightly different encodings of the input. %We discuss this more in \snote{somewhere}.

\paragraph{Our results.} We give quantum WLS and GLS solvers in the model where the input is given as a block-encoding. As a special case, we get quantum WLS and GLS solvers in the quantum data structure input model of Kerenidis and Prakash. Our quantum WLS solver has complexity 
$$\bigOt{\mu\kappa\polylog(MN/\eps)}.\footnote{For comparison with the results of \cite{kerenidis:quantumgraddescent}, we assume $\nrm{\sqrt{W}X}\leq 1$, and the normalized residual error of the fit is bounded by a constant. For details see Section \ref{sec:WLS}.}$$
This is a 6-th power improvement in the dependence on $\kappa$, and an exponential improvement in the dependent on $1/\eps$ as compared with the quantum WLS solver of \cite{kerenidis:quantumgraddescent}. Our quantum WLS solver is presented in Section \ref{sec:WLS}, Theorem \ref{thm:WLS}.  Since our algorithm is designed via the block-encoding framework, we also get an algorithm in the sparse-access input model with the same complexity, where $\mu$ is replaced by $s$, the sparsity. As a special case we get a quantum OLS solver, which compares favourably to previous quantum OLS solvers in the sparse-access model \cite{wiebe2012quantum,LZ15} in having a linear dependence on $\kappa$, and a polylog$(1/\eps)$ dependence on the precision. However, these previous results rely on QLS solver subroutines which have since been improved, so their complexity can also likely be improved. 

In addition, we give the first quantum GLS solver. We first show how to implement a GLS solver when the inputs are given as block-encodings (Theorem \ref{thm:GLSBlockGen}). As a special case, we get a quantum GLS solver in the quantum data structure input model (Corollary~\ref{cor:GLSData}), with complexity 
$$\bigOt{\kappa_X\kappa_{\Omega}\left(\mu_X+\mu_\Omega\kappa_{\Omega}\right)\polylog(MN/\eps)},\footnote{Assuming $\nrm{X}$ and $\nrm{\Omega}$ are bounded by $1$, and the normalized residual error of the fit is bounded by a constant. For details see Corollary~\ref{cor:GLSData}.}$$
where $\kappa_\Omega$ and $\kappa_X$ are the condition numbers of $\Omega$ and $X$, and $\mu_\Omega$ and $\mu_X$ are quantities that depend on how the input is given, as in the case of WLS. 
As before, since our algorithm is designed via the block-encoding framework, we also get an algorithm in the sparse-access input model with the same complexity, replacing $\mu_X$ and $\mu_\Omega$ with the respective sparsities of $X$ and $\Omega$. 
For details, see Section~\ref{sec:GLS}.

\subsection{Application to estimating electrical network quantities}

\paragraph{Problem statements.}
An electrical network is a weighted graph $G(V,E,w)$ of $|V|=N$ vertices and $|E|=M$ edges, with the weight of each edge $w_e$ being the inverse of the resistance (conductance) between the two underlying nodes. Given an external current input to the network, represented by a vector spanned by the vertices of the network, we consider the problem of estimating the power dissipated in the network, up to a multiplicative error $\eps$. A special case of this, where the external current just has a unit entering at vertex $s$ and leaving at vertex $t$ is the effective resistance between $s$ and $t$. 

The electrical networks we consider here can be dense, i.e.\ the maximum degree of the network, $d$, can scale with $N$. In addition to considering the sparse-access input model, we consider the {quantum data structure input model}, i.e.\ we assume that certain matrices representing the network are stored in a  quantum-accessible data structure. 
%In particular, for the aforementioned problems, we assume access to the weighted vertex-edge incidence matrix of $G$ in a quantum data structure. %entries of the incidence matrix of $G$ and to the weights of each edge $w_e$ of $G$. %Given these requirements, we estimate the aforementioned quantities to a multiplicative error $\bigO{\eps}$.

\paragraph{Prior work.}

Electrical networks have previously been studied in several quantum algorithmic contexts. Belovs~\cite{belovs2013quantum} established a relationship between the problem of finding a marked node in a graph by a quantum walk and the effective resistance of the graph. Building on this work, several other quantum algorithms have been developed \cite{belovs2013time, montanaro2015quantum, moylett2017quantum}. %Jeffery and Kimmel~\cite{Jeffery2017algorithmsgraph} showed that the problem of determining whether two nodes of a graph are connected is related to finding the effective resistance between them.

Ref.~\cite{IJ15} gave a quantum algorithm for estimating the effective resistance between two nodes in a network when the input is given in the \emph{adjacency query model}. This allows one to query the $ij$-th entry of the adjacency matrix, in contrast to what we are referring to as the sparse-access model, in which one can query the $i$-th nonzero entry of the $j$-th row. They considered the unweighted case, where all conductances are in $\{0,1\}$, and showed how to estimate the effective resistance between two nodes, $R_{s,t}$, to multiplicative accuracy $\eps$ in complexity:
$$\bigOt{\min\left\{\frac{N\sqrt{R_{s,t}}}{\epsilon^{3/2}},\frac{N\sqrt{R_{s,t}}}{\epsilon\sqrt{d\lambda}}\right\}},$$
where $\lambda$ is the spectral gap of the normalized Laplacian. 

Wang~\cite{wang2017efficient} presented two quantum algorithms for estimating certain quantities in large sparse electrical networks: one based on the quantum linear system solver for sparse matrices by \cite{childs2015quantum}, the other based on quantum walks. Using both of these algorithms, Wang estimates the power dissipated, the effective resistance between two nodes, the current across an edge and the voltage between two nodes. Wang's algorithms work in the {sparse-access input model}, in which the algorithm accesses the input by querying the $i$-th neighbour of the $j$-th vertex, for $i\in [d]$ and $j\in [N]$, where $d$ is the maximum degree. 

In particular, we focus on the problems of approximating the power dissipated in a network and the effective resistance between two nodes. If the maximum edge weight of the network is $w_{\mathrm{max}}$ and $\lambda$ is the spectral gap of the normalized Laplacian of the network, then Wang's first algorithm for solving these tasks is based on solving a certain quantum linear system and then estimating the norm of the output state by using amplitude estimation. The resulting complexity is 
$$\bigOt{\dfrac{w_{\mathrm{max}}d^2}{\lambda\epsilon}\mathrm{polylog}(N)}.$$ 
On the other hand, the quantum walk based algorithm by Wang solves these problems in complexity 
$$\widetilde{\mathcal{O}}\left(\min\left\{\dfrac{\sqrt{w_{\mathrm{max}}}d^{3/2}}{\lambda\epsilon}\right.\right.,
    \left.\left.\dfrac{w_{\mathrm{max}}\sqrt{d}}{\lambda^{3/2}\epsilon}\right\}\mathrm{polylog}(N)\right).$$

\paragraph{Our results.} Using the block-encoding framework, we give algorithms that improve on Wang's sparse-access input model algorithms for certain parameters, and in addition, we give the first quantum algorithms for estimating the effective resistance between two nodes and the power dissipated by the network in the quantum data structure input model, where the weighted vertex-edge incidence matrix of the electrical network, as well as the input current vector, are given as a quantum data structure. %We consider networks that may be dense, so the maximum degree $d=d(N)$ may scale with $N$. 
%For estimating the power dissipated and effective resistance between two nodes, we consider an electrical network that may be dense. As such the maximum degree of the network, $d_G(N)$, scales with $N$. We assume that the entries of the incidence matrix of the underlying electrical network can be stored in the quantum data structure and that the weight of each edge of the graph is known. 
%Furthermore, given access to the entries of the input current vector $\mathbf{i}_{\mathrm{ext}}$, we can construct the quantum state $\ket{i_\mathrm{ext}}$ by using the quantum data structure.

Our algorithms are based on the quantum-linear-system-solver-based algorithms of Wang. As described in Section \ref{sec:electrical_networks}, we replace the quantum linear systems algorithm used by Wang, which assumes sparse access to the input \cite{childs2015quantum}, with the QLS solver that we develop here, which assumes the input is given as a block-encoding. We also replace standard amplitude estimation with variable time amplitude estimation. As such, we are not only able to improve upon Wang's algorithms in the sparse-access model, but also provide new algorithms for the same problem in the quantum data structure model. 

In Corollary \ref{cor:dissipated-sparse}, we prove that in the sparse-access input model, there is a quantum algorithm for estimating the dissipated power (or as a special case, the effective resistance) to an $\eps$-multiplicative error in complexity
$$\bigOt{\dfrac{d^{3/2}}{\epsilon}\sqrt{\dfrac{ w_{\max}}{\lambda}}\mathrm{polylog}(N)}.$$

Thus our algorithm always outperforms the linear-systems based algorithm by Wang. As compared to the quantum-walk based-algorithm:

\begin{itemize}
\item[(i)] When $d<\sqrt{w_{\mathrm{max}}/\lambda}$, we have a speedup of $\bigOt{1/\sqrt{\lambda}}$.

\item[(ii)] When $d>\sqrt{w_{\mathrm{max}}/\lambda}$, we have a speedup as long as $\sqrt{w_{\max}/\lambda} < d < \sqrt{w_{\max}}/\lambda$. 
\end{itemize}

In comparison to the algorithm of Ref.~\cite{IJ15}, our algorithm (in the sparse-access input model) has a speedup as long as $R_{s,t}\gg d^4/N^2$, although we note that these results are not directly comparable, as they assume very different input models.

In Corollary \ref{cor:dissipated-qram1}, we give the first quantum algorithm for estimating the dissipated power (or as a special case, the effective resistance) in the quantum data structure model, with complexity:
$$
\bigOt{\dfrac{1}{\eps}\sqrt{\dfrac{dw_{\mathrm{max}}N}{\lambda}}}.
$$ 
This algorithm outperforms the quantum-linear-system-based algorithms by Wang for both these tasks when the maximum degree of the electrical network is $\Omega(N^{1/3})$. On the other hand, as compared to the quantum walk based algorithm:

\begin{itemize}
\item[(i)] When $d<\sqrt{w_{\mathrm{max}}/\lambda}$, we have a speedup as long as $\lambda < d^2/N$.

\item[(ii)] When $d>\sqrt{w_{\mathrm{max}}/\lambda}$, we have a speedup as long as $\lambda<\sqrt{w_{\mathrm{max}}/N}.$ 
\end{itemize}

In comparison to the algorithm for estimating effective resistance in the adjacency query model from Ref.~\cite{IJ15}, we get an improvement whenever $\lambda = \Omega(1)$ and $R_{s,t}\gg d^2/N$. 

We emphasize that our algorithm in Corollary \ref{cor:dissipated-qram1} is not directly comparable to any of these previous results, since the input models are different. 

\section{Preliminaries}

\subsection{Notation}

We begin by introducing some notation. 
For $A\in\mathbb{C}^{M\times N}$, define $\overline{A}\in\mathbb{C}^{(M+N)\times (M+N)}$ by
\begin{equation}
\label{eq:symmetrized-A}
\overline{A}=\left[\begin{array}{cc}0 & A\\ A^\dagger & 0 \end{array}\right].
\end{equation}
For many applications where we want to simulate $A$, or a function of $A$, it suffices to simulate~$\overline{A}$.

\vskip10pt
\noindent For $A\in\mathbb{C}^{N\times N}$, we define the following norms:
\begin{itemize}
\item Spectral norm: $\nrm{A} = \max\{\nrm{A\ket{u}}:\nrm{\ket{u}}=1\}$
\item Frobenius norm: $\nrm{A}_F =\sqrt{\sum_{i,j}A_{i,j}^2}$
\end{itemize}

\vskip10pt
\noindent For $A\in \mathbb{C}^{M\times N}$, let $A_{i,\cdot}$ denote the $i$-th row of $A$, $\mathrm{row}(A)$ the span of the rows of $A$, and $\mathrm{col}(A)=\mathrm{row}(A^T)$. Define the following:
\begin{itemize}
\item For $q\in [0,1]$, $s_q(A)=\max_{i\in M}\nrm{A_{i,\cdot}}_q^q$
\item For $p\in [0,1]$, $\mu_p(A)=\sqrt{s_{2p}(A)s_{2(1-p)}(A^T)}$
\item $\sigma_{\min}(A)=\min\{\nrm{A\ket{u}}:\ket{u}\in\mathrm{row}(A),\nrm{\ket{u}}=1\}$ (the smallest non-zero singular value)
\item $\sigma_{\max}(A)=\max\{\nrm{A\ket{u}}:\nrm{\ket{u}}=1\}$ (the larges singular value)
\item $\nrm{A}=\nrm{\overline{A}}=\sigma_{\max}(A)$
\end{itemize}

For $A\in \mathbb{C}^{M\times N}$ with singular value decomposition $A=\sum_i\sigma_i\ket{u_i}\bra{v_i}$, we define the Moore-Penrose pseudoinverse of $A$ by $A^+=\sum_i\sigma_i^{-1}\ket{v_i}\bra{u_i}$. 
For a matrix $A$, we let $A^{(p)}$ be defined $A_{i,j}^{(p)}=(A_{i,j})^p$.

\subsection{Quantum-accessible data structure}\label{sec:data-structure}

We will consider the following data structure, studied in \cite{kerenidis2016quantum}. We will refer to this data structure as a \emph{quantum-accessible} data structure, because it is a classical data structure, which, if stored in QROM, is addressable in superposition, but needn't be able to store a quantum state, facilitates the implementation of certain useful quantum operations. In our complexity analysis, we consider the cost of accessing a QROM of size $N$ to be $\polylog(N)$. Although this operation requires order $N$ gates \cite{GLM07,ArunachalamOnRobustnessBucketBrig15}, but the gates can be arranged in parallel such that the depth of the circuit indeed remains $\polylog(N)$.

The following is proven in \cite{kerenidis2016quantum}. We include the proof for completeness. 
\begin{theorem}[Implementing quantum operators using an efficient data structure \cite{kerenidis2016quantum}]
\label{theorem:data_structure}
Let\\ $A\in\mathbb{R}^{M\times N}$ be a matrix with $A_{ij}\in \mathbb{R}$ being the entry of the $i$-th row and the $j$-th column. If $w$ is the number of non-zero entries of $A$, then there exists a data structure of size\footnote{\label{foot:bitCount}Here, for simplicity we assume that we can store a real number in $1$ data register, however more realistically we should actually count the number of bits, incurring logarithmic overheads. Also in this theorem we assign unit cost for classical arithmetic operations.} $\bigO{w\log^2(MN)}$ that, given the entries $(i,j,A_{ij})$ in an arbitrary order, stores them such that time\textsuperscript{\emph{\ref{foot:bitCount}}} taken to store each entry of $A$ is $\mathcal{O}\left(\log(MN)\right)$. Once this data structure has been initiated with all non-zero entries of $A$, there exists a quantum algorithm that can perform the following maps with $\eps$-precision in $\mathcal{O}\left(\text{polylog}(MN/\eps)\right)$ time:
\[\widetilde{U}:\ket{i}\ket{0}\mapsto \ket{i}\dfrac{1}{\nrm{A_{i,\cdot}}}\sum_{j=1}^{N}A_{i,j}\ket{j}=\ket{i,A_i},\]
\[\widetilde{V}:\ket{0}\ket{j}\mapsto \dfrac{1}{\nrm{A}_F}\sum_{i=1}^{M}\nrm{A_{i,\cdot}}\ket{i}\ket{j}=\ket{\widetilde{A},j},\]
where 
$\ket{A_{i,\cdot}}$ is the normalized quantum state corresponding to the $i$-th row of $A$ and $\ket{\widetilde{A}}$ is a normalized quantum state such that $\braket{i}{\widetilde{A}}=\nrm{A_{i,\cdot}}$, i.e.\ the norm of the $i$-th row of $A$.   

In particular, given a vector $\vec{v}\in \mathbb{R}^{M\times 1}$ stored in this data structure, we can generate an $\eps$-approximation of the superposition $\sum_{i=1}^Mv_i\ket{i}/\nrm{\vec{v}}$ in complexity $\polylog(M/\eps)$. 
\end{theorem}
\begin{proof}
The idea is to have a classical data structure to which the quantum algorithm has access. The data structure includes an array of $M$ full binary trees, each having $N$ leaves. For the incoming entry $(A_{ij},i,j)$, the tuple $\left(A_{i,j}^2,\mathrm{sign}(A_{ij})\right)$ is stored in leaf $j$ of binary tree $B_i$. An internal node $l$ stores the sum of the entries of the leaves in the subtree rooted at $l$. In this way, the root of binary tree $B_i$ contains the entry $\sum_{j=1}^N A_{i,j}^2$. Let the value of any internal node $l$ of $B_i$, at depth $d$ be denoted by $B_{i,l}$. Then if $j_b$ represents the $b$-th bit of $j$, then 
$$B_{i,l}=\sum_{\substack{j_1...j_d=l;\\
j_{d+1}...j_{\log(N)}\in\{0,1\}}} A_{i,j}^2.$$
This implies that the first $d$ bits of $j$ written in binary is fixed to $l$, indicating that we are at depth $d$. So whenever a new entry comes in, all nodes of the binary tree
corresponding to the entry gets updated. In the end the root stores $\nrm{A_{i,\cdot}}^2$.
As there are at most $\mathcal{O}(\log N)$ nodes from the leaf to the root of any binary tree and to find the address of each entry takes $\mathcal{O}(\log(MN))$, inserting each entry into this data structure takes $\mathcal{O}(\log^2\left(MN\right))$ time. If there are $w$ non-zero entries in $A$, then the memory requirement of this data structure is $\mathcal{O}(w\log^2(MN))$, because each entry can cause $\ceil{\log(N)}$ new nodes to be added, each of which require $\mathcal{O}(\log(MN))$ registers.

To construct the unitary $\widetilde{U}$ in $\mathcal{O}(\text{polylog}(MN/\eps))$ time, quantum access to this data structure is required. A sequence of controlled-rotations is performed, similarly to the ideas of~\cite{GroverRudolf02}. For any internal node $B_{i,l}$ at depth $d$, conditioned on the first register being $\ket{i}$ and the first $d$ qubits of the second register being equal to $l$, the following rotation is made to the $(d+1)$-th qubit
\begin{equation*}
%\label{eq:controlled_rotation_data_structure}
\ket{i}\ket{l}\ket{0....0}\mapsto\ket{i}\ket{l}\dfrac{1}{\sqrt{B_{i,l}}}\left(\sqrt{B_{i,2l}}\ket{0}+\sqrt{B_{i,2l+1}}\ket{1}\right)\ket{0....0}.
\end{equation*}
For the last qubit, i.e.\ the $\ceil{\log(n)}$-th qubit, the sign of the entry is also included
\begin{equation*}
%\label{eq:controlled_rotation_data_structure_leaf}
\ket{i}\ket{l}\ket{0}\mapsto\ket{i}\ket{l}\dfrac{1}{\sqrt{B_{i,l}}}\left(\mathrm{sign}(a_{2l})\sqrt{B_{i,2l}}\ket{0}+\mathrm{sign}(a_{2l+1})\sqrt{B_{i,2l+1}}\ket{1}\right).
\end{equation*}
So performing $\widetilde{U}$ requires $\ceil{\log(N)}$ controlled rotations and for each of which two queries to the classical database is made to query the children of the node under consideration. Discretization errors can be nicely bounded and one can see that an $\eps$-approximation of $\widetilde{U}$ can be implemented in $\mathcal{O}\left(\text{polylog}(MN/\eps)\right)$ time.

To implement $\widetilde{V}$, we require an additional binary tree $\mathcal{B}$ having $M$ leaf nodes. Leaf $j$ stores the entry of the root of binary tree $B_j$. As before, all internal nodes $l$ store the sum of the entries of the subtree rooted at $l$. So just as before, by applying $\ceil{\log(M)}$ controlled rotations, each of which queries the database twice, we can implement an $\eps$-approximation of $\widetilde{V}$ in $\mathcal{O}(\text{polylog}(MN/\eps))$ time.
\\~\\
\textit{Preparation of quantum states:} Note that this data structure is also useful for preparing a quantum state when the entries of a classical vector arrive in an online manner. Formally speaking, if $\vec{v}\in\mathbb{R}^M$ is a vector with $i$-th entry $v_i$, then using the quantum-accessible data structure described above, one can prepare the quantum state
$\ket{\vec v}=\frac{1}{\nrm{\vec v}}\sum_i v_i\ket{i}$. 
The idea is similar to the case of a matrix. One can store the tuple $(v_i^2,\mathrm{sign}(v_i))$ in the $i$-th leaf of a binary tree. As before, any internal node $l$ stores the sum of squares of the entries of the subtree rooted at $l$. So, we can use the same unitary $\widetilde{U}$ as before to obtain $\ket{\vec v}$.
\end{proof}

As a corollary, we have the following, which allows us to generate alternative quantum state representations of the rows of $A$, as long as we have stored $A$ appropriately beforehand:\vskip-3mm
\begin{corollary}
If $A^{(p)}$ is stored in a quantum data structure, then there exists a quantum algorithm that can perform the following map with $\eps$-precision in $\polylog(MN/\eps)$ time:
$$\ket{i}\ket{0}\mapsto \ket{i}\frac{1}{s_{2p}(A)}\sum_{j=1}^NA_{i,j}^p\ket{j}.$$
\end{corollary}
In Section~\ref{sec:tools} we show that using these techniques one can efficiently implement a block-encoding of the matrix $A$, as defined below. 
    
\subsection{Unitary block-encoding of matrices}
    We will take advantage of recent techniques in Hamiltonian simulation~\cite{LowChuangQubitization2016,LowChuangHamSpectraAmp2017}, which enable us to present our results in a nice unified framework. The presented techniques give rise to exponential improvements in the dependence on precision in several applications. In this framework we will represents a subnormalised matrix as the top-left block of a unitary.
    $$ U=\left(\begin{array}{cc} A/\alpha & . \\ . & .\end{array}\right) $$
	Following the exposition of Gilyén et al.~\cite{gilyenBlockMatrices} we use the following definition: 

	\begin{definition}[Block-encoding]\label{def:standardForm}
		Suppose that $A$ is an $s$-qubit operator, $\alpha,\eps\in\R_+$ and $a\in \mathbb{N}$. Then we say that the $(s+a)$-qubit unitary $U$ is an $(\alpha,a,\eps)$-block-encoding\footnote{Note that since $\nrm{U}=1$ we necessarily have $\nrm{A}\leq \alpha+\eps$.} of $A$, if 
		$$ \nrm{A - \alpha(\bra{0}^{\otimes a}\otimes I)U(\ket{0}^{\otimes a}\otimes I)}\leq \eps. $$
	\end{definition}

 	Block-encodings are really intuitive to work with. For example, one can easily take the product of two block-encoded matrices by keeping their ancilla qubits separately. The following lemma shows that the errors during such a multiplication simply add up as one would expect, and the block-encoding does not introduce any additional errors.
	
	\begin{lemma}[Product of block-encoded matrices]\label{lemma:disjointAncillaProduct}
		If $U$ is an $(\alpha,a,\delta)$-block-encoding of an $s$-qubit operator $A$, and $V$ is a $(\beta,b,\eps)$-block-encoding of an $s$-qubit operator $B$ then\footnote{In the expression $(I_b\otimes U)(I_a\otimes V)$, the identity operator $I_b$ should be seen as acting on the ancilla qubits of $V$, and $I_a$ on those of $U$.}
%s act on each other's ancilla qubits, which is hard to express properly using simple tensor notation, but the reader should read this tensor product this way.} 
$(I_b\otimes U)(I_a\otimes V)$ is an $(\alpha\beta,a+b,\alpha\eps+\beta\delta)$-block-encoding of $AB$.
	\end{lemma}
	\begin{proof}
		\begin{align*}
		&\nrm{AB - \alpha\beta(\bra{0}^{\otimes a+b}\otimes I)(I_b\otimes U)(I_a\otimes V)(\ket{0}^{\otimes a+b}\otimes I)}\\
		=&\Big\lVert AB - 
		\underset{\tilde{A}}{\underbrace{\alpha(\bra{0}^{\otimes a}\otimes I)U(\ket{0}^{\otimes a}\otimes I)}}
		\underset{\tilde{B}}{\underbrace{\beta(\bra{0}^{\otimes b}\otimes I)V(\ket{0}^{\otimes b}\otimes I)}}\Big\rVert\\
		=&\nrm{AB -\tilde{A}B+\tilde{A}B-\tilde{A}\tilde{B}}\\
		=&\nrm{(A-\tilde{A})B+\tilde{A}(B-\tilde{B})}\\
		=&\nrm{A-\tilde{A}}\beta+\alpha\nrm{B-\tilde{B}}\\		
		\leq& \alpha\eps+\beta\delta.\qedhere
		\end{align*}
	\end{proof}

\begin{comment}
The following is proven similarly, but we defer the proof to Appendix \ref{app:proofs}.
\begin{restatable}{lemma}{nonSquareProd}\label{lem:block-encoding-product-non-square} %\snote{lem:block-encoding-product-non-square}
Let $A\in \mathbb{R}^{S\times S}$ and $B\in\mathbb{R}^{S\times T}$. If $U$ is a $(\alpha,a,\delta)$-block-encoding of $A$ that can be implemented in time $T_U$, and $V$ is a $(\beta,b,\eps)$-block-encoding of $\overline{B}$ that can be implemented in time $T_V$, then there is a $(\alpha\beta,a+b,2\alpha\eps+\beta\delta)$-block-encoding of $\overline{AB}$ that can be implemented in time $\bigO{T_U+T_V}$.  
\end{restatable}
\end{comment}

The above lemma assumes that both matrices are of size $2^s\times 2^s$. This is in fact without loss of generality, if the two matrices have size say $M\times N$ and $N \times K$ where $M,N,K\leq 2^s$ we can simply ``pad'' the matrices with zero entries, which does not affect the result of the matrix product.

Also the above lemma can be made more efficient in some cases when both $A$ and $B$ are significantly subnormalized. In such a situation we can first amplify the block-encodings and then only after that take their product. This improvement is based on the fact a subnormalized block-encoding can be efficiently amplified, as shown by Low and Chuang~\cite{LowChuangHamSpectraAmp2017} and Gilyén et al.~\cite{gilyenBlockMatrices}. The precise argument can be found in Appendix \ref{app:proofs}.

\begin{restatable}{lemma}{nonSquareProdPreAmp}\emph{(Poduct of preamplified block-matrices~\cite{LowChuangQubitization2016})}\label{lem:block-encoding-product-non-square-pre-amp} %\snote{lem:block-encoding-product-non-square}
	Let $A\in \mathbb{R}^{M\times N}$ and $B\in\mathbb{R}^{N\times K}$ such that $\nrm{A}\leq 1, \nrm{B}\leq 1$. If $\alpha\geq 1$ and $U$ is a $(\alpha,a,\delta)$-block-encoding of $A$ that can be implemented in time $T_U$; $\beta\geq 1$ and $V$ is a $(\beta,b,\eps)$-block-encoding of $B$ that can be implemented in time $T_V$, then there is a $(2,a+b+2,\sqrt{2}(\delta+\eps+\gamma))$-block-encoding of $AB$ that can be implemented in time $\bigO{\left(\alpha(T_U+a)+\beta(T_V+b)\right)\log(1/\gamma)}$.  
\end{restatable}

Also note that if we have a block-encoding for $A$, it can be easily converted to a block-encoding of $\overline{A}$, as shown by the following lemma.
\begin{lemma}[Complementing block-encoded matrices]
	Let $U$ is be an $(\alpha,a,\delta)$-block-encoding of an $s$-qubit operator $A$, and let $\mathrm{c}U$ denote the $(a+1+s)$-qubit controlled-$U$ operator, that acts on the first $a$ and last $s$ qubits controlled on the $(a+1)$st qubit. Then $\mathrm{c}U^\dagger(I_a\otimes X\otimes I_s)\mathrm{c}U$ is an $(\alpha,a+1,\delta)$-block-encoding of $\overline{A}$.
\end{lemma}
	
	The following theorem about block-Hamiltonian simulation is a corollary of the results of \cite[Theorem 1]{LowChuangQubitization2016}, which also includes bounds on the propagation of errors. For more details see  Appendix~\ref{app:error}.
	\begin{restatable}{theorem}{optHamSim}\emph{(Block-Hamiltonian simulation \cite{LowChuangQubitization2016})}\label{thm:blockHamSim}
		Suppose that $U$ is an $(\alpha,a,\eps/|2t|)$-block-encoding of the Hamiltonian $H$. Then we can implement an $\eps$-precise Hamiltonian simulation unitary $V$ which is an $(1,a+2,\eps)$-block-encoding of $e^{itH}$, with $\bigO{|\alpha t|+\log(1/\eps)}$ uses of controlled-$U$ or its inverse and with $\bigO{a|\alpha t|+a\log(1/\eps)}$ two-qubit gates.
	\end{restatable}

\noindent From this, we can prove the following useful statement (proven in Appendix~\ref{app:smooth} as Lemma~\ref{lemma:controlledHamsin}).

\begin{lemma}[Implementing controlled Hamiltonian simulation operators]\label{lem:controlled}
Let $T=2^J$ for some $J\in \mathbb{N}$ and $\epsilon\geq0$. Suppose that $U$ is an $(\alpha,a,\eps/|8 (J+1)^2 T|)$-block-encoding of the Hamiltonian $H$. Then we can implement a $(1,a+2,\eps)$-block-encoding of $\sum_{t=1}^{T-1}\ket{t}\bra{t}\otimes e^{it H}$, with $\bigO{\alpha  T+J\log(J/\eps)}$ uses of controlled-$U$ or its inverse and with $\bigO{a(\alpha T+J\log(J/\eps))}$ two-qubit gates.	
\end{lemma}

Apeldoorn et al.~developed some general techniques~\cite[Appendix B]{AGGW:SDP} that make it possible to implement smooth-functions of a Hamiltonian $H$, accessing $H$ only via controlled-Hamiltonian simulation. Using their techniques, we show in Appendix~\ref{app:smooth} the following results about implementing negative and positive powers of Hermitian matrices.

\begin{restatable}{lemma}{negPower}\emph{(Implementing negative powers of Hermitian matrices)}\label{lem:negative-power-restated} %\shnote{lem:negative-power-restated}
	Let $c\in (0,\infty),\kappa\geq 2$, and let $H$ be a Hermitian matrix such that $I/\kappa \preceq H \preceq I$. Suppose that $\delta = \littleo{\eps/\left(\kappa^{1+c}(1+c)\log^3\frac{\kappa^{1+c}}{\eps}\right)}$, and $U$ is an $(\alpha,a,\delta)$-block-encoding of $H$, that can be implemented using $T_U$ elementary gates. 
	Then for any $\eps$, we can implement a unitary $\widetilde{U}$ that is a $(2\kappa^c,a+\bigO{\log(\kappa^{1+c}\log\frac{1}{\eps}},\eps)$-block-encoding of $H^{-c}$ in cost
	$$\bigO{\alpha \kappa(a+T_U)(1+c)\log^2\!\left(\frac{\kappa^{1+c}}{\eps}\right)}.$$
	%$$\bigO{\alpha \mathrm{max}(1,c) \kappa\log\left(\frac{\kappa^c}{\eps}\right)(a+T_U)+\kappa \mathrm{max}(1,c) \log^2\left(\frac{\mathrm{max}(1,c)\kappa^{\mathrm{max}(1,c)}}{\eps}\right)}.$$
\end{restatable}

\begin{restatable}{lemma}{posPower}\emph{(Implementing positive powers of Hermitian matrices)}\label{lem:positive-power-restated} %\shnote{lem:positive-power-restated}
	Let $c\in (0,1],\, \kappa\geq 2$, and $H$ a Hermitian matrix such that $I/\kappa \preceq H \preceq I$. Suppose that for $\delta = \littleo{\eps/(\kappa\log^3\frac{\kappa}{\eps})}$, and we are given a unitary $U$ that is an $(\alpha,a,\delta)$-block-encoding of $H$, that can be implemented using $T_U$ elementary gates. 
Then for any $\eps$, we can implement a unitary $\widetilde{U}$ that is a $(2,a+\bigO{\log\log(1/\eps)}, \eps)$-block-encoding of $H^{c}$ in cost
$$\bigO{\alpha\kappa(a+T_U)\log^2(\kappa/\eps)}.$$
%$$\bigO{\alpha(\kappa+\log\log(1/\eps))\log(1/\eps)(a+T_U)+\kappa\log(1/\eps)\log(\kappa/\eps)}.$$
	\end{restatable}
	
	Finally, we note that subsequent work of Gilyén et al.~\cite{gilyenBlockMatrices} improved the log factor of the above two lemmas quadratically, and reduced the ancilla space overhead to a constant. This also directly implies an improvement in the log factors of the results presented in Section~\ref{sec:tools}.
	
\subsection{Sparse-access input model}
In the sparse-access model we assume that the input matrix $A\in\C^{M\times N}$ has $s_r$-sparse rows and $s_c$-sparse columns, such that the matrix elements can be queried via an oracle 
\begin{align*}
&\mathrm{O}_A\colon \ket{i}\ket{j}\ket{0}^{\!\otimes b} \mapsto \ket{i}\ket{j}\ket{a_{ij}}& &\kern-30mm\forall i\in[M],j\in[N].
\end{align*}
Moreover, the indices of non-zero elements of each row can be queried via an oracle 
\begin{align*}
&\mathrm{O}_r\colon \ket{i}\ket{k} \mapsto \ket{i}\ket{r_{ik}}& &\kern-30mm\forall i\in[N], k\in [s_r],\text{ where}		
\end{align*}
$r_{ij}$ is the index for the $j$-th non-zero entry of the $i$-th row of $A$, or if there are less than $i$ non-zero entries, then it is $j+N$. If $A$ is not symmetric (or Hermitian) then we also assume the analogous oracle for columns.
It is not difficult to prove~\cite{Chi10} that a block-encoding of $A$ can be efficiently implemented in the sparse-access input model, see \cite[Lemma 48]{gilyenBlockMatrices} 
for a direct proof.

\begin{lemma}[Constructing block-encodings for sparse-access matrices~{\cite[Lemma 48]{gilyenBlockMatrices}}]\label{lem:sparse-block}
	Let $A\in\mathbb{C}^{M\times N}$ be an $s^r,s^c$ row and column-sparse matrix given in the sparse-access input model. Then for any $\eps\in(0,1)$, we can implement a $(\sqrt{s^r s^c},\polylog(MN/\eps),\eps)$-block-encoding of $A$ with $\bigO{1}$ queries and $\polylog(MN/\eps)$ elementary gates.
	\end{lemma}

\section{Variable-time amplitude amplification and estimation}\label{sec:vtae}

	Following the work of Ambainis~\cite{AmbainisVariableTime12} we define variable-stopping-time quantum algorithms. In our presentation we use the formulation of Childs et al.~\cite{childs2015quantum} which makes the statements easier to read, while one does not lose much of the generality. 
	
	In the problem of variable-time amplitude amplification the goal is to amplify the success probability of a variable-stopping-time algorithm by exploiting that the computation may end after time $t_j$ marking a significant portion of the quantum state as ``bad''. Here we define the problem of variable-time amplitude estimation which asks for an $\eps$-multiplicative estimate of the initial unamplified amplitude/probability of success.
	
	Our approach to variable-time amplitude estimation is that we first solve the mindful-amplification problem,
	where we amplify the amplitude to $\Theta(1)$, while also determining the amplification gain up to $\eps/3$-multiplicative precision. Then we estimate to $\eps/3$-multiplicative precision the amplitude after the mindful-amplification using amplitude estimation, incurring an overhead of $\approx 1/\eps$. This then results in an $\eps$-multiplicative approximation of the initial amplitude.
	
	\begin{definition}[Mindful-amplification problem]
		For a given $\eps>0$, a quantum algorithm $\A$ and an orthogonal projector $\Pi$, the \emph{$\eps$-mindful-amplification problem} is the following: Construct an algorithm $\A'$ such that $\Pi \A' \ket{\pmb{0}}\propto \Pi \A \ket{\pmb{0}}$ and $\nrm{\Pi \A' \ket{\pmb{0}}}=\Theta(1)$, moreover output a number $\Gamma$ such that $\frac{\nrm{\Pi \A' \ket{\pmb{0}}}}{\Gamma\nrm{\Pi \A \ket{\pmb{0}}}}\in \left[1-\eps, 1+\eps\right]$.
	\end{definition}
	
\subsection{Variable-stopping-time algorithms and variable-time amplification}\label{sec:variable-stopping}
	
	Now we turn to discussing variable-stopping-time quantum algorithms. The main idea of such an algorithm is that there are $m$ possible stopping times, and for each stopping time $t_j$, there is a control register that can be set to $1$ at time $t_j$, indicating that the computation has stopped on that branch. More precisely it means that after time $t_j$, the algorithm does not alter the part of the quantum state for which the control flag has been set to $1$ by time $t_j$. 
	
	\begin{definition}[Variable-stopping-time quantum algorithm]
		We say that $\A=\A_m\cdot \ldots \cdot \A_1$ is a \emph{variable-stopping-time} quantum algorithm if $\A$ acts on $\Hi=\Hi_C\otimes \Hi_\A$, where $\Hi_C=\otimes_{i=1}^m \Hi_{C_i}$ with $\Hi_{C_i}=\mathrm{Span}(\ket{0},\ket{1})$, and each unitary $\A_j$ acts on $\Hi_{C_j}\otimes \Hi_\A$ controlled on the first $j-1$ qubits $\ket{0}^{\otimes j-1}\in \otimes_{i=1}^{j-1} \Hi_{C_i}$  being in the all-$0$ state.
	\end{definition}
	In the case of variable-time amplitude amplification the space $\Hi_\A$ on which the algorithm acts has a flag which indicates success, i.e., $\Hi_\A=\Hi_F\otimes \Hi_W$, where the flag space $\Hi_F=\mathrm{Span}(\ket{g},\ket{b})$ indicates ``good'' and ``bad'' outcomes. Also we define stopping times $0=t_0<t_1<t_2<\ldots<t_m=T_{\max}$ such that for all $j\in[m]$ the algorithm $\A_j\cdot \ldots \cdot \A_1$ has (query/gate) complexity $t_j$. In order to analyse such an algorithm we define the probability of the different stopping times.
	We use $\ket{\pmb{0}}\in\Hi$ to denote the all-$0$ initial state on which we run the algorithm~$\A$.
	
	\begin{definition}[Probability of stopping by time $t$]\label{def:stoppingProbabilites}
	We define the orthogonal projector 
	$$\Pi_{\mathrm{stop}\leq t}:=\sum_{j\colon t_j\leq t}\ketbra{1}{1}_{C_j}\otimes I_{\Hi_{\A}},$$
	where by $\ketbra{1}{1}_{C_j}$ we denote the orthogonal projector on $\Hi_C$ which projects onto the state $$\ket{0}_{\Hi_{C_1}}\otimes\cdots\otimes\ket{0}_{\Hi_{C_{j-1}}}
	\otimes\ket{1}_{\Hi_{C_{j}}}
	\otimes\ket{0}_{\Hi_{C_{j+1}}}\otimes\cdots\otimes\ket{0}_{\Hi_{C_m}}.$$ 
	We define $p_{\mathrm{stop}\leq t}:=\nrm{\Pi_{\mathrm{stop}\leq t}\A\ket{\pmb{0}}}^2$, and similarly $p_{\mathrm{stop}\geq t}$. Finally we define the projector $$\Pi_{\mathrm{mg}}^{(j)}:=I-\Pi_{\mathrm{stop}\leq t_j}\cdot \left(I_{\Hi_{C}}\otimes\ketbra{b}{b}_{\Hi_{F}}\otimes I_{\Hi_{W}}\right),$$
	and $p_{\mathrm{mg}}^{(j)}:=\nrm{\Pi_{\mathrm{mg}}^{(j)}\A\ket{\pmb{0}}}^2=\nrm{\Pi_{\mathrm{mg}}^{(j)}\A_j\cdot\ldots\cdot \A_1\ket{\pmb{0}}}^2$ expressing the probability that the state ``maybe good'' after the $j$-th segment of the algorithm has been used. This is $1$ minus the probability that the state was found to be ``bad'' by the end of the $j$-th segment of the algorithm.
	\end{definition}

	For simplicity from now on we assume that $p_{\mathrm{stop}\leq t_m}=1$. Using the above notation we can say that in the problem of variable-time amplitude amplification the goal is to prepare a state $\propto \Pi_{\mathrm{mg}}^{(m)}\A\ket{\pmb{0}}$; in variable-time amplitude estimation the goal is to estimate $p_{\mathrm{succ}}:=\nrm{\Pi_{\mathrm{mg}}^{(m)}\A\ket{\pmb{0}}}^2$. 
	
	Now we define what we precisely mean by variable-time amplification.
	\begin{definition}[Variable-time amplification]\label{def:vtAmp}
		We say that $\A'=(\A'_1,\A'_2,\ldots,\A'_m)$ is a \emph{variable-time amplification} of $\A$ if $\A'_0=I$ and $\forall j\in [m]\colon$ $\Pi_{\mathrm{mg}}^{(j)}\A'_j\ket{\pmb{0}}\propto \Pi_{\mathrm{mg}}^{(j)}\A_j\A'_{j-1}\ket{\pmb{0}}$, moreover $\A'_j$ uses the circuit $\A_j\A'_{j-1}$ and its inverse a total of $q_j$ times and on top of that it uses at most $g_j$ elementary gates. We define $a_j:=\frac{\nrm{\Pi_{\mathrm{mg}}^{(j)}\A'_j\ket{\pmb{0}}}}{\nrm{\Pi_{\mathrm{mg}}^{(j)}\A_j\A'_{j-1}\ket{\pmb{0}}}}$ as the amplification of the $j$-th phase, and $o_j:=\frac{q_j}{a_j}$ as the (query) overhead of the $j$-th amplification phase.
	\end{definition}
	Note that the above definition implies that for a variable-time amplification $\A'$ we have that $\forall j\in [m]\colon$ $\Pi_{\mathrm{mg}}^{(j)}\A'_j\ket{\pmb{0}}\propto \Pi_{\mathrm{mg}}^{(j)}\A_j\A_{j-1}\cdots \A_1\ket{\pmb{0}}$, in particular $\Pi_{\mathrm{mg}}^{(j)}\A'_m\ket{\pmb{0}}\propto \Pi_{\mathrm{mg}}^{(j)}\A\ket{\pmb{0}}$. 
	
	The following lemma analyses the efficiency of a variable-time amplification $\A'$.
	\begin{lemma}\label{lemma:basicRecursion}
		For all $j< k\in [m]$ we have that $\A'_k$ uses $\A'_j$ a total of 
		\begin{equation}\label{eq:basicRecursion}
			 \frac{\nrm{\Pi_{\mathrm{mg}}^{(k)}\A'_k\ket{\pmb{0}}}}{\nrm{\Pi_{\mathrm{mg}}^{(j)}\A'_j\ket{\pmb{0}}}}\sqrt{\frac{p_{\mathrm{mg}}^{(j)}}{p_{\mathrm{mg}}^{(k)}}} \cdot \prod_{i=j+1}^{k}o_i
		\end{equation} times.
	\end{lemma}
	\begin{proof}
		We prove the claim by induction on $k-j$. For $j=k$ the statement is trivial. For $j=k-1$ we have that
		$\A'_k$ uses $\A'_{k-1}$ a total of $q_k=a_k\cdot o_k$ times by definition. Now observe that
		\begin{align*}
			a_k&=\frac{\nrm{\Pi_{\mathrm{mg}}^{(k)}\A'_k\ket{\pmb{0}}}}{\nrm{\Pi_{\mathrm{mg}}^{(k)}\A_k\A'_{k-1}\ket{\pmb{0}}}}\\
			&=\frac{\nrm{\Pi_{\mathrm{mg}}^{(k)}\A'_k\ket{\pmb{0}}}}{\nrm{\Pi_{\mathrm{mg}}^{(k-1)}\A'_{k-1}\ket{\pmb{0}}}}
			\frac{\nrm{\Pi_{\mathrm{mg}}^{(k-1)}\A'_{k-1}\ket{\pmb{0}}}}{\nrm{\Pi_{\mathrm{mg}}^{(k)}\A_k\A'_{k-1}\ket{\pmb{0}}}}\\
			&=\frac{\nrm{\Pi_{\mathrm{mg}}^{(k)}\A'_k\ket{\pmb{0}}}}{\nrm{\Pi_{\mathrm{mg}}^{(k-1)}\A'_{k-1}\ket{\pmb{0}}}}
			\frac{\nrm{\Pi_{\mathrm{mg}}^{(k-1)}\A\ket{\pmb{0}}}}{\nrm{\Pi_{\mathrm{mg}}^{(k)}\A\ket{\pmb{0}}}}\\
			&=\frac{\nrm{\Pi_{\mathrm{mg}}^{(k)}\A'_k\ket{\pmb{0}}}}{\nrm{\Pi_{\mathrm{mg}}^{(k-1)}\A'_{k-1}\ket{\pmb{0}}}}
			\sqrt{\frac{p_{\mathrm{mg}}^{(k-1)}}{p_{\mathrm{mg}}^{(k)}}}.
		\end{align*}
		Finally we show the induction step when $j<k-1$. As we observed above $\A'_k$ uses $\A'_{k-1}$ a total of $$\frac{\nrm{\Pi_{\mathrm{mg}}^{(k)}\A'_k\ket{\pmb{0}}}}{\nrm{\Pi_{\mathrm{mg}}^{(k-1)}\A'_{k-1}\ket{\pmb{0}}}}
		\sqrt{\frac{p_{\mathrm{mg}}^{(k-1)}}{p_{\mathrm{mg}}^{(k)}}}\cdot o_k$$ times. Note that $\A'_{k}$ only uses $\A'_j$ via $\A'_{k-1}$. By the induction hypothesis we know that $\A'_{k-1}$ uses $\A'_j$ a total of 
		$$ \frac{\nrm{\Pi_{\mathrm{mg}}^{(k-1)}\A'_k\ket{\pmb{0}}}}{\nrm{\Pi_{\mathrm{mg}}^{(j)}\A'_j\ket{\pmb{0}}}}\sqrt{\frac{p_{\mathrm{mg}}^{(j)}}{p_{\mathrm{mg}}^{(k-1)}}} \cdot\prod_{i=j+1}^{k-1}o_i$$ 
		times, which then implies the statement.
	\end{proof}
	\begin{corollary}\label{cor:segmenQueries}
		$\Pi_{\mathrm{mg}}^{(m)}\A'_m\ket{\pmb{0}}\propto \Pi_{\mathrm{mg}}^{(m)}\A\ket{\pmb{0}}$ and for all $j\in [m]$, $\A'_m$ uses $\A_j$ a total of at most 
		\begin{equation}\label{eq:preciseUpperBound}
			\frac{\nrm{\Pi_{\mathrm{mg}}^{(m)}\A'_m\ket{\pmb{0}}}}{\nrm{\Pi_{\mathrm{mg}}^{(j-1)}\A'_{j-1}\ket{\pmb{0}}}}
			\left(1+\sqrt{\frac{p_{\mathrm{stop}\geq t_j}}{p_{\mathrm{succ}}}} \right)\cdot \prod_{i=j}^{m}o_i
		\end{equation} times.
	\end{corollary}
	\begin{proof}
		By Definition~\ref{def:vtAmp} we have that $\A'_m$ uses $\A_{j}$ and $\A'_{j-1}$ the same number of times. By Lemma~\ref{lemma:basicRecursion} we know the latter is used a total of $\frac{\nrm{\Pi_{\mathrm{mg}}^{(m)}\A'_m\ket{\pmb{0}}}}{\nrm{\Pi_{\mathrm{mg}}^{(j-1)}\A'_{j-1}\ket{\pmb{0}}}}\sqrt{\frac{p_{\mathrm{mg}}^{(j-1)}}{p_{\mathrm{mg}}^{(m)}}} \cdot \prod_{i=j+1}^{m}o_i$ times. By Definition~\ref{def:stoppingProbabilites} we have that $p_{\mathrm{succ}}=p_{\mathrm{mg}}^{(m)}$ and $p_{\mathrm{mg}}^{(j-1)}\leq p_{\mathrm{succ}}+p_{\mathrm{stop}\geq t_j}$ from which the statement follows using the simple observation that $\forall a,b\in\mathbb{R}^+_0\colon \sqrt{a+b}\leq\sqrt{a}+\sqrt{b}$.
	\end{proof}

	Now we define some uniform bound quantities, which make it easier to analyze the performance of variable-time amplification.
	\begin{definition}[Uniformly bounded variable-time amplification]
		If the variable-time amplification algorithm $\A'$ is such that for some $E,G,O\in\R_+$ we have that $\forall j\in [m]\colon$
		\begin{align}\label{eq:uniformBounds}
			\frac{\nrm{\Pi_{\mathrm{mg}}^{(m)}\A'_m\ket{\pmb{0}}}}{\nrm{\Pi_{\mathrm{mg}}^{(j-1)}\A'_{j-1}\ket{\pmb{0}}}}\leq E,& & 
			g_j\leq G(t_j-t_{j-1}),
			& \text{ and} & \prod_{i=1}^{m}o_i\leq O,
		\end{align}
		%\begin{align}\label{eq:uniformBounds}
			%\frac{\nrm{\Pi_{\mathrm{mg}}^{(m)}\A'_m\ket{\pmb{0}}}}{\nrm{\Pi_{\mathrm{mg}}^{(j-1)}\A'_{j-1}\ket{\pmb{0}}}}
			%&\leq E\cdot\min\left[\sqrt{\log(T_{\max})},\sqrt{\log(t_j/T_{\max})}\left(\log\log(t_j/T_{\max})+1\right)\right],\\
			%g_j&\leq G(t_j-t_{j-1}), 
			%\quad\text{ and }\quad  \prod_{i=1}^{m}o_i\leq O,
		%\end{align}	
		then we say that $\A'$ is $(E,G,O)$-bounded.
	\end{definition}
	Using the above definition we can derive some intuitive complexity bounds on variable-time amplifications, essentially recovering\footnote{Actually we improve Ambainis's bound by a factor of $\sqrt{\log(T_{\max})}$.} a bound used by Ambainis~\cite{AmbainisVariableTime12}.
	\begin{corollary}\label{cor:boundWithI}
		If $\A'$ is an $(E,G,O)$-bounded variable-time amplification, 
		then $\A'_m$ has complexity at most $EO\left(T_{\max} + \frac{I}{\sqrt{p_{\mathrm{succ}}}}\right)$ coming from the use of the variable-time algorithm $\A$, and it uses at most $EGO\left(T_{\max} + \frac{I}{\sqrt{p_{\mathrm{succ}}}}\right)$ additional elementary gates, where 
		$I\leq t_1+\nrm{T}_2\sqrt{\ln(T_{\max}/t_1)}$ and $\nrm{T}_2:=\sqrt{\sum_{j=1}^{m}t_j^2\cdot p_{\mathrm{stop}=t_j}}$.
	\end{corollary}
	\begin{proof}
		The complexity of the algorithm segment $\A_j$ is $t_j-t_{j-1}$, and due to Corollary~\ref{cor:segmenQueries} $\A'_m$ uses $\A_j$ at most $EO\left(1+\sqrt{\frac{p_{\mathrm{stop}\geq t_j}}{p_{\mathrm{succ}}}} \right)$ times. So we can bound the complexity coming from the use of $\A_j$ by $(t_j-t_{j-1})EO\left(1+\sqrt{\frac{p_{\mathrm{stop}\geq t_j}}{p_{\mathrm{succ}}}} \right)$, and we can bound the number of additional elementary gates coming from the implementation of $\A'_j$ by $(t_j-t_{j-1})EGO\left(1+\sqrt{\frac{p_{\mathrm{stop}\geq t_j}}{p_{\mathrm{succ}}}} \right)$.
		
		We get the total complexities by summing up these contributions for all $j\in [m]$:
		\begin{align}
		E(G)O\sum_{i=1}^{m}(t_j-t_{j-1})\left(1+\sqrt{\frac{p_{\mathrm{stop}\geq t_j}}{p_{\mathrm{succ}}}}\right)
		&= E(G)O\left(T_{\max} + \frac{1}{\sqrt{p_{\mathrm{succ}}}}\sum_{j=1}^{m}(t_j-t_{j-1})\sqrt{p_{\mathrm{stop}\geq t_j}}\right).\label{eq:egoBound}
		\end{align}
		Before bounding the above expression let us introduce some notation. Let $T$ be a random variable corresponding to the stopping times such that $\mathbb{P}(T=t_j)=p_{\mathrm{stop}=t}$. Let $F$ be the distribution function of $T$, i.e., $F(t):=p_{\mathrm{stop}\leq t}$. Also let $F^{-1}(p):=\inf\{t\in \R \colon F(t)\geq p\}$ be the generalized inverse distribution function. The intuitive meaning of $F^{-1}$ is the following: $F^{-1}$ is a monotone increasing function with the property that picking a number $p\in[0,1]$ uniformly at random, and then outputting the value $F^{-1}(p)$ results in a random variable with the same distribution as $T$.
		%$I=int_{0}^1\frac{F^{-1}(p)}{2\sqrt{1-p}}dp$, 
		
		Now using the above definitions we rewrite the summation of \eqref{eq:egoBound} in the following way:
		\begin{align*}
		\sum_{j=1}^{m}(t_j-t_{j-1})\sqrt{p_{\mathrm{stop}\geq t_j}}
		&=\sum_{j=1}^{m}(t_j-t_{j-1})\sum_{k=j}^{m}\left(\sqrt{p_{\mathrm{stop}> t_{k-1}}}-\sqrt{p_{\mathrm{stop}> t_{k}}}\right)\\
		&=\sum_{k=1}^{m}\left(\sqrt{p_{\mathrm{stop}> t_{k-1}}}-\sqrt{p_{\mathrm{stop}> t_{k}}}\right)\left(\sum_{j=1}^{k}t_{j}-\sum_{j=1}^{k-1}t_{j}\right)\\		
		&=\sum_{k=1}^{m}t_k\left(\sqrt{p_{\mathrm{stop}> t_{k-1}}}-\sqrt{p_{\mathrm{stop}> t_{k}}}\right)\\	
		&=\sum_{k=1}^{m}t_k\left(\sqrt{1-F( t_{k-1})}-\sqrt{1-F( t_{k})}\right) \tag{by definition}\\
		&=\sum_{k=1}^{m}t_k\left[-\sqrt{1-p}\right]^{p=F( t_{k})}_{p=F( t_{k-1})}\\
		&=\sum_{k=1}^{m}\int_{F( t_{k-1})}^{F(t_{k})}\frac{t_k}{2\sqrt{1-p}} dp\\
		&=\sum_{k=1}^{m}\int_{F( t_{k-1})}^{F(t_{k})}\frac{F^{-1}(p)}{2\sqrt{1-p}}dp \tag{by definition}\\		
		&=\int_{0}^{1}\frac{F^{-1}(p)}{2\sqrt{1-p}}dp\\
		&=:I.
		\end{align*}
		Thus we get that the total query and gate complexity are bounded by:
		\begin{align*}
		E(G)O\left(T_{\max} + \frac{I}{\sqrt{p_{\mathrm{succ}}}}\right).
		\end{align*}
		Now we finish the proof by bounding $I$. Let $c=t_1^2/T_{\max}^{2}\in(0,1)$, then we can bound $I$ as 
		follows\footnote{Note that the presented upper bound is quite tight, and one may not be able to remove the $\sqrt{\ln(T_{\max})}$ factor. Take for example $T_{\max}\geq 1$ and $F^{-1}(p):=\min(T_{max},(1-p)^{-1/2})$. Then 
			$$I=\frac{1}{2}\int_{0}^{1-T^{-2}_{max}} (1-p)^{-1}dp+\int_{1-T^{-2}_{max}}^{1} \frac{T_{max}}{2\sqrt{1-p}}dp
			=1+\ln(T_{max}),$$ 
			whereas $\nrm{T}_2=\sqrt{1+2\ln(T_{max})}$.
			\anote{For next version: The lower bound of Ambainis~\cite{AmbainisVariableTimeSearch06} seems to imply that $I=\Omega(\nrm{T}_2)$.}
}:
		\begin{align}
		I&=\int_{0}^{1-c}\frac{F^{-1}(p)}{2\sqrt{1-p}}dp+\int_{1-c}^{1}\frac{F^{-1}(p)}{2\sqrt{1-p}}dp\nonumber\\
		&\leq\int_{0}^{1-c}\frac{F^{-1}(p)}{2\sqrt{1-p}}dp+\int_{1-c}^{1}\frac{T_{\max}}{2\sqrt{1-p}}dp\tag{since $F^{-1}\leq T_{\max}$}\\
		&=\int_{0}^{1-c}\frac{F^{-1}(p)}{2\sqrt{1-p}}dp+\sqrt{c}T_{\max}\nonumber\\
		&\leq \frac{1}{2}\sqrt{\int_{0}^{1-c}\left(F^{-1}(p)\right)^2 dp}
		\sqrt{\int_{0}^{1-c}\frac{1}{1-p} dp}+\sqrt{c}T_{\max}\tag{by Hölder's inequality}\\
		&\leq \frac{1}{2}\sqrt{\int_{0}^{1}\left(F^{-1}(p)\right)^2 dp}
		\sqrt{\int_{0}^{1-c}\frac{1}{1-p} dp}+\sqrt{c}T_{\max}\tag{since $F^{-1}\geq 0$}\\						
		&= \sqrt{c}T_{\max}+ \frac{\nrm{T}_2}{2} \sqrt{\int_{0}^{1-c}\frac{1}{1-p} dp}\tag{by the definition of $F^{-1}$}\\	
		&= \sqrt{c}T_{\max}+ \frac{\nrm{T}_2}{2} \sqrt{\int_{c}^{1}\frac{1}{x} dx}\tag{by a change of variable}\\			
		&= t_1+ \frac{\nrm{T}_2}{\sqrt{2}} \sqrt{\ln(T_{\max}/t_1)}.\tag{by the definition of $c$}	
		%&= \sqrt{c}T_{\max}+ \frac{\nrm{T}_q}{2} \left(\left[\frac{2}{2-r}p^{1-r/2}\right]_{c}^{1}\right)^{1/r}
		\end{align}			
	\end{proof}	

	Note that the above upper bound depends on the distribution of stopping times only via $\nrm{T}_2$. 
	Also  we can reduce the number of distinct stopping times to $\leq1+\log(T_{\max}/t_1)$ while increasing the value of $\nrm{T}_2$ by at most a constant factor. Therefore one may assume that $m\leq1+\log(T_{\max}/t_1)$.

	\subsection{Efficient variable-time amplitude amplification and estimation}
	
	Now that we have carefully analyzed the complexity of uniformly bounded variable-time amplifications, 
	%we give a few alternative constructions providing bounds of different flavor.
	we finally apply the results in order to obtain efficient algorithms.
	
	The basic method is to use ordinary amplitude amplification in each amplification phase, which was the original approach that Ambainis used~\cite{AmbainisVariableTime12}.
	In order to understand the efficiency of this approach we invoke a result of Aaronson and Ambainis~\cite[Lemma~9]{AaronsonSpatialSearch}, which carefully analyses the efficiency of amplitude amplification.
	We present their result in a slightly reformulated way, fitting the framework of the presented work better.
	\begin{lemma}[Efficiency of ordinary amplitude amplification]\label{lemma:ampAmpEff}
		Suppose that $\A$ is a quantum algorithm, $\Pi$ is an orthogonal projector, and $\alpha=\nrm{\Pi\A\ket{\pmb{0}}}$. Let $\A^{(k)}$ denote the quantum algorithm that applies $k$ amplitude amplification steps on the outcome of $\A$. If
		$$ k \leq \frac{\pi}{4\arcsin(\alpha)}-\frac{1}{2},$$
		then
		$$ \nrm{\Pi\A^{(k)}\ket{\pmb{0}}} \geq \sqrt{1-\frac{(2k+1)^2\alpha^2}{3}}(2k+1)\alpha.$$		
	\end{lemma}
	The above result essentially states that if the amplification does not start to wrap around, then the inefficiency of the amplification step is bounded by the final amplitude squared. We make this claim precise in the following corollary.
	\begin{corollary}[A bound on amplification ratio]\label{cor:ampAmpEff}
		Suppose that $\A, \Pi, \alpha$ and $\A^{(k)}$ are as in Lemma~\ref{lemma:ampAmpEff}. If we do not overamplify, i.e., $(2k+1)\arcsin(\alpha)\leq \pi/2$, then 
		$$
				\frac{2k+1}{\left(\frac{\nrm{\Pi\A^{(k)}\ket{\pmb{0}}}}{\alpha}\right)}\leq 1 + \frac{3}{2}\nrm{\Pi\A^{(k)}\ket{\pmb{0}}}^2.
		$$ 
	\end{corollary}
	\begin{proof}
		\begin{align*}
			\frac{(2k+1)\alpha}{\nrm{\Pi\A^{(k)}\ket{\pmb{0}}}}
			&\leq \frac{1}{\sqrt{1-\frac{(2k+1)^2\alpha^2}{3}}}\tag{by Lemma~\ref{lemma:ampAmpEff}}\\
			&\leq \frac{1}{\sqrt{1-\frac{(2k+1)^2\arcsin^2(\alpha)}{3}}}\tag*{$\left(\forall \alpha\in [0,1]\colon \arcsin(\alpha)\geq \alpha\right)$}\\		
			&\leq 1 + \frac{5(2k+1)^2\arcsin^2(\alpha)}{9} \tag*{$\left(\forall x\in [0,\pi^2/12]\colon \frac{1}{\sqrt{1-x}}\leq 1+\frac{5}{3}x\right)$}\\
			%Checked by Mathematica: 
			%Plot[{1/Sqrt[1 - x], 1 + 5 x/3}, {x, 0, Pi^2/12}]
			%Solve[1/(1 - x^2) == (1 + 6 x^2/5)^2, x] // N
			%Pi^2/12 // N
			&\leq 1 + \frac{5\pi^2\sin^2\left((2k+1)\arcsin(\alpha)\right)}{9 \cdot 4} \tag*{$\left(\forall y\in [0,\pi/2]\colon y\leq \frac{\pi}{2}\sin(y)\right)$}\\
			&\leq 1 + \frac{3\sin^2\left((2k+1)\arcsin(\alpha)\right)}{2}\tag*{$\left(\frac{5\pi^2}{36}\leq \frac{3}{2}\right)$}\\		
			&=1+\frac{3}{2}\nrm{\Pi\A^{(k)}\ket{\pmb{0}}}^2.
		\end{align*}
		The last equality comes from the usual geometric analysis of amplitude amplification.
	\end{proof}

	Now we turn to proving our result about the efficiency of variable-time amplitude amplification and estimation. A trick we employ is to carefully select the amount of amplification in each phase so that the inefficiencies remain bounded.

	\begin{lemma}[Analysis of variable-time amplitude amplification]
	\label{lemma:vtaa-analysis}
		Suppose $\A'$ is a variable-time amplification such that for all $j\in[m]$ the $j$-th amplification phase $\A'_j$ uses $k_j\geq 0$ ordinary amplitude amplification steps such that
		\begin{equation}\label{eq:dontOverAmp}
			k_j \leq \frac{\pi}{4\arcsin(\alpha)}-\frac{1}{2},
		\end{equation}
		and for all $j\in[m]$ 
		\begin{equation}\label{eq:incrementalAmp}
		\nrm{\Pi_{\mathrm{mg}}^{(j)}\A'_{j}\ket{\pmb{0}}}=\Theta\left(\max\left[\frac{1}{\sqrt{m}},\frac{1}{\sqrt{m-j+1}\left(1+\ln(m-j+1)\right)}\right]\right).
		\end{equation}
		%Then $E=\bigO{\sqrt{m}}$ and $O=\bigO{1}$.
		%If $m=\bigO{\log(T_{\max})}$ and $t_1=\Omega(1)$, then u
		Using $\A'$ the variable-time amplification problem can be solved with query complexity\footnote{Using $m=\bigO{\log(2T_{\max}/t_1)}$, by a little bit more careful analysis in Corollary~\ref{cor:boundWithI}, one can further improve this bound, in particular one can reduce the $T_{\max}\sqrt{m}$ term to $T_{\max}$.}
		$$\bigO{T_{\max}\sqrt{m}+ \frac{\nrm{T}_2\sqrt{\log(2T_{\max}/t_1)}\sqrt{m}}{\sqrt{p_\mathrm{succ}}}}.$$
	\end{lemma}
	\begin{proof}
		We get that $E= \bigO{\sqrt{m}}$ from \eqref{eq:incrementalAmp} immediately.
		We need to work a bit more for bounding $O$:
		\begin{align*}
			O&=\prod_{j=1}^m o_j=\prod_{j=1}^m \frac{q_j}{a_j}=\prod_{j=1}^m \frac{2k_j+1}{\left(\frac{\nrm{\Pi_{\mathrm{mg}}^{(j)}\A'_{j}\ket{\pmb{0}}}}{\nrm{\Pi_{\mathrm{mg}}^{(j)}\A_j\A'_{j-1}\ket{\pmb{0}}}}\right)}\\
			&\leq\prod_{j=1}^m \left(1+ \frac{3}{2} \nrm{\Pi_{\mathrm{mg}}^{(j)}\A'_{j}\ket{\pmb{0}}}^2 \right)\tag{by Corollary~\ref{cor:ampAmpEff}}\\
			&\leq\exp\left(\sum_{j=1}^m \ln\left(1+\frac{3}{2}\nrm{\Pi_{\mathrm{mg}}^{(j)}\A'_{j}\ket{\pmb{0}}}^2\right)\right) \leq\exp\left(\sum_{j=1}^m \frac{3}{2} \nrm{\Pi_{\mathrm{mg}}^{(j)}\A'_{j}\ket{\pmb{0}}}^2\right) \\			
			&\leq\exp\left(C\sum_{j=1}^m \max\left[\frac{1}{m},\frac{1}{(m-j+1)\left(1+\ln(m-j+1)\right)^2}\right]\right)\tag{for some $C\in\mathbb{R}_+$ by \eqref{eq:incrementalAmp}} \\				
		\end{align*}
		Now we show that $O\leq\exp(3C)=\bigO{1}$ by bounding the expression inside the exponent:
		\begin{align*}
			\sum_{j=1}^m \max\left[\frac{1}{m},\frac{1}{(m-j+1)\left(1+\ln(m-j+1)\right)^2}\right]
			&\leq \sum_{j=1}^m \frac{1}{m}+ \sum_{j=1}^m\frac{1}{(m-j+1)\left(1+\ln(m-j+1)\right)^2}\\
			&=1+\sum_{j=1}^m\frac{1}{j\left(1+\ln(j)\right)^2}		=2+\sum_{j=2}^m\frac{1}{j\left(1+\ln(j)\right)^2}\\
			&\leq 2+\int_{1}^{m} \frac{1}{x (1+\ln(x))^2} dx\\
			&=2+\left[\frac{-1}{1+\ln(x)}\right]_{1}^{m}\leq 3.
		\end{align*}		
		Using Corollary~\ref{cor:boundWithI} we get the final complexity claim of variable-time amplification.
	\end{proof}
	
	Now we describe how to efficiently construct a variable-time amplitude amplification algorithm satisfying the above requirements, and derive an efficient algorithm for variable-time amplitude estimation.

	\begin{theorem}[Efficient variable-time amplitude amplification and estimation]\label{thm:efficient-vtaa-vtae}
		Let $U$ be a state-preparation unitary that prepares the sate $U\ket{0}^{\!\otimes k}=\sqrt{p_{\mathrm{prep}}}\ket0\ket{\psi_0}+\sqrt{1-p_{\text{prep}}}\ket1\ket{\psi_1}$ and has query complexity $T_U$. Suppose that $\A$ is a variable-stopping-time algorithm such that 
		%$t_1\geq1$, and 
		we know lower bounds $p_{\mathrm{prep}}\geq p'_{\mathrm{prep}}$ and $p_{\mathrm{succ}}\geq p'_{\mathrm{succ}}$. Let $T'_{\max}:=2T_{\max}/t_1$ and
		$$Q=\left(T_{\max}+\frac{T_U+k}{\sqrt{p_{\mathrm{prep}}}}\right)\sqrt{\log(T'_{\max})}+ \frac{\left(\nrm{T}_2+\frac{T_U+k}{\sqrt{p_{\mathrm{prep}}}}\right)\log(T'_{\max})}{\sqrt{p_\mathrm{succ}}}.$$
		We can construct with success probability at least $1-\delta$ a variable-stopping time algorithm $\A'$ that prepares a state $a\ket{0}\A'\ket{\psi_0}+\sqrt{1-a^2}\ket{1}\ket{\psi_{\mathrm{garbage}}}$, such that $a=\Theta(1)$ and $\A'$ has complexity $\bigO{Q}$, moreover the quantum procedure constructing the classical description of the circuit of $\A'$ has query complexity
		$$\bigO{Q\log\left(T'_{\max}\right)\log\left(\frac1\delta\log\left(\frac{T'_{\max}}{p'_{\mathrm{prep}}p'_{\mathrm{succ}}}\right)\right)}.$$		
		Also, for any $\eps\in(0,1)$ the $\eps$-mindful-amplification problem can be solved using $\A'$ with complexity
		$$\bigO{\frac{Q}{\eps}\log^2\left(T'_{\max}\right)\log\left(\frac{\log(T'_{\max})}{\delta}\right)}.$$ 	
	\end{theorem}
	\begin{proof}
		We describe how to construct a variable-time amplification algorithm as described in Lemma~\ref{lemma:vtaa-analysis}. 
				
		We will use the following fact throughout the proof. If $\B$ is a quantum algorithm such that $\B\ket{0}^{\!\otimes k}=b\ket0\ket{\phi_0}+\sqrt{1-b^2}\ket1\ket{\phi_1}$, then for arbitrary $j\in\N$ we can boost the success probability of amplitude estimation in a way that it outputs either $b<2^j$ or $b\geq 2^{-j}$ such that the output is correct with probability at least $1-\delta'$. Moreover, if the implementation cost of $\B$ is $T_{\B}$, then the cost of the procedure is $\bigO{2^j(T_{\B}+k)\log(1/\delta)}$.
		
		Using the above version of amplitude estimation we can estimate $p_{\mathrm{prep}}$ with constant multiplicative precision and success probability at least $1-\delta/4$ with complexity $\bigO{\frac{T_{U}+k}{\sqrt{p_{\mathrm{prep}}}}\log\left(\frac{\log(1/p'_{\mathrm{prep}})}{\delta}\right)}$.
		Then we amplify $T_U$ using $\Theta(1/\sqrt{p_{\mathrm{prep}}})$ amplification steps, to get amplitude $\Theta(1)$ on the state $\ket{0}\ket{\psi_0}$, and define a new variable-time algorithm $\widetilde{\A}$ by appending the amplified version of $T_U$ to the beginning of the algorithm. This adds $C:=\Theta\left(\frac{T_{U}+k}{\sqrt{p_{\mathrm{prep}}}}\right)$ to the complexity of the first step of the algorithm.
		
		In order to get the claimed bounds we ``sparsify'' the stopping times yielding $\widetilde{m}=\bigO{\log(T'_{\max})}$, without changing $\nrm{\widetilde{T}}_2$ too much. Let us define $\widetilde{m}:=\lceil\log_2(T'_{\max})\rceil$, and also for all $j\in[\widetilde{m}]$ $t'_j:=\max\left(t_j:t_j\text{ is a stopping time of }\A\text{ which is less than or equal }2^jt_1\right)$. Then we define the stopping times of $\widetilde{\A}$ for all $j\in[\widetilde{m}]$ such that $\tilde{t}_j=C+t'_j$. Then clearly we have that $\widetilde{m}=\bigO{\log(T'_{\max})}$, and
		$\widetilde{T}_{\max}=T_{\max}+\Theta\left(\frac{T_{U}+k}{\sqrt{p_{\mathrm{prep}}}}\right)$, and $\nrm{\widetilde{T}}_2
		%\leq\nrm{T'}_2+\Theta\left(\frac{T_{U}+k}{\sqrt{p_{\mathrm{prep}}}}\right)
		\leq2\left(\nrm{T}_2+\Theta\left(\frac{T_{U}+k}{\sqrt{p_{\mathrm{prep}}}}\right)\right)$.
		
		Following Definition~\ref{def:vtAmp} we construct the variable-time amplification $\widetilde{\A}'$ inductively. For each $j\in[\widetilde{m}]$ after running the algorithm $\widetilde{\A}_j\widetilde{\A}_{j-1}'$ we estimate the maybe-good amplitude with constant multiplicative precision and success probability at least $1-\delta/4/(\log(\widetilde{m})+\log(1/p'_{\mathrm{succ}}))$.
		Then we get the algorithm segment $\widetilde{\A}'_j$ by applying amplitude amplification $k_j$ times on $\widetilde{\A}_j\widetilde{\A}_{j-1}'$, such that requirements \eqref{eq:dontOverAmp}-\eqref{eq:incrementalAmp} are satisfied. Observe that upon success the the cost of the amplitude estimation procedure is at most $\bigO{\log\left(\frac{\widetilde{m}+\log_2\left(1/p'_{\mathrm{succ}}\right)}{\delta}\right)}$ times the cost of running $\widetilde{\A}'_j$. Moreover the overall success probability is at least $1-\delta/2$, since upon success there can be at most $\widetilde{m}+\log_2\left(1/p'_{\mathrm{succ}}\right)$ amplification steps.
		
		%We define two thresholds $l_j,u_j\in (0,1)$ which are our desired lower and upper bounds on the amplitude after the $j$-th phase (these can be deduced, e.g., from \eqref{eq:incrementalAmp}.). First we define $\A_0:=\A'_0:=I$. Let $j\in[m]$ and suppose we already constructed $\A'_{j-1}$, then we run the algorithm segment $\A_{j}$ on $\A'_{j-1}\ket{\pmb{0}}$, and determine the amplitude $\alpha'_{j}:=\nrm{\Pi_{\mathrm{mg}}^{(1)}\A_{j}\A'_{j-1}\ket{\pmb{0}}}$ up to multiplicative precision $1/2$. If our estimate $\tilde{\alpha}'_j > l_j$, then we simply define $\A'_{j}:=\A_{j}\A'_{j-1}$, in which case we call the $j$-th step a \emph{trivial step}. 
		
		Note that the above procedure needs to be completed only once in order to construct a variable-time amplification $\widetilde{\A}'$, that satisfies the requirements of Lemma~\ref{lemma:vtaa-analysis}.
		The complexity bound on $\widetilde{\A}'$ follows from  Lemma~\ref{lemma:vtaa-analysis}.  There is no need to use the full procedure constructing $\widetilde{\A}'$ when we use the variable-time amplification itself. The query and gate complexity of the above procedure matches the query and gate complexity of the resulting variable-time amplification up to a factor $\bigO{\widetilde{m}\log\left(\frac{\widetilde{m}+\log\left(1/p'_{\mathrm{succ}}\right)}{\delta}\right)}$, since the sum of the cost of the algorithms  $\widetilde{\A}_j$ is upper bounded by $\widetilde{m}$ times the cost of the variable-time amplified algorithm $\widetilde{\A}'$.		
		
		Finally observe that in order to get an estimate $\Gamma$ such that $\frac{\nrm{\Pi_{\mathrm{mg}}^{(\widetilde{m})} \widetilde{\A}' \ket{\pmb{0}}}}{\Gamma\nrm{\Pi_{\mathrm{mg}}^{(\widetilde{m})} \widetilde{\A} \ket{\pmb{0}}}}\in \left[1-\eps, 1+\eps\right]$, it suffices to obtain estimates $\gamma_j$
		such that $\frac{\nrm{\Pi_{\mathrm{mg}}^{(j)} \widetilde{\A}'_j \ket{\pmb{0}}}}{\gamma_j\nrm{\Pi_{\mathrm{mg}}^{(j)} \widetilde{\A}_j\widetilde{\A}'_{j-1} \ket{\pmb{0}}}}\in \left[1-\frac{\eps}{2\widetilde{m}}, 1+\frac{\eps}{2\widetilde{m}}\right]$. Then $\Gamma:=\prod_{j=1}^{\widetilde{m}}\gamma_j$ is a good enough estimate since
		\begin{align*}
		\prod_{j=1}^{\widetilde{m}} \frac{\nrm{\Pi_{\mathrm{mg}}^{(j)} \widetilde{\A}'_j \ket{\pmb{0}}}}{\nrm{\Pi_{\mathrm{mg}}^{(j)} \widetilde{\A}_j\widetilde{\A}'_{j-1} \ket{\pmb{0}}}}
		=\prod_{j=1}^{\widetilde{m}} \frac{\nrm{\Pi_{\mathrm{mg}}^{(\widetilde{m})} \widetilde{\A}_{\widetilde{m}}\cdot\ldots\cdot \widetilde{\A}_{j+1}\widetilde{\A}'_j \ket{\pmb{0}}}}{\nrm{\Pi_{\mathrm{mg}}^{(\widetilde{m})} \widetilde{\A}_{\widetilde{m}}\cdot\ldots\cdot \widetilde{\A}_{j+1}\widetilde{\A}_j\widetilde{\A}'_{j-1} \ket{\pmb{0}}}}
		=\frac{\nrm{\Pi_{\mathrm{mg}}^{(\widetilde{m})} \widetilde{\A}'\ket{\pmb{0}}}}{\nrm{\Pi_{\mathrm{mg}}^{(\widetilde{m})} \widetilde{\A}\ket{\pmb{0}}}}.
		%=\frac{\nrm{\Pi \widetilde{\A}' \ket{\pmb{0}}}}{\nrm{\Pi \widetilde{\A} \ket{\pmb{0}}}}.
		\end{align*}
		In order to get an estimate of $\gamma_j$ it suffices to estimate both $\nrm{\Pi_{\mathrm{mg}}^{(j)} \widetilde{\A}'_j \ket{\pmb{0}}}$ and $\nrm{\Pi_{\mathrm{mg}}^{(j)} \widetilde{\A}_j\widetilde{\A}'_{j-1} \ket{\pmb{0}}}$ with multiplicative precision $1\pm\frac{\eps}{5\widetilde{m}}$. Note that such an estimate can be computed with success probability at least $1-\delta/(2\widetilde{m})$ with complexity that is at most $\bigO{\frac{\widetilde{m}}{\eps}\log\left(\frac{\widetilde{m}}{\delta}\right)}$ times bigger than the complexity of the algorithm $\widetilde{\A}'$. Since we need to compute only $\widetilde{m}=\bigO{\log(T'_{\max})}$ such estimates, the complexity bound follows from the complexity bound on $\widetilde{\A}'$. \anote{For next version: Here in the last step one can probably save a factor of $m$ with a more careful bound.}
	\end{proof}
	
\section{Linear system solving using blocks of unitaries}\label{sec:tools}

Given a way to implement a block-encoding for some matrix $A$, there are a number of useful basic operations one can do. We have already seen how the product of two block-encodings for $A$ and $B$ respectively gives a block-encoding for $AB$ (Lemma \ref{lemma:disjointAncillaProduct}), and how, given a block-encoding for $A$, we can implement a block-encoding for $A^{-c}$, for some $c\in(0,\infty)$ (Lemma \ref{lem:negative-power-restated}). 

Given a block-encoding $U$ of $A$, and a state $\ket{b}$, it is straightforward to approximate the state $A\ket{b}/\nrm{A\ket{b}}$, by applying $U$ to $\ket{b}$, and then using amplitude amplification on the $\ket{0}A\ket{b}$ part of the resulting state. For convenience, we make this precise in the following lemma.

\begin{lemma}[Applying a block-encoded matrix to a quantum state]\label{lem:block-encoding-to-state}
Fix any $\eps\in(0,1/2)$. Let $A\in \mathbb{C}^{N\times N}$ such that $\nrm{A}\leq 1$, and $\ket{b}$ a normalized vector in $\mathbb{C}^N$ such that $\nrm{A\ket{b}}\geq \gamma$. Suppose $\ket{b}$ can be generated in complexity $T_b$, and there is an $(\alpha,a,\epsilon)$-block-encoding of $A$ for some $\alpha\geq 1$, with $\epsilon\leq {\eps\gamma}/{2}$, that can be implemented in cost $T_A$. Then there is a quantum algorithm with complexity 
$$\bigO{\min\left(\frac{\alpha(T_A+T_b)}{\gamma},\frac{\alpha T_A\log\left(\frac{1}{\epsilon}\right)+T_b}{\gamma}\right)},$$ 
that terminates with success with probability at least $\frac{2}{3}$, and upon success generates the state $A\ket{b}/\nrm{A\ket{b}}$ to precision $\eps$. 
\end{lemma}
\begin{proof}
First we prove the first complexity upper bound.
Let $U$ be the block-encoding of $A$ referred to in the statement of the lemma, so 
\begin{align*}
\nrm{A-\alpha(\bra{0}^{\otimes a}\otimes I)U(\ket{0}^{\otimes a}\otimes I)} &\leq \epsilon\\
\nrm{\frac{1}{\alpha}A\ket{b} - (\bra{0}^{\otimes a}\otimes I)U(\ket{0}^{\otimes a}\otimes \ket{b})} &\leq \epsilon/\alpha.
\end{align*}

By generating $\ket{b}$ and then applying $U$, in cost $T_b+T_A$, we get a state that is $\epsilon/\alpha$-close to a state of the form 
$$\ket{0}^{\otimes a}\left(\frac{1}{\alpha}A\ket{b}\right)+\ket{0^\bot},$$
for some unnormalized state $\ket{0^\bot}$ that is orthogonal to every state with $\ket{0}^{\otimes a}$ in the first register. We have $\nrm{\frac{1}{\alpha}A\ket{b}}\geq \gamma/\alpha,$
so using ${\alpha}/{\gamma}$ rounds of amplitude amplification on $\ket{0}^{\otimes a}$ in the first register, we can get within a constant of a state that is $\frac{\epsilon}{\alpha}\frac{\alpha}{\gamma}=\frac{\epsilon}{\gamma}$-close to $\ket{0}^{\otimes a}\frac{A\ket{b}}{\nrm{A\ket{b}}}$ (because the error is also amplified by the amplitude amplification). Since $\epsilon\leq {\eps \gamma}/{2}$, the error on the $\ket{0}^{\otimes a}$ part of the state is at most $\eps/2$. Thus, if we measure a $\ket{0}^{\otimes a}$ in the first register at this stage, we will be within $\eps/2$ of the desired state. 

To get the second complexity bound we first amplify the block-encoding resulting in a unitary $U'$ that is a $(\sqrt{2},a,2\epsilon)$-block-encoding of $A$, that can by implemented in complexity $T'_A:=\bigO{\alpha T_A\log\left(\frac{1}{\epsilon}\right)}$ due to Lemma~\ref{lem:uniformBlockAmp}. Then we use $U'$ in the previous argument, replacing $T_A$ with $T'_A$ and $\alpha$ with $\sqrt{2}$.
%With $\bigO{\log\frac{1}{\eps}}$ rounds of amplification, we can bring the total error within $\eps$. 
\end{proof}

Given a block-encoding of $A$, we can implement a block-encoding of $A^{-c}$, from which we can approximately generate the state $A^{-c}\ket{b}/\nrm{A^{-c}\ket{b}}$ given a circuit for generating $\ket{b}$. When $c=1$, this is simply a quantum linear system solver, and more generally, we call this \emph{implementing negative powers of a Hamiltonian}. However, we can get a better algorithm for this problem using the technique of variable-time amplitude amplification, which we do in Section \ref{sec:qls}.

Although block-encodings are quite a general way of representing a matrix, we motivate them by connecting them to quantum data structures, showing that if a matrix is stored in a quantum data structure, in one of a number of possible ways, then there is an efficiently implementable block-encoding of the matrix. 

Specifically, for $p\in [0,1]$, define $\mu_p(A)=\sqrt{s_{2p}(A)s_{2(1-p)}(A^T)}$, where $s_{q}(A)=\max_j\nrm{A_{j,\cdot}}_q^q$ is the $q$-th power of the maximum $q$-norm of any row of $A$. We let $A^{(p)}$ denote the matrix of the same dimensions as $A$, with $A^{(p)}_{j,k}=(A_{j,k})^p$. The following was proven in \cite{kerenidis:quantumgraddescent}, although not in the language of block-encodings. We include the proof of \cite{kerenidis:quantumgraddescent} for completeness. 

\begin{lemma}[Implementing block-encodings from quantum data structures]\label{lem:kp}
Let $A\in\mathbb{C}^{M\times N}$.
\begin{enumerate}
\item Fix $p\in [0,1]$. If $A\in\mathbb{C}^{M\times N}$, and $A^{(p)}$ and $(A^{(1-p)})^\dagger$ are both stored in quantum-accessible data structures\footnote{\label{foot:precData}Here we assume that the datastructure stores the matrices with sufficient precision.}, then there exist unitaries $U_R$ and $U_L$ that can be implemented in time $\bigO{\polylog(MN/\eps)}$ such that 
$U_R^\dagger U_L$ is a $(\mu_p(A),\ceil{\log (N+M+1)},\eps)$-block-encoding of $\overline{A}$. 
%$A/\mu_p(A)=(I\otimes \bra{0}^{\otimes\ell}) U_R^\dagger U_L(I\otimes \ket{0}^{\otimes \ell'})$, where $\ell=\log N+2$ and $\ell' = \log M+2$.
\item On the other hand, if $A$ is stored in a quantum-accessible data structure\textsuperscript{\emph{\ref{foot:precData}}}, then there exist unitaries $U_R$ and $U_L$ that can be implemented in time $\bigO{\polylog(MN)/\eps}$ such that $U_R^\dagger U_L$ is a $(\nrm{A}_F,\ceil{\log(M+N)},\eps)$-block-encoding of $\overline{A}$.
%$A/\nrm{A}_F=(I\otimes \bra{0}^{\otimes\ell})U_R^\dagger U_L(I\otimes \ket{0}^{\otimes \ell'})$ for $\ell=\log N$ and $\ell'=\log M$.
\end{enumerate}
\end{lemma}
(Note that in the above lemma one could replace $\overline{A}$ by $A$, the proof remains almost the same.)

This allows us to apply our block-encoding results in the quantum data structure setting, including Hamiltonian simulation (Section \ref{sec:Ham-sim}), quantum linear system solvers (Section \ref{sec:qls}) and implementing negative powers of a Hamiltonian (Section \ref{sec:qls}). 

We now proceed with the proof of Lemma \ref{lem:kp}. 

\begin{proof}
Similarly to \cite[Theorem 4.4]{kerenidis:quantumgraddescent}, for $j\in [M]$, we define
$$\ket{\psi_j}=\frac{\sum_{k\in [N]} A_{j,k}^p\ket{j,M+k}}{\sqrt{s_{2p}(A)}}+\sqrt{1-\frac{\sum_{k\in[N]}A_{j,k}^{2p}}{s_{2p}(A)}}\ket{j,N+M+1}$$
and for $k\in [N]$, define
$$\ket{\psi_{M+k}}=\frac{\sum_{j\in [M]}A_{j,k}^{1-p}\ket{M+k,j}}{\sqrt{s_{2(1-p)}(A^T)}}+\sqrt{1-\frac{\sum_{j\in[M]}A_{j,k}^{2(1-p)}}{s_{2(1-p)}(A^T)}}\ket{M+k,M+N+1}.
$$

For $j\in[M]$, define
$$\ket{\phi_j}=\frac{\sum_{k\in [N]}A_{j,k}^p\ket{M+k,j} }{\sqrt{s_{2p}(A)}} + \sqrt{1-\frac{\sum_{k\in [N]}A_{j,k}^{2p}}{s_{2p}(A)}}\ket{M+N+1,j},$$
and for $k\in [N]$, define 
$$\ket{\phi_{M+k}}=\frac{\sum_{j\in [M]} A_{j,k}^{1-p}\ket{j,M+k}}{\sqrt{s_{2(1-p)}(A^T) }}+\sqrt{1-\frac{\sum_{j\in[M]} A_{j,k}^{2(1-p)}}{s_{2(1-p)}(A^T)}}\ket{M+N+1,M+k}.$$

Then for $j,j'\in [M]$, and $k,k'\in[N]$, we have
$$\braket{\psi_j}{\phi_{j'}}=\braket{\psi_{M+k}}{\phi_{M+k'}}=0,$$
but for $(j,k)\in [M]\times [N]$, we have:
$$\braket{\psi_j}{\phi_{M+k}}=
\frac{A_{j,k}}{\mu_p(A)} \; \mbox{and} \; \braket{\psi_{M+k}}{\phi_j}=\frac{A_{j,k}}{\mu_p(A)}=\frac{A_{k,j}^T}{\mu_p(A)}.$$
Thus, for any $i,i'\in [M+N]$, $\braket{\psi_i}{\phi_{i'}}=\overline{A}_{i,i'}/\mu_p(A)$. 

Letting $\ell=\ceil{\log (M+N+1)}$, we define a unitary $U_R$ on $\mathbb{C}^{(M+N)\times 2^{\ell}}$ by 
$$U_R:\ket{i}\ket{0^\ell}\mapsto \ket{\psi_i}$$
for all $i\in [M+N]$. Similarly, we define a unitary $U_L$ on $\mathbb{C}^{(M+N)\times 2^{\ell}}$ by
$$U_L:\ket{i}\ket{0^{\ell}}\mapsto \ket{\phi_i}$$
for all $i\in [M+N]$. Then the first result follows.

For the second result, we define, for each $j\in [M]$, 
$$\ket{\psi_j} = \sum_{k\in [N]}\frac{A_{j,k}}{\nrm{A_{j,\cdot}}}\ket{j,M+k}\;\mbox{and}\; 
\ket{\phi_j} = \sum_{k\in [N]}\frac{A_{j,k}}{\nrm{A_{j,\cdot}}}\ket{M+k,j},$$
and for each $k\in[N]$,
$$\ket{\psi_{M+k}}=\sum_{j\in[M]}\frac{\nrm{A_{j,\cdot}}}{\nrm{A}_F}\ket{M+k,j}
\;\mbox{and} \;
\ket{\phi_{M+k}}=\sum_{j\in[M]}\frac{\nrm{A_{j,\cdot}}}{\nrm{A}_F}\ket{j,M+k}.$$
These vectors can be constructed from the quantum-accessible data structure described in Theorem \ref{theorem:data_structure}, and we have, for any $j,j'\in [M]$ and $k,k'\in [N]$:
$$\braket{\psi_j}{\phi_{j'}}=\braket{\psi_{M+k}}{\phi_{M+k'}}=0,\;  \braket{\psi_j}{\phi_{M+k}} = \frac{A_{j,k}}{\nrm{A}_F}\;\mbox{and}\; \braket{\psi_{M+k}}{\phi_j}=\frac{A_{k,j}^T}{\nrm{A}_F}.$$
The second result follows, similarly to the first.
\end{proof}

\subsection{Hamiltonian simulation with quantum data structure}\label{sec:Ham-sim}

Low and Chuang~\cite{LowChuangQubitization2016} showed how to implement an optimal Hamiltonian simulation algorithm given a block-encoding of the Hamiltonian (Theorem~\ref{thm:blockHamSim}). Their result combined with Lemma~\ref{lem:kp} gives the following:
\begin{theorem}[Hamiltonian simulation using quantum data structure]\label{thm:dataHamSim}
For any $t\in\mathbb{R}$ and $\eps\in(0,1/2)$, we have the following:
\begin{enumerate}
\item Fix $p\in[0,1]$. Let $H\in\mathbb{C}^{N\times N}$ be a Hermitian matrix, and suppose $H^{(p)}$ and $(H^{(1-p)})^\dagger$ are stored in quantum-accessible data structures\textsuperscript{\emph{\ref{foot:precData}}}. Then we can implement a unitary $\widetilde{U}$ that is a $(1,n+3,\eps)$-block-encoding of $e^{it H}$ in time $\bigOt{t\mu_p(A)\polylog(N/\eps)}$. 
\item If $H$ is stored in a quantum-accessible data structure\textsuperscript{\emph{\ref{foot:precData}}}, then we can implement a unitary $\widetilde{U}$ that is a $(1,n+3,\eps)$-block-encoding of $e^{itH}$ in time $\bigOt{t\nrm{A}_F\polylog(N/\eps)}$. 
\end{enumerate}
\end{theorem}

\subsection{Quantum singular value estimation}

The quantum singular value estimation (QSVE) problem is the following\footnote{Note that the presented definition is a somewhat relaxed version of singular value estimation. Here we allow producing an entangled auxiliary/garbage state, which can be undesirable in certain scenarios. There are ways around this issue, see for example Ta-Shma's consistent phase estimation \cite{TaShma2013PhaseEst} or the singular value transformation results of Gilyén et al.~\cite{gilyenBlockMatrices}.}: Given access to a matrix $A\in \mathbb{R}^{M\times N}$ with singular value decomposition $A=\sum_j\sigma_j\ket{u_j}\bra{v_j}$, and given input $\ket{\psi}=\sum_j c_j\ket{u_j}$, output $\sum_j c_j\ket{\phi(\sigma_j)}\ket{u_j}$, where $\ket{\phi(\sigma_j)}$ is a unit vector on a space with a phase register and an auxiliary register, such that, when the phase register is measured, it outputs an estimate of $\sigma_j$, $\tilde{\sigma}_j$ such that with probability $1-\eps$, $|\sigma_j-\tilde{\sigma}_j|\leq \Delta$. 

Kerenidis and Prakash~\cite{kerenidis2016quantum} gave a quantum algorithm for estimating singular values wherein they showed that if a matrix $A$ is stored in a quantum-accessible data structure, the singular values of $A$ can be estimated to a precision $\delta$ in time $\bigOt{(\mu/\delta)\polylog(MN)}$, where $\mu=\nrm{A}_F$, or if $A^{(p)}$ and $A^{(1-p)}$ are stored in quantum-accessible data structures for some $p$, $\mu=\mu_p(A)$. We provide an alternative quantum algorithm for singular value estimation when the matrix $A$ is given as a block-encoding --- the quantum-accessible data structure case considered by Kerenidis and Prakash is a special case of this. In the scenario where $A$ is stored in a quantum-accessible data structure, we recover the same running time as Kerenidis and Prakash. 

A subsequent improved version of this algorithm, without the need for an auxiliary register in addition to the phase register, can be found in \cite{gilyenBlockMatrices}.

\begin{theorem}[Quantum singular value estimation]\label{thm:sve}
Let $\eps,\Delta\in (0,1)$, and $\eps' = \frac{\eps\Delta}{4\log^2(1/\Delta)}$. Let $U$ be an $(\alpha,a,\eps')$-block-encoding of a matrix $A$ that can be implemented in cost $T_U$.  Then we can implement a quantum algorithm that solves QSVE of $A$ in complexity
$$\bigO{\frac{\alpha}{\Delta}(a+T_U)\mathrm{polylog}(1/\eps)}.$$
\end{theorem}

\begin{proof} At a high-level, the algorithm works by using phase estimation of Hamiltonian simulation of $A$, however, $A$ is not necessarily Hermitian, so we instead use $\overline{A}$.

\paragraph{Hamiltonian simulation of $\overline{A}$:}
	Let $A=\sum_{j=1}^r \sigma_j\ket{u_j}\bra{v_j}$, where $r=\mathrm{min}\{M,N\}$, $\{\sigma_j\}_j$ are the singular values of $A$, while $\ket{u_j}$ ($\ket{v_j}$) are the left (right) singular vectors of $A$. Then the matrix
	$$\overline{A}=\begin{bmatrix}
	0 & A\\
	A^\dagger & 0
	\end{bmatrix},$$ 
	has non-zero eigenvalues $\{\pm\sigma_1,....,\pm\sigma_r\}$ and corresponding eigenvectors $$\ket{\lambda^{\pm}_j}=\dfrac{1}{\sqrt{2}}\left(\ket{0}\ket{u_j}\pm\ket{1}\ket{v_j}\right).$$

	Observe that for all $j\in [r]$, $\ket{0}\ket{u_j}=\left(\ket{\lambda^+_j}+\ket{\lambda^-_j}\right)/\sqrt{2}$. The remaining zero eigenvalues of $\overline{A}$ belong to $\mathrm{span}\{\ket{v_1},...,\ket{v_r}\}^\perp$. So any quantum state $\ket{\psi}$ that is spanned by the right singular vectors will have no support on the zero eigenspace of $\overline{A}$.

	If $U$ is an $(\alpha,a,\eps')$-block-encoding of $A$, then the unitary 
	$$\overline{U}=\begin{bmatrix}
	0 &  U\\
	U^\dag & 0
	\end{bmatrix},
	$$
	composed with appropriate SWAP gates is an $(\alpha, 2a, \eps')$-block encoding of $\overline{A}$. 
	So if $T_U$ is the cost of implementing the unitary $U$, then the cost of implementing $\overline{U}$ is $2T_U+\mathcal{O}(1)$. 

	From Lemma~\ref{lem:controlled}, there exists a unitary $V$ that is an $(1,2a+2,\eps/2)$-block encoding of $\sum_{t=-\frac{T-1}{2}}^{\frac{T-1}{2}}\ket{t}\bra{t}\otimes e^{i t\overline{A}}$ that can be implemented in complexity 
	$$\mathcal{O}\left(\left(\frac{\alpha}{\Delta}+\log\left(\frac{1}{\eps}\log\frac{1}{\Delta}\right)\right)(a+T_U)\right).$$
This will be our main subroutine.

\paragraph{Dirichlet kernels:} We will use the fact that for all $x$, $|\sin x| \leq |x|$, and for all $x\in[-\pi/2,\pi/2]$, $|\sin x|\geq |2x/\pi|$.

Let $D_n(x)$ be the Dirichlet kernel, defined:
$$D_n(x)=\frac{\sin((n+1/2)x)}{\sin(x/2)}=\sum_{t=-n}^n \cos(2tx)$$
This function is peaked around 0, with the peak becoming more extreme as $n$ increases. 
We will make use of a few easily verified facts about $D_n(x)$:
\begin{itemize}
\item $D_n(x)=D_n(-x)$
\item $D_n(x)\leq \frac{1}{\sin(x/2)}\leq \frac{\pi}{|x|}$ for $x\in [-\pi,\pi]$
\item $D_n(x) = \frac{\sin((n+1/2)x)}{\sin(x/2)}\geq \frac{(2n+1)\delta}{\pi \delta/2}=\frac{2(2n+1)}{\pi}$ for $x\in [0,\frac{\pi}{2n+1}]$.
\end{itemize}

\paragraph{Algorithm:} We now describe the sve algorithm. Let $T$ be an odd number such that $T\geq 2\pi/\Delta$. Begin by generating the state:
\begin{eqnarray*}
 \sum_{t=-\frac{T-1}{2}}^{\frac{T-1}{2}}\frac{1}{\sqrt{T}}\ket{t}\ket{+}\ket{0}\ket{\psi}
&=& \sum_jc_j \sum_{t=-\frac{T-1}{2}}^{\frac{T-1}{2}}\frac{1}{\sqrt{2T}}\ket{t}\left(\ket{0}\ket{0}\ket{u_j}+\ket{1}\ket{0}\ket{u_j}\right)\\
&=& \sum_jc_j \sum_{t=-\frac{T-1}{2}}^{\frac{T-1}{2}}\frac{1}{2\sqrt{T}}\ket{t}\Big(\ket{0}(\ket{\lambda_j^+}+\ket{\lambda_j^-})+\ket{1}(\ket{\lambda_j^+}+\ket{\lambda_j^-}).
%&&\qquad\qquad +\ket{01} (\ket{\tau_j^+}-\ket{\tau_j^-}) + \ket{11}(\ket{\tau_j^+}-\ket{\tau_j^-})\Big).
\end{eqnarray*}
Next, we apply $e^{(-1)^bt\overline{A}}$ conditioned on $\ket{t}\ket{b}$ in the first registers, %and $e^{(-1)^bt\overline{A^T}}$ conditioned on $\ket{t}\ket{b1}$ in the first registers, 
to get:
\begin{eqnarray*}
&&\sum_jc_j \sum_{t=-\frac{T-1}{2}}^{\frac{T-1}{2}}\frac{1}{2\sqrt{T}}\ket{t}\Big(\ket{0}(e^{it\sigma_j}\ket{\lambda_j^+}+e^{-it\sigma_j}\ket{\lambda_j^-})+\ket{1}(e^{-it\sigma_j}\ket{\lambda_j^+}+e^{it\sigma_j}\ket{\lambda_j^-}).
\end{eqnarray*}
We perform a Hadamard gate on the second register, to get:
\begin{eqnarray}
&&\sum_jc_j \sum_{t=-\frac{T-1}{2}}^{\frac{T-1}{2}}\frac{1}{\sqrt{2T}}\ket{t}\left(\cos(t\sigma)\ket{0}(\ket{\lambda_j^+}+\ket{\lambda_j^-})+i\sin(t\sigma)\ket{1}(\ket{\lambda_j^+}-\ket{\lambda_j^-})\right)\nonumber\\
&=& \sum_jc_j \sum_{t=-\frac{T-1}{2}}^{\frac{T-1}{2}}\frac{1}{\sqrt{T}}\ket{t}\left(\cos(t\sigma)\ket{0}\ket{0}\ket{u_j}+i\sin(t\sigma)\ket{1}\ket{1}\ket{v_j}\right).
\label{eq:pre-amp-amp}
\end{eqnarray}

We are interested in the part of the state with $\ket{0}$ in the second register. This part of the state has squared amplitude:
\begin{eqnarray*}
\beta^2 := \sum_{t=-\frac{T-1}{2}}^{\frac{T-1}{2}}\frac{1}{{T}}\cos^2(t\sigma_j) &=& \frac{1}{2T}\sum_{t=-\frac{T-1}{2}}^{\frac{T-1}{2}}\left(1+\cos(2t\sigma_j)\right)\\
&= & \frac{1}{2}+\frac{1}{2T}D_{\frac{T-1}{2}}(2\sigma_j) %\qquad\mbox{by Lagrange's identity}\\
\;\geq \; \frac{1}{2} - \frac{1}{2T}\frac{T-1}{4} 
\; = \; \frac{3}{8},
\end{eqnarray*}
where in the last line, we have used a lower bound on the Dirichlet kernel, $D_{\frac{T-1}{2}}(2\sigma_j)\geq - C\frac{T-1}{2}$, where $C<1/2$ is the absolute minimum of $2\sin(x)/x$ (See, for example, \cite{dirichlet-note}). %Thus, if we repeat this procedure in parallel $\log(1/\eps)$ times (each repetition using the same copy of $\ket{\psi}$, but fresh ancilla for all other registers), then each 
%measuring the second register in \eqref{eq:pre-amp-amp}, and post-selecting on $\ket{0}$, with probability $1-\eps$, we will succeed in measuring $\ket{0}$, leaving the state:
We will consider the $\ket{1}$ part of the state the ``bad'' part of the state, and the $\ket{0}$ part of the state the ``good'' part of the state. We analyze the remainder of the algorithm's action on the ``good'' part of the state, which is:
\begin{equation*}
\beta^{-1}\sum_jc_j \sum_{t=-\frac{T-1}{2}}^{\frac{T-1}{2}}\frac{1}{\sqrt{2T}}\ket{t}\cos(t\sigma_j)\ket{0}(\ket{\lambda_j^+}+\ket{\lambda_j^-}) = \beta^{-1}\sum_jc_j \sum_{t=-\frac{T-1}{2}}^{\frac{T-1}{2}}\frac{1}{\sqrt{T}}\ket{t}\cos(t\sigma_j)\ket{0}\ket{u_j}.
\end{equation*}
Discarding the second register, and performing a Fourier transform on the first, we get:
\begin{eqnarray}
&& \beta^{-1}\sum_jc_j \sum_{z=-\frac{T-1}{2}}^{\frac{T-1}{2}}\left(\frac{1}{T}\sum_{t=-\frac{T-1}{2}}^{\frac{T-1}{2}}\cos(t\sigma_j)e^{i2\pi tz/T}\right)\ket{z}\ket{u_j}\nonumber\\
&=& \beta^{-1}\sum_jc_j \sum_{z=-\frac{T-1}{2}}^{\frac{T-1}{2}}\left( \frac{1}{2T}\sum_{t=-\frac{T-1}{2}}^{\frac{T-1}{2}} e^{it(2\pi z/T+\sigma_j)} 
+\frac{1}{2T}\sum_{t=-\frac{T-1}{2}}^{\frac{T-1}{2}} e^{it(2\pi z/T-\sigma_j)}
\right)\ket{z}\ket{u_j}\nonumber\\
&=& \beta^{-1}\sum_jc_j \sum_{z=-\frac{T-1}{2}}^{\frac{T-1}{2}} \frac{1}{2T}\left(D_{\frac{T-1}{2}}\left(\frac{2\pi z}{T}+\sigma_j\right) 
+D_{\frac{T-1}{2}}\left(\frac{2\pi z}{T}-\sigma_j\right)
\right)\ket{z}\ket{u_j}
.\nonumber
\end{eqnarray}
Thus, we are adding two functions of $\ket{z}$: one peaking around $z\approx -\frac{T\sigma_j}{2\pi}$, and the other around $z\approx \frac{T\sigma_j}{2\pi}$. Finally, we can reversibly map $\ket{z}$ to $\ket{\mathrm{sgn}(z)}\ket{|z|}$, where $\ket{\mathrm{sgn}(z)}$ is 0 if $z=0$, $+$ if $z>0$ and $-$ if $z<0$. Then we have:
\begin{equation}
\ket{\phi(\sigma_j)}=\beta^{-1}\sum_{z=-\frac{T-1}{2}}^{\frac{T-1}{2}} \frac{1}{2T}\left(D_{\frac{T-1}{2}}\left(\frac{2\pi z}{T}+\sigma_j\right) 
+D_{\frac{T-1}{2}}\left(\frac{2\pi z}{T}-\sigma_j\right)
\right)\ket{\mathrm{sgn}(z)}\ket{|z|}.\nonumber
\end{equation}
We will interpret $|z|$ as an estimate $\pi|z|/T$ of $\sigma_j$. 

\paragraph{Correctness:} Fix some $j$, and let $z^*$ be the closest integer to $\frac{T\sigma_j}{2\pi}$. Define $\delta = \sigma_j - \frac{2\pi z^*}{T}$, so $|\delta|\leq \frac{\pi}{T}$. We will first argue that if we measure the last register of $\ket{\phi(\sigma_j)}$, we will get $z^*$ with probability at least $\frac{1}{4}(2/\pi - 1/4)^2>.037$. 
Suppose $z^*\neq 0$. In that case, we can measure either $\ket{-,z^*}$ or $\ket{+,z^*}$, with respective probabilities at least:
$$\beta^{-2}\left| \frac{1}{2T}\left( D_{\frac{T-1}{2}}\left(\frac{2\pi z^*}{T}+\sigma_j\right)+D_{\frac{T-1}{2}}\left(\frac{2\pi z^*}{T}-\sigma_j\right)\right)\right|^2$$
and
$$\qquad\qquad  \beta^{-2}\frac{1}{2T}\left|\left( D_{\frac{T-1}{2}}\left( - \frac{2\pi z^*}{T}+\sigma_j\right)+D_{\frac{T-1}{2}}\left(-\frac{2\pi z^*}{T}-\sigma_j\right) \right)\right|^2.$$
These are equal, by the symmetry of $D_{\frac{T-1}{2}}$, so the total probability of measuring $z^*$ is:
\begin{eqnarray*}
&& 2\beta^{-2}\left| \frac{1}{2T}\left( D_{\frac{T-1}{2}}\left(\frac{2\pi z^*}{T}+\sigma_j\right)+D_{\frac{T-1}{2}}\left(\frac{2\pi z^*}{T}-\sigma_j\right)\right)\right|^2\\
&\geq & \frac{1}{2T^2}\left| D_{\frac{T-1}{2}}(|\delta|) + D_{\frac{T-1}{2}}\left(\frac{2\pi z^*}{T}+\sigma_j\right) \right|^2\qquad\mbox{ since $\beta\leq 1$}\\
&\geq & \frac{1}{2T^2}\left| \frac{2T}{\pi} - \frac{T-1}{4} \right|^2
\; \geq \; \frac{1}{2}\left(2/\pi - 1/4\right)^2 \qquad\mbox{ since $|\delta|\leq \pi/T$}. 
\end{eqnarray*}
 The case for $z^*=0$ is similar, but the probability is only half as much, since there is only one contribution.

Next, for an integer $d$, let $p_d$ be the probability we measure $z^*+d$. We can upper bound this as
\begin{eqnarray*}
p_d &=& 2\beta^{-2}\left| \frac{1}{2T}\left( D_{\frac{T-1}{2}}\left(\frac{2\pi (z^*+d)}{T}+\sigma_j\right)+D_{\frac{T-1}{2}}\left(\frac{2\pi (z^*+d)}{T}-\sigma_j\right)\right)\right|^2\\
&=& 2\beta^{-2}\left| \frac{1}{2T}\left( D_{\frac{T-1}{2}}\left(\frac{2\pi (z^*+d)}{T}+\sigma_j\right)+D_{\frac{T-1}{2}}\left(\frac{2\pi d}{T}+\delta\right)\right)\right|^2\\
&\leq & \frac{1}{2T^2}\frac{8}{3}\left| \frac{T}{|2(z^*+d)+T\sigma_j/\pi|} + \frac{T}{|2d+T\delta/\pi|} \right|^2\\
&\leq & \frac{4}{3}\left| \frac{2}{2|d|-T|\delta|/\pi} \right|^2
\;\leq\; \frac{16}{3(2|d|-1)^2}
\; \leq \; \frac{16}{3d^2}
.
\end{eqnarray*}
Then the probability of measuring a $z$ such that $|z^* - z| \geq k$ for some $k<T$ is at most
$$2\sum_{d=k}^{\frac{T-1}{2}-z^*}\frac{16}{3d^2} =O(1/k).$$ %\anote{Could you include approximate value for $k=3$?}
Thus, as in \cite{CEMM98}, we can boost the success probability to $\eps$ of measuring $z$ such that $|z-z^*|<2$ at a cost of $\log(1/\eps)$ parallel repetitions. Note that in each repetition, only $\geq 3/8$ of the state will be ``good''. In any branch of the superposition, we will only compute the final estimate off of those repetitions where we're in the good state. 
The probability that we measure $z$ such that $\tilde\sigma = \pi z/T$ satisfies $|\tilde\sigma - \sigma_j|\leq \Delta$ is at least the probability that we measure $z$ such that $|\pi z^*/ T - \pi z/T| < 2\pi/T$, or equivalently, $|z^* - z| < 2$.
\end{proof}

In the case where $\ket{\psi}$ and $A$ are stored in a quantum-accessible data structure, we have that $\alpha=\mu$, where $\mu=\nrm{A}_F$ or if $A^{(p)}$ and $A^{(1-p)}$ are stored in quantum-accessible data structures for some $p$, $\mu=\mu_p(A)$. In each such case, we have $T_U=\bigO{\polylog(MN/\eps)}$ and $T_\psi=\bigO{\polylog(MN/\eps)}$, which gives us a running time of 
$$\bigOt{\dfrac{\mu}{\Delta}\polylog(MN/\eps)},$$
and thus we recover the running time of the QSVE algorithm by Kerenidis and Prakash \cite{kerenidis2016quantum}.

\subsection{Quantum linear system solvers}\label{sec:qls}

The quantum linear system problem (QLS problem) is the following. Given access to an $N\times N$ matrix $A$, and a procedure for computing a quantum state $\ket{b}$ in the image of $A$, prepare a state that is within $\eps$ of $A^+\ket{b}/\nrm{A^+\ket{b}}$. As with Hamiltonian simulation, several methods of encoding the part of the input $A$ have been considered. We consider the case where $A$ is given as a block-encoding. In that case, as a special case of Lemma \ref{lem:negative-power-restated} when $c=1$, given a $(\alpha,a,\delta)$-block-encoding $U$ of $A$ with implementation cost $T_U$, we can implement a $(2\kappa, a+\bigO{\log(\kappa\log({1}/{\eps}))},\eps)$-block-encoding of $A^{-1}$, assuming $\delta = o(\eps/(\kappa^2\log^3\frac{\kappa}{\eps}))$, in complexity $\bigO{\alpha\kappa(a+T_U)\polylog(\kappa/\eps)}$. \footnote{\shnote{added this~}Note that when the eigenvalues of $A$ are between $[-1,-1/\kappa]$, we follow the same argument as before by replacing $A$ with $-A$. Then the whole procedure involves applying quantum phase estimation to separate out the projection of $\ket{b}$ on to the positive eigenspace of $A$ from the projection of $\ket{b}$ on to the negative eigenspace of $A$. Controlled on this, we apply the procedure to implement  $A^{-1}$ or $-(-A)^{-1}$.}

From this block-encoding of $A^{-1}$, we can solve the QLS problem by applying the unitary $U$ to the state $\ket{b}$, and then doing amplitude amplification. However, since $U$ is a $(2\kappa,a+\bigO{\log(\kappa\log({1}/{\eps}))},\eps)$-block-encoding, we will require a number of amplitude amplification rounds that is linear in $\kappa$, giving an overall quadratic dependence on $\kappa$. Quantum linear system solvers in the sparse-access input model have only a linear dependence on $\kappa$, thanks to techniques of Ambainis \cite{AmbainisVariableTime12}, which were also successfully applied in a setting more similar to ours by Childs, Kothari and Somma \cite{childs2015quantum}. Using these techniques, we can also reduce our dependence on $\kappa$ to linear. 

\paragraph{Reducing the dependence on condition number.} To reduce the dependence on $\kappa$, we use the technique of variable time amplitude amplification (VTAA) instead of standard amplitude amplification. For this, we need to adapt the quantum linear systems algorithm to be a \emph{variable-stopping-time algorithm}, to which VTAA can be applied (see Section \ref{sec:variable-stopping}). Our setting is similar to that of Childs, Kothari and Somma \cite{childs2015quantum}, so we will follow their notation and proof.

First we formally state a version of quantum phase estimation that determines whether an eigenphase $\lambda\in [-1,1]$ of a given unitary $U$ satisfies $0\leq |\lambda|\leq\phi$ or $2\phi\leq |\lambda|\leq 1$. This is known as gapped quantum phase estimation (GPE) and was introduced in \cite{childs2015quantum}. We restate  it here. 
\begin{lemma}[Gapped phase estimation~\cite{childs2015quantum}]\label{lem:gapped-qpe} 
Let $U$ be a unitary such that $U\ket{\lambda}=e^{i\lambda}\ket{\lambda}$ and $\lambda\in [-1,1]$. Let $\phi\in(0,1/4]$ and $\epsilon>0$. Then there exists a quantum algorithm that performs the transformation
$$\ket{0}_C\ket{0}_P\ket{\lambda}_I\mapsto \alpha_0\ket{0}_C\ket{g_0}_P\ket{\lambda}_I+\alpha_1\ket{1}_C\ket{g_1}_P\ket{\lambda}_I,$$
for some unit vectors $\ket{g_0}$ and $\ket{g_1}$, where $C$ and $P$ are registers of $1$ and $\bigO{\frac{1}{\phi}\log\frac{1}{\epsilon}}$ qubits, respectively, $|\alpha_0|^2+|\alpha_1|^2=1$ and
\begin{itemize}
\item[$\bullet$] if $0\leq |\lambda|\leq \phi$ then $|\alpha_1|\leq\epsilon$ and
\item[$\bullet$] if $2\phi\leq |\lambda|\leq 1$ then $|\alpha_0|\leq\epsilon$.
\end{itemize}
If $T_U$ is the cost of implementing $U$, then the cost of this quantum algorithm is $\bigO{\frac{T_U}{\phi}\log\frac{1}{\epsilon}}$. 
\end{lemma}

As a corollary to Lemma \ref{lem:negative-power-restated}, which, in the special case when $c=1$, says we can get a block-encoding of $H^{-1}$ from a block-encoding of $H$, we have the following, which allows us to invert $H$ on a certain range of its eigenspaces:
\begin{corollary}[Efficient inversion of a block-encoded matrix]\label{cor:block-encoding-inverse-patch}
Let $\eps,\lambda>0$ and $H$ an $N\times N$ Hermitian matrix. Suppose that for $\delta=o\left(\eps\lambda^2/\log^3\frac{1}{\lambda\eps}\right)$ the unitary $U$ is an $(\alpha,a,\delta)$-block-encoding of $H$ that can be implemented using $T_U$ elementary gates. Then for any state $\ket{\psi}$ that is spanned by eigenvectors of $H$ with eigenvalues in the range $[-1,-\lambda]\bigcup [\lambda,1]$, there exists a unitary $W(\lambda,\eps)$ that implements
\begin{equation*}
W(\lambda,\eps)\ket{0}_F\ket{0}_Q\ket{\psi}_I=\dfrac{1}{\alpha_\mathrm{max}}\ket{1}_F\ket{0}_Q f(H)\ket{\psi}_I +\ket{0}_F\ket{\widetilde\psi^\perp}_{QI}
\end{equation*}
where $\alpha_\mathrm{max}\leq \lambda$ is a constant, $\ket{\widetilde\psi^\perp}_{QI}$ is an unnormalized quantum state, orthogonal to $\ket{0}_Q$ and $\nrm{f(H)\ket{\psi}-H^{-1}\ket{\psi}}\leq \eps$. Here $F$, $Q$ and $I$ are registers of $1$ qubit, $\alpha$ qubits and $\log(N)$ qubits respectively. The cost of implementing $W(\lambda,\eps)$ is 
\begin{equation}
\bigO{\alpha \lambda^{-1}\log^2\left(\frac{1}{\lambda\eps}\right)(a+T_U)}.\label{eq:patch-cost}
\end{equation}
\end{corollary}
\begin{proof}
Let $H_{\lambda}$ be the restriction of $H$ to its eigenspaces with corresponding eigenvalues in $[-1,-\lambda]\bigcup[\lambda,1]$. An application of Lemma~\ref{lem:negative-power-restated} with $c=1$ and $\kappa=\lambda^{-1}$ yields a $(2\lambda^{-1},a+\bigO{\log(\lambda^{-1}\log\frac{1}{\lambda\eps})},\eps)$ block encoding of $H_{\lambda}^{-1}$. That is, there exists a unitary $\widetilde{U}$ such that
\begin{equation*}
%\label{eq:block-encoding-inverse-patch}
\widetilde{U}\ket{0}_Q\ket{\psi}_I=\dfrac{\lambda}{2}\ket{0}_Qf(H)\ket{\psi}_I+\ket{\widetilde{\psi}^\perp}_{QI},
\end{equation*}
where $\nrm{f(H)\ket{\psi}-H^{-1}\ket{\psi}}\leq\eps$ whenever $\ket{\psi}$ is a unit vector in the span of eigenvectors of $H$ with eigenvalues in $[-1,-\lambda]\cup [\lambda,1]$. By Lemma \ref{lem:negative-power-restated}, such a $\widetilde{U}$ can be implemented in complexity given by the expression in \eqref{eq:patch-cost}.

If we add a single qubit flag register, initialized to $\ket{0}_F$ to the aforementioned procedure, and flip this register controlled on the register $Q$ being in the state $\ket{0}$, then the resulting unitary $\widetilde{U'}$ acts as
\begin{equation*}
%\label{eq:block-encoding-inverse-with-control}
\ket{0}_F\ket{0}_Q\ket{\psi}_I\mapsto\dfrac{\lambda}{2}\ket{1}_F\ket{0}_Q f(H)\ket{\psi}_I+\ket{0}_F\ket{\widetilde{\psi}^\perp}_{QI}.
\end{equation*}  
Finally, we implement the following rotation controlled on register $Q$ being in $\ket{0}$ that replaces $\lambda/2$ with some constant $\alpha_{\mathrm{max}}^{-1}$, independent of $\lambda$, such that\footnote{Note that we can assume without loss of generality that $\lambda\alpha_{\max}\geq 2$.} $\alpha_{\max}=\bigO{\kappa}$. That is we implement
\begin{equation*}
\ket{1}_F\mapsto \dfrac{2}{\lambda\alpha_\mathrm{max}}\ket{1}_F+\sqrt{1-\dfrac{4\lambda^2}{\alpha^2_\mathrm{max}}}\ket{0}_F \,.
\end{equation*}
These two operations together give us $W(\lambda,\eps)$ and the cost of implementing this is of the same order as that of implementing $\widetilde{U}$.
\end{proof}
\paragraph{Variable-time algorithm.} Now we will describe a variable time algorithm $\mathcal{A}$ that, given a block-encoding of an $N\times N$ matrix $A$, can be applied to an input state $\ket{\psi}$ to produce a state close to $A^{-1}\ket{\psi}$. The algorithm $\mathcal{A}$ can be thought of as a sequence of steps $\mathcal{A}_1,\dots,\mathcal{A}_m$, where $m=\ceil{\log_2\kappa}+1$. The goal is that the whole algorithm retains a block-encoded form so that it enables us to use this easily in applications in subsequent sections. $\mathcal{A}$ will work on the following registers:
\begin{itemize}
\item $m$ single-qubit clock registers $C_1,\dots,C_m$, collectively referred to as $C$.
\item An input register $I$, initialized to $\ket{\psi}$. 
\item A single qubit flag register $F$, used to indicate success.
\item $m$ registers $P_1,\dots,P_m$ used as ancilla for GPE.
\item An ancilla register $Q$ required for the block-encoding, initialized to $\ket{0}^{\otimes a}$.
\end{itemize}
Let $\epsilon'=\eps/(m\alpha_{\mathrm{max}})$. We define algorithm $\mathcal{A}_j$, as follows:
\begin{enumerate}
\item If $C_1,\dots,C_{j-1}$ is in the state $\ket{0}^{\otimes (j-1)}$, apply GPE to $e^{iA}$, defined in Lemma~\ref{lem:gapped-qpe}, with precision $2^{-j}$ and accuracy $\epsilon'$ to input $\ket{\psi}$, using $P_j$ as workspace, and writing the output qubit in $C_j$. 
\item If $C_j$ is now in the state $\ket{1}$, apply the unitary $W(2^{-j},\epsilon')$, as defined in Corollary~\ref{cor:block-encoding-inverse-patch}, on $I\otimes F\otimes Q$.
\end{enumerate}
We shall also require algorithms $\mathcal{A}'=\mathcal{A}'_m\dots \mathcal{A}'_1$ that are similar to ${\cal A}$ except that in step 2, ${\cal A}'_j$ implements the following:
\begin{equation*}
%\label{eq:step-2-a-prime-j}
W'\ket{\psi}_I\ket{0}_F\ket{0}_Q=\ket{\psi}_I\ket{1}_F\ket{0}_Q.
\end{equation*}
Then we can define the final variable time algorithm formally using the following lemma.\footnote{Note that the construction of the variable-time amplification algorithm can have a logarithmically higher complexity than the actual variable-time amplification algorithm itself, as in Theorem~\ref{thm:efficient-vtaa-vtae}. For simplicity throughout this section we only discuss the complexity of the resulting variable-time amplification algorithm. This is also justified by the fact that for a fixed input state preparation unitary we only need to construct the variable-time amplification algorithm once.}

\anote{We should change the statement stating the result for $M\times N$ matrices, since we use it in such a way later.}
\begin{theorem}[Variable-time quantum linear systems algorithm]\label{thm:qls-vtaa}%\shnote{thm:qls-vtaa}
Let $\kappa\geq 2$, and $H$ be an $N\times N$ Hermitian matrix\footnote{Since for any matrix $C\in\C^{M'\times N'}$ we have that $\overline{C}\in\C^{(M'+N')\times (M'+N')}$ is Hermitian, and the eigenvalues of $\overline{C}$ are $\pm 1$ times the singular values of $C$, this statement and its corollaries also apply to non-symmetric matrices.} such that the non-zero eigenvalues of $H$ lie in the range $[-1,-1/\kappa]\bigcup [1/\kappa,1]$. Suppose that for $\delta = \littleo{\eps/(\kappa^2\log^3\frac{\kappa}{\eps})}$ we have a unitary $U$ that is a $(\alpha,a,\delta)$-block-encoding of $H$ that can be implemented using $T_U$ elementary gates.  Also suppose that we can prepare an input state $\ket{\psi}$ which spans the eigenvectors of $H$ in time $T_{\psi}$. Then there exists a variable time quantum algorithm that outputs a state that is ${\eps}$-close to
$H^{-1}\ket{\psi}/\nrm{H^{-1}\ket{\psi}}$ at a cost 
$$\bigO{\kappa\left(\alpha \big(T_U+a\big)\log^2\left(\dfrac{\kappa}{\eps}\right)+ T_\psi\right)\log(\kappa)}.$$
%$$\mathcal{O}\left(\alpha \kappa (a+T_U)\log^2\left(\dfrac{\kappa\log\kappa}{\eps}\right)\log\left(\alpha\kappa\log^2\left(\dfrac{\kappa\log\kappa}{\eps}\right)(a+T_U)\right)\right).$$ 
\end{theorem}
\begin{proof}
Given an $(\alpha,a,\delta)$-block encoding of $H$, we append some ancilla qubits in order to be in the framework for applying VTAA to algorithm $\mathcal{A}$. At the end of the algorithm we will discard these additional registers. We append a single qubit flag register $F$, and registers $C$, $P$ and $Q$ (defined previously) all initialized in $\ket{0}$. So now we are in a framework where the VTAA can be applied to algorithm $\mathcal{A}$ to the state $\ket{\psi}_I\ket{0}_{CPFQ}$.  
The final algorithm $\mathcal{V}$ involves using VTAA from Theorem~\ref{thm:efficient-vtaa-vtae} to $\mathcal{A}$. The resulting output is a state that has performed $f(H)\ket{\psi}$ conditioned on the flag register being in $\ket{1}_F$. Subsequently, we apply the unitary $(\mathcal{A}')^\dag$ that erases the ancillary states. The final algorithm results in the following transformation
 
\begin{equation}
\label{eq:transformation-vtaa-qls}
\mathcal{V}\ket{\psi}_I\ket{0}_{CFPQ}\mapsto \dfrac{f(H)\ket{\psi}_I}{\nrm{f(H)\ket{\psi}_I}}\ket{0}_{CFPQ},
\end{equation}
such that $\nrm{\dfrac{f(H)}{\nrm{f(H)\ket{\psi}}}-\dfrac{H^{-1}}{\nrm{H^{-1}\ket{\psi}}}}\leq 	\mathcal{O}\left(\eps\right)$. We can then discard ancilla registers $C, F, \mathrm{and}~P$. So the transformation in the space $I\otimes Q$ is
$$\ket{\psi}_I\ket{0}_{Q}\mapsto \dfrac{f(H)\ket{\psi}_I}{\nrm{f(H)\ket{\psi}_I}}\ket{0}_{Q}.$$
%We can still consider this as a $\left(1,a+\bigO{\log(\kappa\log\frac{\kappa}{\eps})},\mathcal{O}(\eps)\right)$-blocking encoding of $H^{-1}/\nrm{H^{-1}\ket{\psi}}$, however, here the block encoding is trivial in the sense that it is encoded in the entire unitary and not in some block.

The correctness of this algorithm is similar to that of Childs, Kothari and Somma \cite{childs2015quantum}. Let the input quantum state $\ket{\psi}=\sum_k c_k\ket{\lambda_k}$, where $\ket{\lambda_k}'\mathrm{s}$ are the eigenstates of $H$. Let us consider an eigenstate of $H$, say $\ket{\lambda}$ with eigenvalue $\lambda\in [-1,1]$ such that $2^{-j}<|\lambda|<2^{1-j}$ for $1\leq j \leq m$. Such a $j$ exists because $1/\kappa\leq |\lambda|\leq 1$. For such a $\lambda$, applying $\mathcal{A}_{j-1}\dots\mathcal{A}_1$ to the state $\ket{\lambda}_I\ket{0}_{CFPQ}$, does nothing but modify the ancilla registers $P_{j-1},\dots,P_1$ due to the output of GPE. This is because the precision of GPE for any of $\mathcal{A}_k$ such that $1\leq k\leq j-1$ is greater than $2^{-j}$ and the register $C_k$ for $1\leq k\leq j-1$ is always in $\ket{0}$.

When $\mathcal{A}_j$ is applied, however, the output of GPE is in a superposition of $\ket{0}$ and $\ket{1}$ on $C_j$, as $2^{-j}<|\lambda|<2^{1-j}$. So in the part of the resulting state where register $C$ is not in $\ket{0}^{\otimes m}$, step 2 of $\mathcal{A}_j$ is implemented and $W(2^{-j},\epsilon')$ is applied to $I\otimes F\otimes Q$. The computation stops on this part of the resulting state as register $C$ is non-zero. On the other hand, on the part where register $C$ is in $\ket{0}^{\otimes m}$, the computation continues. Applying $\A_{j+1}$ to this part, first results in $\ket{1}$ in $C_{j+1}$ with a very high probability as $|\lambda|>2^{-j}$. Applying step 2 of $\mathcal{A}_{j+1}$ again implements $W(2^{-j-1},\epsilon')$ on $I\otimes F\otimes Q$. Since the resulting state has no overlap with $\ket{0}_C$, $\mathcal{A}_{j+2}\dots \mathcal{A}_m$ has no effect.

We observe that actually for any $2^{-j}<|\lambda|<2^{1-j}$, only $\mathcal{A}_{j}$ and $\mathcal{A}_{j+1}$ implements $H^{-1}$ through the unitary $W$ defined in Corollary \ref{cor:block-encoding-inverse-patch}. The requirements of Corollary \ref{cor:block-encoding-inverse-patch} are satisfied as $\lambda$ lies between $[-1,2^{-j}]\bigcup[2^{-j},1]$ and also between $[-1,2^{-j-1}]\bigcup[2^{-j-1},1]$.

By linearity on $\ket{\psi}=\sum_k c_k\ket{\lambda_k}$, the algorithm $\mathcal{A}$ implements $H^{-1}/{\alpha_{\mathrm{max}}}$ on register $I$ conditioned on the flag register being in $\ket{1}_F$. Next VTAA is applied to the resulting state and following that $(\mathcal{A}')^\dag$ is used to erase the ancilla registers and output the state in \eqref{eq:transformation-vtaa-qls}. For more details, readers can refer to \cite{childs2015quantum}. 

Next, we analyse the complexity of this algorithm. Note that the complexity of $\mathcal{V}$ is the same order as the cost of applying VTAA to algorithm $\mathcal{A}$ as the cost of running algorithm $(\mathcal{A}')^\dag$ is at most twice that of $\mathcal{A}$. So the contribution of $(\mathcal{A}')^\dag$ to the overall complexity can be ignored.

To estimate the cost of implementing each algorithm $\mathcal{A}_j$ we first observe that the cost of implementing GPE with precision $2^{-j}$ and error probability $\epsilon'=\eps/(m\alpha_{\mathrm{max}})$ is 
$$\mathcal{O}\left(\alpha 2^j\log \left(1/\epsilon'\right)(a+T_U)\right),$$
up to additive $\log$ factors. 
The cost of implementing $W(2^{-j},\epsilon')$ is given by Corollary~\ref{cor:block-encoding-inverse-patch} as 
$$\mathcal{O}\left(\alpha 2^j\log^2\dfrac{2^j}{\epsilon'}(a+T_U)\right).$$
So the time required to implement $\mathcal{A}_j$ is 
$$\mathcal{O}\left(\alpha 2^j\log^2\dfrac{2^j}{\epsilon'}(a+T_U)\right).$$
This implies that the time $t_j$ required to implement $\mathcal{A}_j\dots\mathcal{A}_1$ is also 
$$\mathcal{O}\left(\alpha 2^j\log^2\dfrac{\kappa}{\epsilon'}(a+T_U)\right).$$
Also, $T_{\mathrm{max}}$, the time required to execute $\mathcal{A}_m\dots\mathcal{A}_1$ is
\begin{equation}
\label{eq:t-max-vtaa}
T_{\mathrm{max}}=\mathcal{O}\left(\alpha \kappa\log^2\dfrac{\kappa}{\epsilon'}(a+T_U)\right)=\mathcal{O}\left(\alpha \kappa (a+T_U)\log^2\left(\dfrac{\kappa\log\kappa}{\eps}\right)\right).
\end{equation}
Now in order to upper bound the cost of applying VTAA to the algorithm $\mathcal{A}$, we need to now upper bound the probability that $\mathcal{A}$ stops at the $j$-th step. This is given by $p_j=\nrm{\Pi_{C_j}\mathcal{A}_j\dots \mathcal{A}_1\ket{\psi}_I\ket{0}_{CFPQ}}^2$,\footnote{$p_j$ is called $p_{\mathrm{stop}=t_j}$ in Section \ref{sec:vtae}.} where $\Pi_{C_j}$ denotes the projector on to $\ket{1}_{C_j}$. Then we can calculate the $l_2$-averaged stopping time of $\mathcal{A}$, $\nrm{T}_2$ as
\begin{align}
\nrm{T}_2^2&=\sum_j p_j t_j^2\nonumber\\
           &=\sum_j \nrm{\Pi_{C_j}\mathcal{A}_j\dots \mathcal{A}_1\ket{\psi}_I\ket{0}_{CFPQ}}^2t_j^2\nonumber\\
           &=\sum_k |c_k|^2\sum_j \left(\nrm{\Pi_{C_j}\mathcal{A}_j\dots \mathcal{A}_1\ket{\lambda_k}_I\ket{0}_{CFPQ}}^2t_j^2\right)\nonumber\\
           &=\mathcal{O}\left(\alpha^2 (a+T_U)^2 \sum_k \dfrac{|c_k|^2}{\lambda_k^2}\log^4\dfrac{1}{\lambda_k\epsilon'} \right)\nonumber\\
           \label{eq:l2-stopping-time} 
\implies \nrm{T}_2&\leq \alpha(a+T_U)\log^2\left(\dfrac{\kappa\log\kappa}{\eps}\right)\sqrt{\sum_k\dfrac{|c_k|^2}{\lambda_k^2}}.         
\end{align} 
The final thing that we need for calculating the final complexity of VTAA applied to $\mathcal{A}$ is the success probability, $p_{\mathrm{succ}}$ which can be written as 
\begin{align}
\sqrt{p_{\mathrm{succ}}}&=\nrm{\Pi_F \dfrac{H^{-1}}{\alpha_{\mathrm{max}}}\ket{\psi}_I\ket{\Phi}_{CFPQ}}+\mathcal{O}\left(m\epsilon'\right)\nonumber \\
                        &=\dfrac{1}{\alpha_{\mathrm{max}}}\left(\sum_k\dfrac{|c_k|^2}{\lambda^2_k}\right)^{1/2}+\mathcal{O}\left(\dfrac{\eps}{\alpha_{\mathrm{max}}}\right)\nonumber\\
\label{eq:succ-prob-vtaa}
                        &=\Omega\left(\dfrac{1}{\kappa}\sqrt{\sum_k\dfrac{|c_k|^2}{\lambda^2_k}}\right).
\end{align}
So the final complexity of applying VTAA to algorithm $\mathcal{A}$ is given by Theorem~\ref{thm:efficient-vtaa-vtae}. Thus the overall cost is given by (neglecting constants):
\begin{align*}
&T_{\max}+T_{\psi}+\frac{\left(\nrm{T}_2+T_\psi\right)\log(T'_{\max})}{\sqrt{p_\mathrm{succ}}}\nonumber\\
=&\bigO{\alpha\kappa\log^2\left(\dfrac{\kappa\log\kappa}{\eps}\right)(a+T_U)+\kappa\left(\alpha(a+T_U)\log^2\left(\dfrac{\kappa\log\kappa}{\eps}\right)+ T_\psi\right)\log\left(\kappa\right)}
\nonumber
\\
%\label{eq:final-complexity-vtaa-qls}
=&\bigO{\kappa\left(\alpha \big(T_U+a\big)\log^2\left(\dfrac{\kappa}{\eps}\right)+ T_\psi\right)\log(\kappa)}.
\end{align*}
\end{proof}
Next we show that in the scenario where the state $\ket{\psi}$ does not belong entirely to the range of $H$, i.e. $\nrm{\Pi{\mathrm{col}(H)}\ket{\psi}}<1$, we can prepare the state $H^{+}\ket{\psi}/\nrm{H^+\ket{\psi}}$. We only assume that a lower bound for $\nrm{\Pi{\mathrm{col}(H)}\ket{\psi}}$ is known. 

\begin{corollary}[Complexity of pseudoinverse state preparation]\label{cor:pseudoinv-vtaa}%\shnote{cor:pseudoinv-vtaa}
Let $\kappa\geq 2$, and $H$ be an $N\times N$ Hermitian matrix such that the non-zero eigenvalues of $H$ lie in the range $[-1,-1/\kappa]\bigcup [1/\kappa,1]$. Suppose that for $\delta = \littleo{\eps/(\kappa^2\log^3\frac{\kappa}{\eps})}$we have a unitary $U$ that is a $(\alpha,a,\delta)$-block-encoding of $H$ that can be implemented using $T_U$ elementary gates.  Also suppose that we can prepare a state $\ket{\psi}$ in time $T_{\psi}$ such that $\nrm{\Pi_{\mathrm{col}(H)}\ket{\psi}}\geq \sqrt{\gamma}$. Then there exists a variable time quantum algorithm that outputs a state that is $\eps$-close to
$H^{+}\ket{\psi}/\nrm{H^{+}\ket{\psi}}$ at a cost 
$$\bigO{\kappa\left(\alpha\big(T_U+a\big)\log^2\left(\dfrac{\kappa}{\eps}\right)+ T_{\psi}\right)\frac{\log\left(\kappa\right)}{\sqrt{\gamma}}}.$$
%$$\mathcal{O}\left(\alpha \kappa (a+T_U)\log^2\left(\dfrac{\kappa\log\kappa}{\eps}\right)\log\left(\alpha\kappa\log^2\left(\dfrac{\kappa\log\kappa}{\eps}\right)(a+T_U)\right)\right).$$
\end{corollary}
\begin{proof}
The result follows similarly to Theorem \ref{thm:qls-vtaa} after decreasing $p_{\mathrm{succ}}$ by a factor of $\gamma$.
\end{proof}

Often, in several applications the norm of the output of the QLS problem needs to be estimated. In such cases, one needs to replace amplitude amplification with amplitude estimation. We shall use the variable time amplitude estimation algorithm (VTAE) defined in Theorem~\ref{thm:efficient-vtaa-vtae} in order to estimate the norm of the output of QLS which gives us an improved dependence on $\kappa$ as compared to ordinary amplitude estimation. In order to implement this, we convert the QLS algorithm to a variable-time algorithm in the same way as the case of applying VTAA. Then we have the following corollary.
\begin{corollary}[Complexity of pseudoinverse state preparation and its amplitude estimation]\label{cor:qls-vtae} %\shnote{cor:qls-vtae}
Let $\eps>0$. Then under the same assumptions as in Corollary~\ref{cor:pseudoinv-vtaa}, there exists a variable time quantum algorithm that outputs a number $\Gamma$ such that 
$$1-\eps\leq \dfrac{\Gamma}{\nrm{H^{+}\ket{\psi}}}\leq 1+\eps,$$ 
at a cost
$$\bigO{\dfrac{\kappa}{\eps} \left(\alpha\big(T_U+a\big)\log^2\left(\dfrac{\kappa}{\eps}\right)+ T_\psi\right)\frac{\log^3\left(\kappa\right)}{\sqrt{\gamma}}\log\left(\dfrac{\log\left(\kappa\right)}{\delta}\right)},$$
with success probability at least $1-\delta$.
\end{corollary}
\begin{proof}
The framework is the same as that in Theorem~\ref{thm:qls-vtaa}, except instead of VTAA we use VTAE algorithm defined in Theorem~\ref{thm:efficient-vtaa-vtae} to obtain $\Gamma$. Thus, the quantities $T_{\mathrm{max}}$, $T'_{\mathrm{max}}$, $\nrm{T}_2$ and $\sqrt{p_{\mathrm{succ}}}$ are the same as in \eqref{eq:t-max-vtaa}, \eqref{eq:l2-stopping-time} and \eqref{eq:succ-prob-vtaa}, except $p_{\mathrm{succ}}$ is decreased by a factor of $\gamma$. Thus the overall complexity is given by Theorem~\ref{thm:efficient-vtaa-vtae} which is
\begin{equation}
\label{eq:final-complexity-vtae}
\qquad=\bigO{\dfrac{\kappa}{\eps} \left(\alpha\big(T_U+a\big)\log^2\left(\dfrac{\kappa}{\eps}\right)+ T_\psi\right)\frac{\log^3\left(\kappa\right)}{\sqrt{\gamma}}\log\left(\dfrac{\log\left(\kappa\right)}{\delta}\right)}.
\end{equation} 
\end{proof}

%%%%%%%%%%%%%%%%%%%%%%%%%%%%%%%%%%%%%%%%%%%%%%%%%%%%%%%%%%%%%%%%%%%%%%%%%%%%%%%%%%%%%%%%%%
Observe that VTAA or VTAE algorithms can be applied to a variable-time version of the algorithm that implements a block encoding of $H^{-c}$, for any $c>0$. Consider the quantum algorithm to implement a block encoding of $H^{-c}$ as in Lemma~\ref{lem:negative-power-restated}.

In order to amplify the amplitude of the output state, we use VTAA by converting this to a variable-stopping-time algorithm. As seen before, we need to apply this procedure in certain patches of the overall domain of $H$. For this we use Corollary~\ref{cor:block-encoding-inverse-patch} and simply replace the value of $\delta$ there with
$\delta=\littleo{\eps/\left(\kappa^{1+c}(1+c)\log^3\frac{\kappa^{1+c}}{\eps}\right)}$ and $\alpha_{\mathrm{max}}=\bigO{\kappa^c}$. So now $W(\lambda,\eps)$ implements the following transformation
\begin{equation}
\label{eq:transformation-patch-negative-power}
W(\lambda,\eps)\ket{0}_F\ket{0}_Q\ket{\psi}_I=\dfrac{1}{\alpha_\mathrm{max}}\ket{1}_F\ket{0}_Q f(H)\ket{\psi}_I +\ket{0}_F\ket{\psi^\perp}_{QI},
\end{equation}
where $\alpha_\mathrm{max}=\mathcal{O}(\kappa^c)$ and $\nrm{f(H)-H^{-c}}\leq \eps$ while the rest of the parameters are the same as Corollary~\ref{cor:block-encoding-inverse-patch}. By using Lemma~\ref{lem:negative-power-restated} we get the cost of implementing $W(\lambda,\eps)$ is 
$$\bigO{\frac{\alpha}{\lambda}\big(T_U+a\big)(1+c)\log^2\!\left(\frac{\kappa^{1+c}}{\eps}\right)}.$$

%$$\bigO{\alpha\lambda^{-1}\mathrm{max}(1,c)\log\left(\frac{\lambda^{-c}}{\eps}\right)(a+T_U)+\lambda^{-1}\mathrm{max}(1,c)\log^2\frac{\lambda^{-1-c}\mathrm{max}(1,c)}{\eps}}.$$

The variable-stopping-time algorithm can be defined $\mathcal{A}=\mathcal{A}_m...\mathcal{A}_1$ for $m=\ceil{\log_2\kappa}+1$. Each $\mathcal{A}_j$ can be defined in the same way as for $H^{-1}$. So we can define a variable-time quantum algorithm, similar to Theorem~\ref{thm:qls-vtaa}, to implement $H^{-c}$. 

\begin{restatable}{theorem}{negpowvtaa}\emph{(Variable-time quantum algorithm for implementing negative powers)}% of Hamiltonian]
\label{thm:negative-power-vtaa}%\shnote{thm:negative-power-vtaa}
~Let $\kappa\geq 2$, $c\in (0,\infty)$, $q=\mathrm{max}(1,c)$, and $H$ be an $N\times N$ Hermitian matrix such that the eigenvalues of $H$ lie in the range $[-1,-1/\kappa]\bigcup[1/\kappa,1]$. Suppose that for $\delta =\littleo{\eps/\left(\kappa^q q \log^3\frac{\kappa^q}{\eps}\right)}$ we have a unitary $U$ that is a $(\alpha,a,\delta)$-block-encoding of $H$ which can be implemented using $T_U$ elementary gates. Also suppose that we can prepare an input state $\ket{\psi}$ that is spanned by the eigenvectors of $H$ in time $T_\psi$. Then there exists a variable time quantum algorithm that outputs a state that is ${\eps}$-close to $H^{-c}\ket{\psi}/\nrm{H^{-c}\ket{\psi}}$ with a cost of 
$$\mathcal{O}\left(\left(\alpha  \kappa^{q} \big(T_U+a\big)q\log^2\left(\dfrac{\kappa^{q}}{\eps}\right)+\kappa^{c}T_\psi\right)\log\left(\kappa\right)\right).$$

Also, there exists a variable time quantum algorithm that outputs a number $\Gamma$ such that 
$$1-\eps\leq \dfrac{\Gamma}{\nrm{H^{-c}\ket{\psi}}}\leq 1+\eps,$$ 
at a cost
$$\mathcal{O}\left(\dfrac{1}{\eps} \left(\alpha \kappa^q  \big(T_U+a\big)q\log^2\left(\dfrac{\kappa^q}{\eps}\right)+\kappa^c T_\psi\right)\log^3\left( \kappa\right)\log\left(\dfrac{\log\left( \kappa\right)}{\delta}\right)\right),$$
with success probability at least $1-\delta$.
\end{restatable}
\begin{proof}
Refer to Sec.~\ref{app:negpow-vtaa} of the Appendix.
\end{proof}
%%%%%%%%%%%%%%%%%%%%%%%%%%%%%%%%%%%%%%%%%%%%%%%%%%%%%%%%%%%%%%%%%%%%%%%%%%%%%

Wossnig, Zhao and Prakash \cite{wossnig2018quantum} introduced a new quantum linear system solver based on decomposing $A$ into a product of isometries and using a Szegedy walk to perform singular value estimation. In the setting where $H$ is given by a data structure as in Theorem~\ref{theorem:data_structure}, this decomposition is generic, and both isometries can be implemented efficiently given the data structure storing $H$. The complexity of this algorithm has a better dependence on the sparsity of $H$ as compared to previous algorithms for solving quantum linear systems. Thus the algorithm of \cite{wossnig2018quantum} provides a polynomial advantage in the scenario where $H$ is non-sparse. However this algorithm has a quadratic dependence on the condition number of $H$ and a polynomial dependence on the precision of the output state. As an application of Theorem~\ref{thm:qls-vtaa}, we give a new quantum linear system solver in this setting, with an exponentially better dependence on precision and a linear dependence on the condition number. 

We have a $N\times N$ Hermitian matrix $A$. In this setting either (i) $A$ is stored in the quantum-accessible data structure defined in Theorem~\ref{theorem:data_structure} or (ii) given some $p\in [0,1]$, $A^{(p)}$ and $A^{(1-p)}$ are stored in quantum-accessible data structures, as was considered in \cite{kerenidis:quantumgraddescent}. For the QLS problem and its subsequent applications, it may be the case that $A$ is not Hermitian. In such a case either we store (i) $A$ and $A^\dagger$ in the quantum-accessible data structure or (ii) $A^{(p)}$ and $(A^{(1-p)})^\dagger$ are stored in quantum-accessible data structures so that Lemma~\ref{lem:kp} is applicable. So henceforth it suffices to consider that $A$ is Hermitian.

\begin{theorem}[Quantum Linear System solver with data structure]\label{thm:QLS-data-structure}%\snote{(thm:QLS-data-structure)}
Let $\eps\in (0,1/2)$, suppose that $\nrm{A}\leq 1$, $\nrm{A^{-1}}\leq \kappa$, and either (1) $A$ is stored in a quantum-accessible data structure, in which case, let $\mu=\nrm{A}_F$; or (2) for some $p\in [0,1]$, $A^{(p)}$ and $A^{(1-p)}$ are stored in quantum-accessible data structures, in which case, let $\mu=\mu_p(A)$. Also assume that there is a unitary $U$ which acts on $\polylog(MN/\eps)$ qubits and prepares the state $\ket{b}$ with complexity $T_b$. Then

\begin{itemize}
\item[(i)] The QLS problem can be solved in time $\bigOt{\kappa\left(\mu+T_b\right)\polylog(MN/\eps)}$.

\item[(ii)] If $\eps\in (0,1)$, then an $\eps$-multiplicative approximation of $\nrm{A^+\ket{b}}$ can be obtained in time $\bigOt{\dfrac{\kappa}{\eps} \left(\mu+T_b\right)\polylog(MN)}$
\end{itemize}
\end{theorem}
\begin{proof}
For (i), by Lemma~\ref{lem:kp} and Theorem~\ref{thm:qls-vtaa} we can solve QLS with complexity 
$$\bigO{\left(\mu \kappa\log^2\left(\dfrac{\kappa}{\eps}\right) \polylog(MN/\eps)+\kappa T_b\right)\log\left(\kappa\right)}.$$
As shown by Corollary~\ref{cor:qls-vtae} using VTAE we can estimate $\nrm{A^+\ket{b}}$ with the stated complexity. 
\end{proof}

%\begin{corollary}
%\label{cor:norm_QLS_datastructure}
%Under the same assumptions as Theorem \ref{thm:QLS-data-structure}, there exists a quantum algorithm that approximates $\nrm{A^+\ket{b}}$ in time $\bigOt{\dfrac{\kappa}{\eps} \left(\mu+T_b\right)\polylog(MN)}$. 
%\end{corollary}
%\begin{proof}
%This follows from Corollary \ref{cor:qls-vtae}. 
%\end{proof}

Note that in the scenario where the vector $\overrightarrow{b}=(b_1,\dots,b_N)^T$, is also stored in a quantum-accessible data structure, then from Theorem \ref{theorem:data_structure} we can prepare the state $\ket{b}=\sum_i b_i\ket{i}/\nrm{\overrightarrow{b}}$ in time $T_b=\bigO{\polylog(N/\eps)}$. Thus the complexity of solving (i) in Theorem \ref{thm:QLS-data-structure} in that case, is 
$$\bigOt{\kappa\mu\polylog(MN/\eps)},$$
while that of (ii) is $\bigOt{\dfrac{\kappa\mu}{\eps}\polylog(MN)}$.

\section{Applications}\label{sec:app}

In this section, we apply the QLS algorithm of Section \ref{sec:qls} to solve the least squares problem, which is used in several machine learning applications. We present improved quantum algorithms for the weighted least squares problem (Section \ref{sec:WLS}) and new quantum algorithms for the generalized least squares problem (Section \ref{sec:GLS}). Finally, we apply the QLS solver to design new quantum algorithms for estimating electrical network quantities (Section \ref{sec:electric}).

\subsection{Least squares}\label{sec:least-squares}

The problem of \emph{ordinary least squares} is the following. 
Given data points $\{(\vec{x}^{(i)},y^{(i)})\}_{i=1}^M$ for $\vec{x}^{(1)},\dots,\vec{x}^{(M)}\in \mathbb{R}^N$ and $y^{(1)},\dots,y^{(M)}\in\mathbb{R}$, find $\vec{\beta}\in \mathbb{R}^N$ that minimizes:
\begin{equation}
\sum_{i=1}^M(y^{(i)}-\vec{\beta}^T\vec{x}^{(i)})^2.\label{eq:ols}
\end{equation}
The motivation for this task is the assumption that the samples $(\vec{x}^{(i)},y^{(i)})$ are obtained from some process such that for every  $i$, $y^{(i)}$ depends linearly on $\vec{x}^{(i)}$, up to some random noise, so $y^{(i)}$ is drawn from a random variable $\vec{\beta}^T\vec{x}^{(i)}+E_i$, where $E_i$ is a random variable with mean 0, for example, a Gaussian. The vector $\vec{\beta}$ that minimizes \eqref{eq:ols} represents the underlying linear function. We assume $M\geq N$ so that it is feasible to recover this linear function. 

In particular, if $X\in \mathbb{R}^{M\times N}$ is the matrix with $\vec{x}^{(i)}$ as its $i$-th row, for each $i$, and $\vec{y}\in\mathbb{R}^M$ has $y^{(i)}$ as its $i$-th entry, assuming $X^TX$ is invertible, the optimal $\vec{\beta}$ satisfies:
$$\vec{\beta}=(X^TX)^{-1}X^T\vec{y}.$$
The assumption that $X^TX$ is invertible, or equivalently, that $X$ has rank $N$, is very reasonable, and is generally used in least squares algorithms. This is because $X^TX \in \mathbb{R}^{N\times N}$ is a sum of $M\geq N$ terms, and so it is unlikely to have rank less than $N$. 

We can generalize this task to settings in which certain samples are thought to be of higher quality than others, for example, because the random variables $E_i$ are not identical. We express such a belief by assigning a positive weight $w_i$ to each sample, and minimizing
\begin{equation}
\sum_{i=1}^Mw_i(y^{(i)}-\vec{\beta}^T\vec{x}^{(i)})^2.\label{eq:wls2}
\end{equation}
If $W\in\mathbb{R}^{M\times M}$ denotes the diagonal matrix in which $w_i$ appears in the $i$-th diagonal entry, the vector $\vec{\beta}$ that minimizes \eqref{eq:wls2} is given by:
\begin{equation}
\vec{\beta}=(X^TWX)^{-1}X^TW\vec{y},\label{eq:beta_linear_systems}
\end{equation}
under the justified assumption that $X^TWX$ is invertible. 
Finding $\vec{\beta}$ given $X$, $W$ and $\vec{y}$ is the problem of \emph{weighted least squares}.

We can further generalize to settings in which the random variables $E_i$ for sample $i$ are correlated. In the problem of \emph{generalized least squares}, the presumed correlations in error between pairs of samples are given in a non-singular covariance matrix $\Omega$. We then want to find the vector $\vec{\beta}$ that minimizes
\begin{equation}
\sum_{i,j=1}^M\Omega_{i,j}^{-1}(y^{(i)}-\vec{\beta}^T\vec{x}^{(i)})(y^{(j)}-\vec{\beta}^T\vec{x}^{(j)}).%\label{eq:gls}
\end{equation}
As long as $X^T\Omega^{-1} X$ is invertible, this minimizing vector is given by 
$$\vec{\beta}=(X^T\Omega^{-1} X)^{-1}X^T\Omega^{-1} \vec{y}.$$

In this section, we will consider solving \emph{quantum} versions of these problems. Specifically, a \emph{quantum WLS solver} is given access to $\vec{y}\in\mathbb{R}^M$, $X\in \mathbb{R}^{M\times N}$, and positive weights $w_1,\dots,w_M$, in some specified manner, and outputs a quantum state 
$$(X^TWX)^{-1}X^TW\ket{y}/\nrm{(X^TWX)^{-1}X^TW\ket{y}},$$
up to some specified error $\eps$. 

Similarly, a \emph{quantum GLS solver} is given access to $\vec{y}\in\mathbb{R}^M$, $X\in \mathbb{R}^{M\times N}$, and a positive definite $\Omega\in\mathbb{R}^{M\times M}$, in some specified manner, and outputs a quantum state 
$$(X^T\Omega^{-1}X)^{-1}X^T\Omega^{-1}\ket{y}/\nrm{(X^T\Omega^{-1}X)^{-1}X^T\Omega^{-1}\ket{y}},$$
up to some specified error $\eps$.

\subsubsection{Weighted least squares}\label{sec:WLS}

In this section we describe a quantum algorithm for the weighted least squares problem using our new quantum linear system solver, before considering generalized least squares in the next section. In particular, letting $\overline{SS}_{\mathrm{res}}^W$ be the normalized weighted sum of squares residual (defined shortly),
we prove the following:
\begin{theorem}[Quantum WLS solver using data structure input]\label{thm:WLS}
Let $A=\sqrt{W}X$ such that $\nrm{A^{+}}\leq \kappa_A$. %$\nrm{A}\leq 1$ and 
Suppose $\sqrt{W}\vec{y}$ is stored in a quantum-accessible data structure, and either (1) $A$ is stored in a quantum-accessible data structure, in which case, let $\mu(A)=\nrm{A}_F$; or (2) for some $p\in[0,1]$, $A^{(p)}$ and $A^{(1-p)}$ are stored in quantum-accessible data structures, in which case, let $\mu(A)=\mu_p(A)$. Finally, suppose the data points satisfy $\overline{SS}_{\mathrm{res}}^W\leq \eta$. Then we can implement a quantum WLS solver with error $\eps$ in complexity:
%Let $\ket{\beta}=(X^TWX)^{-1}X^TW\vec{y}$. Then for any $\eps\in(0,1/2)$, there is a quantum algorithm that outputs a state that is $\eps$-close to $\ket{\beta}/\nrm{\ket{\beta}}$ in complexity:
$$\bigOt{\frac{\kappa_A\mu(A)}{\sqrt{1-\eta}}\polylog\left(MN/\eps\right)}.$$
\end{theorem}

Our weighted least squares algorithm improves over the previous best quantum algorithm for this problem, due to \cite{kerenidis:quantumgraddescent}, which has complexity $\bigO{\frac{1}{\eps}\kappa_A^6\mu(A)\log^3\frac{\kappa_A}{\eps}\polylog(MN)}$ (assuming $\nrm{A}=1$ and $\eta$ is bounded by a constant $<1$). Compared to this previous result, our algorithm has an exponential improvement in the dependence on $\eps$, and a 6th power improvement in the dependence on $\kappa_A$. Before proving Theorem \ref{thm:WLS}, we first give a high-level overview of the algorithm.

Let $\ket{y}=\sum_{i=1}^My_i\ket{i}/\nrm{\vec{y}}$. As in \cite{kerenidis:quantumgraddescent}, our algorithm works by first constructing the state $\ket{b}=\sqrt{W}\ket{y}/\nrm{\sqrt{W}\ket{y}}$, and then applying $A^+ = (\sqrt{W}X)^+$. Given a block-encoding of $A$, we can use Corollary \ref{cor:pseudoinv-vtaa} to obtain the state $A^+\ket{b}/\nrm{A^+\ket{b}}$. However, in general, $\ket{b}$ will not be in the rowspace of $A^+$, so $A^+\ket{b}$ might be much smaller than $\sigma_{\min}(A^+)=\nrm{A}^{-1}$. However, as long as the data is not too far from linear --- that is, the fit is not too bad --- the overlap of $\ket{b}$ with $\mathrm{row}(A^+)=\mathrm{col}(A)$ will be high, and so $\nrm{A^+\ket{b}}$ won't be much smaller than $\nrm{A}^{-1}$. Before proving the main theorem of this section, we relate the size of $\Pi_{\mathrm{col}(A)}\ket{b}$ to the quality of the fit.

Define the \emph{weighted sum of squared residuals} with respect to weights $W$ by 
$$SS_{\mathrm{res}}^W = \nrm{(I-\Pi_{\mathrm{col}(A)})\sqrt{W}\vec{y}}^2.$$
This measures the sum of squared errors --- i.e.\ discrepancies between the observed and predicted data points --- weighted by $W$. To make sense of this value, we can define the \emph{normalized} weighted sum of squared residuals:
$$\overline{SS}_{\mathrm{res}}^W = \frac{\nrm{(I-\Pi_{\mathrm{col}(A)})\sqrt{W}\vec{y}}^2}{\nrm{\sqrt{W}\vec{y}}^2} = \frac{\nrm{(I-\Pi_{\mathrm{col}(A)})\ket{b}}^2}{\nrm{\ket{b}}^2} = 1-\nrm{\Pi_{\mathrm{col}(A)}\ket{b}}^2.$$
It's reasonable to assume that $\overline{SS}_{\mathrm{res}}^W$ is not too small, because otherwise, the data is very poorly fit by a linear function. In particular, if $\overline{SS}_{\mathrm{res}}^W\geq \eta$, then $R^2\leq 1-\eta$, where $R^2$ is the coefficient of determination, commonly used to measure the goodness of the fit. 

\paragraph{Proof of Thereom \ref{thm:WLS}:}
We now prove our main theorem of the section.
%The proof will be broken into several lemmas. 
Let $\delta=o(\eps/(\kappa_A^2\log^3\frac{\kappa_A}{\eps}))$.
%First, we have the following fact, which follows directly from Lemma~\ref{lem:kp}:
%\begin{fact}\label{fact:WLS-A}%\snote{fact:WLS-A}
%Let $A=\sqrt{W}X$, and suppose either: (1) $A$ is stored in a quantum-accessible data structure, in which case, let $\mu(A)=\nrm{A}_F$; or (2) for some $p\in [0,1]$, $A^{(p)}$ and $A^{(1-p)}$ are stored in quantum-accessible data structures, in which case let $\mu(A)=\mu_p(A)$. Then t
%There is an $(\bigO{\mu(A)},\lceil 
%\log(N+M+1)\rceil,\delta)$-block-encoding of $\overline{A}$ with implementation cost $\bigO{\polylog(MN/\delta)}$.
%\end{fact}
By Lemma~\ref{lem:kp} we know how to implement a $(\bigO{\mu(A)},\lceil 
	\log(N+M+1)\rceil,\delta)$-block-encoding of $\overline{A}$ with complexity $\bigO{\polylog(MN/\delta)}$.
%The final fact we need is the following, which follows from the assumption that $\sqrt{W}\vec{y}$ is stored in a quantum-accessible data structure.
%\begin{fact}\label{fact:WLS-b}
%	The state $\ket{b}$ can be generated in cost $\polylog(MN/\delta)$.
%\end{fact}
Since $\sqrt{W}\vec{y}$ is stored in a quantum-accessible data structure,
the state $\ket{b}$ can be generated in cost $\polylog(MN/\delta)$.
Using these ingredients Corollary~\ref{cor:pseudoinv-vtaa} implies that we can prepare am $\eps$-approximation of the quantum state $A^+\ket{b}/\nrm{A^+\ket{b}}$ in complexity:
\begin{equation}
\bigOt{\frac{\kappa_A\mu(A)}{\sqrt{1-\eta}}\polylog\left(MN/\eps\right)}.
%\bigO{\frac{\mu(A)\kappa_A}{\nrm{A}\sqrt{1-\eta}}\log^2\left(\frac{\kappa_A}{\eps}\right)\log\left(\frac{\mu(A)\kappa_A}{\nrm{A}\eps}\right)\polylog(MN)}.
\end{equation}
In applying Corollary \ref{cor:pseudoinv-vtaa}, we used the fact that 
$$
\nrm{\Pi_{\mathrm{col}(A)}\ket{b}}^2 = 1-\overline{SS}_{\mathrm{res}}^W\geq 1-\eta.
$$

In some applications it might not be natural to assume that we store $A$ in quantum memory. Therefore we also prove a version where $X$ and $W$ are accessed separately, as a special case of the GLS solver we prove in the next subsection.

\subsubsection{Generalized least squares}\label{sec:GLS}

In this section, we give a quantum GLS solver when the input is given in the block-encoding framework. Given block-encodings of $X$ and $\Omega$, it is straightforward to implement a block-encoding of $(X^T\Omega^{-1}X)^{-1}X^T\Omega^{-1}$ using the following: 1) Given a block-encoding of $A$, we can implement a block-encoding of $A^{-1}$; and 2) Given block-encodings of $A$ and $B$, we can implement a block-encoding of $AB$. The resulting block-encoding can then be applied to $\ket{y}$ to get a state proportional to $\vec{\beta}$, the desired output. (For a detailed analysis of this approach, see \cite{ourpaper}). While this approach is conceptually quite simple, we can get a simpler algorithm with better complexity by observing that if $A=\Omega^{-1/2}X$, then $(X^T\Omega^{-1}X)^{-1}X^T\Omega^{-1} = A^+\Omega^{-1/2}$. 

	\begin{restatable}{theorem}{GLSThmBlockGen}\emph{(Quantum GLS solver using block-encodings)}\label{thm:GLSBlockGen}
		Suppose that we have a unitary $U_y$ preparing a quantum state proportional to $\vec{y}$ in complexity $T_y$. Suppose $X\in\mathbb{R}^{M\times N}$, $\Omega\in \mathbb{R}^{M\times M}$ are such that
		$\nrm{X}\leq 1$, $\nrm{\Omega}\leq 1$ and $\Omega\succ 0$ is positive definite. Suppose that we have access to $U_X$ that is an $(\alpha_X,a_X,0)$-block-encoding of $X$ which has complexity $T_X\geq a_X$, and similarly we have access $U_\Omega$ that is an $(\alpha_\Omega,a_\Omega,0)$-block-encoding of $\Omega^{-\frac{1}{2}}$ which has complexity $T_\Omega\geq a_\Omega$. Let $A:=\Omega^{-\frac{1}{2}}X$, and suppose we have the following upper bounds: $\nrm{A^+}\leq \kappa_A$, $\nrm{\Omega^{-1}}\leq \kappa_\Omega$,
		%$\sigma^{-1}_{\min}(X)\leq \kappa_\Omega$, where $\sigma_{\min}(X)$ is the smallest non-zero singular value of $X$.
		and $\overline{SS}_\mathrm{res}^{\Omega}\leq \eta$. Then we can implement a quantum GLS-solver with error $\eps$ in complexity 
		\begin{equation*}
		\bigO{\frac{\kappa_A\log\left(\kappa_A\right)}{\sqrt{1-\eta}}\left(\left(\sqrt{\kappa_\Omega}\alpha_XT_X+\alpha_\Omega T_\Omega\right)\log^3\left(\dfrac{\kappa_A}{\eps}\right)+\sqrt{\kappa_\Omega}T_y\right)}.
		\end{equation*}
	\end{restatable}
		\begin{proof}
			The goal is to implement a unitary preparing a state proportional to $$(X^T\Omega^{-1}X)^{-1} X^T\Omega^{-1}\ket{y}=\left(\Omega^{-\frac{1}{2}}X\right)^{\!\!+} \Omega^{-\frac{1}{2}}\ket{y}.$$
			
			By Lemma~\ref{lem:block-encoding-to-state} we can implement a unitary $U_{\psi}$, that prepares a $\delta$-approximation of 
			$$\ket{\psi}:=
			\frac{\Omega^{-\frac{1}{2}}\ket{\vec{y}}}{\nrm{\Omega^{-\frac{1}{2}}\vec{y}}} \text{ with complexity }
			T_\psi:=
			\bigO{\alpha_\Omega T_\Omega \log\left(\frac{1}{\delta}\right)+\sqrt{\kappa_\Omega}T_y}.$$
			%\bigO{\alpha_\Omega(T_\Omega+T_y)}.$$
			%\bigO{\alpha_\Omega\min\left((T_\Omega+T_y), T_\Omega\log\left(\frac{1}{\delta}\right)+\frac{\sqrt{\kappa_\Omega}}{\alpha_\Omega}T_y\right)}.$$
			
			Let $:=a_X+a_\Omega+2$, by Lemma~\ref{lem:block-encoding-product-non-square-pre-amp} we can combine the block-encodings of $\Omega^{-\frac{1}{2}}$ and $X$ to implement a unitary $U_A$, that is a $(2\sqrt{\kappa_\Omega},a_A,\delta)$ block-encoding of $A$ in complexity
			$$
			T_A:=\bigO{\left(\alpha_X(T_X+a_X)+\frac{\alpha_\Omega}{\sqrt{\kappa_\Omega}}(T_\Omega+a_\Omega)\right)\log\left(\frac{\kappa_\Omega}{\delta}\right)}.
			$$

			Finally by choosing $\delta=\littleo{\eps\kappa_A^{-2}\log^{-3}(\frac{\kappa_A}{\eps})}$
			and defining $\alpha_A:=\sqrt{\kappa_\Omega}$, using Corollary~\ref{cor:pseudoinv-vtaa} we get that a quantum state proportional to $A^+\ket{\psi}$ can be prepared with $\eps$-precision in complexity 	
			\begin{align*}
				&\quad\,\,\bigO{\left(\alpha_A \kappa_A (a_A+T_A)\log^2\left(\dfrac{\kappa_A}{\eps}\right)+\kappa_A T_{\psi}\right)\frac{\log\left(\kappa_A\right)}{\sqrt{\gamma}}}\\
				&= \bigO{\frac{\kappa_A\log\left(\kappa_A\right)}{\sqrt{1-\eta}}\left(\sqrt{\kappa_\Omega}T_A\log^2\left(\dfrac{\kappa_A}{\eps}\right)+T_\psi\right)}\\
				&= \bigO{\frac{\kappa_A\log\left(\kappa_A\right)}{\sqrt{1-\eta}}\left(\left(\sqrt{\kappa_\Omega}\alpha_XT_X+\alpha_\Omega T_\Omega\right)\log^3\left(\dfrac{\kappa_A}{\eps}\right)+\sqrt{\kappa_\Omega}T_y\right)}.\qedhere
			\end{align*}
		\end{proof}
		Note that the above theorem requests $0$-error block-encoding inputs, however if the algorithm uses $T$ queries to the block-encodings, the error blows up only linearly in $T$, so if we allow a $\delta=c\eps^2/T$ initial error (for some small enough $c\in\R_+$ constant) in the block-encodings, then we do not make more than $\eps/2$ overall error.\footnote{For more details about this argument see~\cite{gilyenBlockMatrices}.}

	\begin{corollary}[Quantum WLS solver using data structure or sparse oracles -- alternate input]\label{cor:WLS-alt} %\snote{thm:WLS-alt}
		Let $W$ be a diagonal matrix such that $1\leq w_i\leq w_{\max}$ for each $i$, moreover $\nrm{X}\leq 1$.
		Let $A=\sqrt{W}X$ and suppose that $\nrm{A^{+}}\leq \kappa_A$.
		Suppose $\vec{y}$ is stored in a quantum data structure, and the diagonal entries of $W$ are stored in QROM so that we can compute $\ket{i}\mapsto \ket{i}\ket{w_i}$ in $\polylog(MN/\eps)$, as well as $w_{\max}$. Further, suppose either (1) $X$ is stored in a quantum-accessible data structure, in which case, let $\mu(X)=\nrm{X}_F$; or (2) for some $p\in[0,1]$, $X^{(p)}$ and $X^{(1-p)}$ are stored in quantum-accessible data structures, in which case, let $\mu(X)=\mu_p(X)$. Finally, suppose the data points satisfy $\overline{SS}_{\mathrm{res}}^W\leq \eta$. Then we can implement a quantum WLS solver with error $\eps$ in complexity:
		%Then for any $\eps\in (0,1/2)$, there is a quantum algorithm that outputs a state that is $\eps$-close to $\ket{\beta}/\nrm{\ket{\beta}}$ in complexity:
		$$\bigOt{\frac{\kappa_A\sqrt{w_{\max}}}{\sqrt{1-\eta}}\mu(X)\polylog(MN/\eps)}.$$
		Similarly, if we are given sparse access to $X$ which has row and column sparsity at most $s^r_X$ and $s^c_X$ respectively, and a unitary $U_y$ preparing $\ket{y}$ in complexity $T_y$, then we can implement a quantum WLS solver with error $\eps$ in complexity:
		$$\bigOt{\frac{\kappa_A\sqrt{w_{\max}}}{\sqrt{1-\eta}}\left(\sqrt{s^r_X s^c_X}+T_y\right)\polylog(MN/\eps)}.$$
	\end{corollary}

	\begin{corollary}[Quantum GLS solver using block-encodings -- alternate input]\label{cor:GLSBlockAlt}
    Suppose that $X$, $\Omega$, and $\vec{y}$ are as in Theorem~\ref{thm:GLSBlockGen}, except we have access to $U_\Omega$ that is an $(\alpha_\Omega,a_\Omega,0)$-block-encoding of $\Omega$ which has complexity $T_\Omega\geq a_\Omega$.
	Then we can implement a quantum GLS-solver with error $\eps$ in complexity 
		\begin{equation*}
		\bigOt{\frac{\kappa_A\sqrt{\kappa_\Omega}}{\sqrt{1-\eta}}\left(\alpha_XT_X+\alpha_\Omega \kappa_\Omega T_\Omega+ T_y\right)\polylog\left(\frac{1}{\eps}\right)}.
		\end{equation*}
	\end{corollary}
	\begin{proof}
		By Lemma~\ref{lem:negative-power-restated} we can implement a unitary $U_\Omega'$ that is a $(2\sqrt{\kappa_\Omega},a_\Omega+\bigO{\log(\kappa_\Omega\log\frac{1}{\delta}},\delta/4)$ block-encoding of $\Omega^{-\frac{1}{2}}$
		in complexity $\bigO{\alpha_\Omega \kappa_\Omega(a_\Omega+T_\Omega)\log^2\!\left(\frac{\kappa_\Omega}{\delta}\right)}.$ Choosing $\delta=\littleo{\mathrm{poly}\left(\frac{\kappa_A}{\eps}\right)}$
		the result follows from Theorem~\ref{thm:GLSBlockGen}.
	\end{proof}

\begin{corollary}[Quantum GLS using quantum data structure or sparse oracles]\label{cor:GLSData}
	Suppose $X$, $\Omega$, and $\vec{y}$ are as in Corollary~\ref{cor:GLSBlockAlt}, and we are given access to $X$ as in Corollary~\ref{cor:WLS-alt}, and similarly to $\Omega$. Then in case of the database input model we can implement a quantum GLS-solver with error $\eps$ in complexity 
	\begin{equation*}
	\bigOt{\frac{\kappa_A\sqrt{\kappa_\Omega}}{\sqrt{1-\eta}}\left(\mu_X+\mu_\Omega\kappa_\Omega\right)\polylog\left(MN/\eps\right)}.
	\end{equation*}
	Similarly, in case of the sparse-access input model we can implement a quantum GLS-solver with error $\eps$ in complexity 
	\begin{equation*}
	\bigOt{\frac{\kappa_A\sqrt{\kappa_\Omega}}{\sqrt{1-\eta}}\left(\sqrt{s^r_X s^c_X}+s_\Omega\kappa_\Omega\right)\polylog\left(MN/\eps\right)}.
	\end{equation*}
\end{corollary}

\anote{Note that if we would try to derive the WLS result from this complexity, we would almost get back the previous complexity -- the only extra term we get has an additional $w_{\max}$ factor.}

\subsection{Estimating electrical network quantities}\label{sec:electric}
\label{sec:electrical_networks}
Analysis of electrical networks finds widespread applications in a plethora of graph-based algorithms. For algorithms such as graph sparsification \cite{spielman2011graph}, computing maximum-flows \cite{christiano2011electrical,lee2013new} and for analyzing several classical random walk-based problems \cite{doyle1984random}, it turns out to be useful to treat the underlying graph as an electrical network.

%Electrical networks have also been studied in the context of several quantum algorithms. Belovs~\cite{belovs2013quantum} established a relationship between the problem of finding a marked node in a graph by quantum walk and the effective resistance of an electrical network. Building up on this analogy, several other quantum algorithms have been developed \cite{belovs2013time, montanaro2015quantum, moylett2017quantum}. Jeffery and Kimmel~\cite{Jeffery2017algorithmsgraph} showed that the problem of determining whether two nodes of a graph are connected is related to finding the effective resistance between them. 

In Ref.~\cite{wang2017efficient}, Wang presents two quantum algorithms for estimating certain quantities in large sparse electrical networks in the sparse-access input model: one is based on using a quantum linear systems algorithm for sparse matrices \cite{childs2015quantum} to invert the weighted signed incidence matrix, defined shortly, while the other is based on quantum walks. The estimated quantities include, among others, the power dissipated across a network, of which the effective resistance between two nodes is a special case. Wang uses the fact that these quantities can be obtained by estimating the norm of the output of a certain QLS problem. 

In this section, we give a quantum algorithm for estimating the dissipated power similar to Wang's linear-system-based algorithm, but in the block-encoding input model, replacing the QLS solver of \cite{childs2015quantum}, which Wang uses, by our QLS solver for block-encodings. In particular, rather than standard amplitude estimation, we make use of our new variable-time amplitude estimation (Corollary~\ref{cor:qls-vtae}). An immediate corollary of this is an algorithm in the sparse-access input model, which outperforms Wang's linear-system-based algorithm for all electrical networks, and in some parameter regimes, also improves on his quantum-walk-based algorithm for this problem. Additionally, our block-encoding algorithm implies the first algorithm for this problem in the quantum data structure input model. Our algorithms also apply to estimating the effective resistance, as a special case. 

It is worth noting that we can also obtain a speedup over the remaining algorithms introduced by Wang in \cite{wang2017efficient} that are based on only solving linear systems such as calculating the current across an edge and approximating voltage across two nodes. However, we do not include this analysis here. 

\noindent We begin by defining an electrical network and related quantities that shall be used subsequently. 

\paragraph{Problem setting and definitions.}
An electrical network is a weighted connected graph with the weight of each edge being the inverse of the resistance --- i.e., the conductance --- of the edge. Let $G(V,E,w)$ denote a connected graph with vertices $V$, edges $E$, and edge weights $w$. Let $N=|V|$ and $M=|E|$. We assume that  the weight of each edge $w_e$ is such that $1\leq w_e\leq w_{\mathrm{max}}$. The degree of $v$ is the number of vertices adjacent to $v$, and is denoted by $d(v)$. The maximum degree of $G$ is denoted $d=\max_{v\in V}d(v)$. As the network may be non-sparse,  $d$ can scale with the size of the network. The complexity of our quantum algorithms depend on the size of the network $N$, the maximum degree $d$, the spectral gap of the normalized Laplacian representing the network $\lambda$ (defined shortly), the precision parameter $\epsilon$, and the maximum edge weight $w_{\max}$.

Let $B_G\in\mathbb{R}^{N\times M}$ be the \emph{signed vertex-edge incidence matrix}, defined so that for each $e\in [M]$, the $e$-th column has a single $1$ and a single $-1$, in the rows corresponding to the two vertices incident to edge $e$, and 0s elsewhere; and let $W_G\in\mathbb{R}^{M\times M}$ be a diagonal matrix where the $e$-th diagonal entry represents the weight $w_e$ of edge $e$. The weighted signed vertex-edge incidence matrix is then $C_G=B_G\sqrt{W_G}$ and the graph Laplacian is $L_G=C_GC_G^T=B_GW_GB_G^T$.

$L_G$ is a positive semidefinite matrix with its minimum eigenvalue being 0 and the corresponding eigenvector being the uniform vector \cite{bollobas2013modern}. We denote the eigenvalues of $L_G$ as
\begin{equation}
\lambda_1(L_G)=0< \lambda_2(L_G)\leq....\lambda_N(L_G)\leq 2 w_{\mathrm{max}}d.\label{eq:lambda2}
\end{equation}
The weighted degree of a vertex is the sum of the weights of the edges incident to it, i.e.\ $\overline{d}_v=\sum_{e:v\in e}w_e$. Define the diagonal weighted degree matrix $D_G=\sum_{v\in V} \overline{d}_v\ket{v}\bra{v}$.
Then one can also define the normalized Laplacian of $G$ as $\mathcal{L}_G=D_G^{-1/2}L_G D_G^{-1/2}$. The spectrum of the normalized Laplacian is denoted
\[\lambda_1(\mathcal{L}_G)=0<\lambda_2(\mathcal{L}_G)\leq....\lambda_N(\mathcal{L}_G)\leq 2.\]
It is easy to show that since $\overline{d}_v\geq 1, \forall v\in V$, $\lambda_2(\mathcal{L}_G)\leq\lambda_2(L_G)$.

We now give a mathematical definition of the dissipated power of an external current applied to a network.

\begin{definition}[Dissipated power]
Given a weighted graph $G(V,E,w)$, and a \emph{current} $\vec{i}\in \mathbb{R}^M$ the \emph{dissipated power} of $\vec i$ is given by ${\cal E}(\vec{i})=\sum_{e\in E}\frac{\vec{i}(e)^2}{w_e} = \nrm{W_G^{-1/2}\vec{i}}^2$. 

An \emph{external current} $\vec{i}_{ext}\in\mathbb{R}^N$ is a real-valued function on $V$ that sums to 0. A positive value $\vec{i}_{ext}(v)$ represents current entering the network at $v$, and a negative value $\vec{i}_{ext}(v)$ represents current leaving the network at $v$. An external current $\vec{i}_{ext}$ on $G$ induces a \emph{potential} (voltage) $\vec{v}\in \mathbb{R}^N$ on the vertices of $G$, given by $\vec{v}=L_G^+\vec{i}_{ext}$. This voltage has a corresponding \emph{induced current} defined via Ohm's Law as $\vec{i}=W_GB_G^T\vec{v}$. The \emph{dissipated power of $\vec{i}_{ext}$} is defined as ${\cal E}(\vec{i})$.
\end{definition}

A well-known special case of the dissipated power is the \emph{effective resistance between $s$ and $t$} for $s,t\in V$, which is the power dissipated by the current induced by injecting a unit of current into $s$, and removing it at $t$. 

\begin{definition}[Effective resistance]
Given a weighted graph $G(V,E,w)$ and a pair of vertices $s,t\in V$, the \emph{effective resistance between $s$ and $t$} is just the dissipated power of the external current $\vec{i}_{ext}=\ket{s}-\ket{t}$. 
\end{definition}

Since the effective resistance is a special case of the dissipated power, algorithms for estimating the dissipated power can be applied to estimate the effective resistance between two nodes.

\paragraph{Algorithms for estimating dissipated power.} From \cite[Lemma 6]{wang2017efficient}, we have:
\begin{lemma}
\label{lemma:alternate_linear_systems}
Let $C_G\in\mathbb{R}^{N\times M}$ be the weighted signed vertex-edge incidence matrix of an electrical network $G(V,E,w)$. Then given an external current $\vec{i}_{ext}\in\mathbb{R}^N$ on $G$, if $\vec{i}$ denotes the induced current, we have
\[
\begin{bmatrix}
0 && C_G\\
C_G^T && 0
\end{bmatrix}^+
\left(
\begin{matrix}
\vec{i}_{ext}\\
0
\end{matrix}
\right)
=
\left(
\begin{matrix}
0\\
W_G^{-1/2}\vec{i}
\end{matrix}\right).
\]
\end{lemma} 

Thus, to estimate the dissipated power of an external current $\vec{i}_{ext}$, it suffices to estimate $\nrm{\overline{C}_G^+\ket{0}\vec{i}_{ext}}^2 = \nrm{W_G^{-1/2}\vec{i}}^2$. This gives the following:

\begin{theorem}[Estimating dissipated power]\label{thm:dissipated}
Fix $\epsilon\in(0,1)$, $w_{\max}\geq 1$, $\lambda >0$, and $d\geq 1$. Fix any $\delta$ in $o\left(\frac{\epsilon \lambda}{dw_{\max}\log^2\frac{dw_{\max}}{\epsilon\lambda}}\right)$.
For a weighted network $G(V,E,w)$, with $|V|=N$, $|E|=M$, maximum degree $d$, $1\leq w_e\leq w_{\max}$ for all $e\in E$, and $\lambda_2({\cal L}_G)\geq \lambda$; and an external current $\vec{i}_{ext}\in\mathbb{R}^N$, suppose we are given the value $\nrm{\vec{i}_{ext}}=\mathrm{poly}(N)$, a unitary $U_{\vec i_{ext}}$ preparing a quantum state proportional to $\vec{i}_{ext}$ in complexity $T_{\vec{i}_{ext}}$, and an $(\alpha,a,\delta)$-block-encoding of $C_G$ that can be implemented in complexity $T_{C_G}$. Then the dissipated power of $\vec{i}_{ext}$ can be estimated to multiplicative accuracy $\epsilon$ with success probability at least $\frac{2}{3}$ in complexity 
$$\bigO{\frac{1}{\epsilon}\sqrt{\frac{dw_{\max}}{\lambda}}\left({\alpha}\big(T_{C_G}+a\big)\log^2\frac{dw_{\max}}{\lambda\epsilon}+ T_{\vec{i}_{ext}}\right)\log^4{\frac{dw_{\max}}{\lambda}}}.$$
In particular, if $\vec{i}_{ext}=\ket{s}-\ket{t}$, then we can estimate the effective resistance between $s$ and $t$ in the given complexity, even without assuming an input oracle for state preparation.
\end{theorem}
%\begin{theorem}[Estimating dissipated power]\label{thm:dissipated}
%Fix $\epsilon\in(0,1)$, $w_{\max}\geq 1$, $\lambda >0$, and $d\geq 1$. Fix any $\delta$ in $o\left(\frac{\epsilon \lambda}{dw_{\max}\log^2\frac{dw_{\max}}{\epsilon\lambda}}\right)$.
%For a weighted network $G(V,E,w)$, with $|V|=N$, $|E|=M$, maximum degree $d$, $1\leq w_e\leq w_{\max}$ for all $e\in E$, and $\lambda_2({\cal L}_G)\geq \lambda$; and an external current $\vec{i}_{ext}\in\mathbb{R}^N$, suppose we are given the value $\nrm{\vec{i}_{ext}}=\mathrm{poly}(N)$, a unitary $U_{\vec i_{ext}}$ preparing a quantum state proportional to $\vec{i}_{ext}$ in complexity $T_{\vec{i}_{ext}}\geq \log\nrm{\vec{i}_{ext}}$, and an $(\alpha,a,\delta)$-block-encoding of $C_G$ that can be implemented in complexity $T_{C_G}$. Then the dissipated power of $\vec{i}_{ext}$ can be estimated to multiplicative accuracy $\epsilon$ with success probability at least $\frac{2}{3}$ in complexity 
%$$\bigO{\frac{1}{\epsilon}\sqrt{\frac{dw_{\max}}{\lambda}}\left({\alpha}\big(T_{C_G}+a\big)\log^2\frac{dw_{\max}}{\lambda\epsilon}+ T_{\vec{i}_{ext}}\right)\log^4{\frac{dw_{\max}}{\lambda}}}.$$
%In particular, if $\vec{i}_{ext}=\ket{s}-\ket{t}$, then we can estimate the effective resistance between $s$ and $t$ in the given complexity, even without assuming an input oracle for state preparation.
%\anote{I find the condition $T_{\vec{i}_{ext}}\geq \log\nrm{\vec{i}_{ext}}$ weird, do we really need it? Also $\eps$ and $\epsilon$ are really easy to mix, should not we use $\delta$ instead of $\eps$?}
%\end{theorem}
\begin{proof}
By Lemma \ref{lemma:alternate_linear_systems} it suffices to compute $\nrm{C_G^+\vec{i}_{ext}}^2$. %We will actually estimate $\nrm{C_G^+\vec{i}_{ext}}$ to multiplicative accuracy $1\pm\epsilon/3$. 
%If $\tilde{x}$ is an estimate of $\nrm{C_G^+\vec{i}_{ext}}$ such that 
%$$\nrm{C_G^+\vec{i}_{ext}}(1-\epsilon/3)\leq \tilde{x}\leq \nrm{C_G^+\vec{i}_{ext}} (1+\epsilon/3),$$
%then
%\begin{eqnarray*}
	%\nrm{C_G^+\vec{i}_{ext}}^2 (1-2\epsilon/3+\epsilon^2/9)\leq & \tilde{x}^2 & \leq \nrm{C_G^+\vec{i}_{ext}}^2 (1+2\epsilon/3+\epsilon^2/9)\\
	%\nrm{C_G^+\vec{i}_{ext}}^2 (1-\epsilon)\leq & \tilde{x}^2 & \leq \nrm{C_G^+\vec{i}_{ext}}^2 (1+\epsilon),
%\end{eqnarray*}
%so $\tilde{x}^2$ is an estimate of the dissipated power $\nrm{C_G^+\vec{i}_{ext}}^2$, with multiplicative error $1\pm\epsilon$. 
We will actually estimate $\nrm{C_G^+\vec{i}_{ext}}$ to $\epsilon/3$-multiplicative accuracy, 
yielding an $\epsilon$-multiplicative estimate of $\nrm{C_G^+\vec{i}_{ext}}^2$. 

We first note that for any external current $\vec{i}_{ext}$, $\vec{i}_{ext}\in \mathrm{col}(C_G)$. This is because the entries of $\vec{i}_{ext}$ must sum to 0, meaning it is orthogonal to the uniform vector. Since $\lambda_2( L_G)\geq \lambda_2({\cal L}_G)>0$, the uniform vector is the unique 0-eigenvector of $L_G$, so $\vec{i}_{ext}\in\mathrm{col}(L_G)=\mathrm{col}(C_G)$, since $L_G=C_GC_G^T$. 

By \eqref{eq:lambda2}, the condition number of $L_G$ is at most $2dw_{\max}/\lambda_2(L_G)\leq 2dw_{\max}/\lambda_2({\cal L}_G)\leq 2dw_{\max}/\lambda$, and since $L_G=C_GC_G^T$, the condition number of $C_G$ is at most $\kappa=\sqrt{2dw_{\max}/\lambda}$. Thus, the eigenvalues of $C_G/\nrm{C_G}$ lie in $[1/\kappa,1]$, and we have an $(\alpha/\nrm{C_G},a,0)$-block-encoding of $C_G/\nrm{C_G}$, so by Corollary \ref{cor:qls-vtae}, we can estimate $\nrm{C_G^+\vec{i}_{ext}}/\nrm{\vec{i}_{ext}}$ to multiplicative accuracy $\epsilon/3$ in complexity %(neglecting constants):
$$\bigO{\frac{\kappa}{\epsilon}\left(\frac{\alpha}{\nrm{C_G}}\big(T_{C_G}+a\big)\log^2\frac{\kappa}{\epsilon}+ T_{\vec{i}_{ext}}\right)\log^4\kappa}.$$
%=\frac{1}{\epsilon}\sqrt{\frac{dw_{\max}}{\lambda}}\left(\alpha(a+T_{C_G})\log^2\frac{dw_{\max}}{\lambda\epsilon}+ T_{\vec{i}_{ext}}\right)\log^4{\frac{dw_{\max}}{\lambda}}.$$
%We can scale our estimate by $\nrm{\vec{i}_{ext}}$ to get an estimate of $\nrm{C^+_G\vec{i}_{ext}}$, which is also correct to multiplicative accuracy $\epsilon/3$.
Observe that for any $e\in E$, $\sqrt{2}\leq \nrm{C_G\ket{e}}\leq \nrm{C_G}$; also $\kappa\leq\sqrt{\frac{dw_{\max}}{\lambda}}$ concluding the proof.
\end{proof}

In Ref.~\cite{wang2017efficient}, Wang considers estimating the dissipated power in an input model that assumes a constant-complexity procedure for generating a state proportional to $\vec{i}_{ext}$, and allows sparse access to $C_G$, whose sparsity is $d$, in constant complexity. Since sparse access can be used to implement a $(d,\mathrm{polylog}(MN/\delta),\delta)$-block-encoding of $C_G$ in complexity $\mathrm{polylog}(MN/\delta)$, we have the following corollary.

\begin{corollary}[Estimating dissipated power in the sparse-access model]\label{cor:dissipated-sparse}
Fix parameters as in Theorem \ref{thm:dissipated}, and assume sparse access to $C_G$, access to the value $\nrm{\vec{i}_{ext}}$, and query access to a subroutine that generates a state proportional to $\vec{i}_{ext}$. Then there is a quantum algorithm that estimates the dissipated power of $\vec{i}_{ext}$ to multiplicative accuracy $\epsilon$ with bounded error in query and gate complexity 
$$\bigOt{\frac{d^{3/2}}{\epsilon}\sqrt{\frac{w_{\max}}{\lambda}}\mathrm{polylog}(N)}.$$ 
In particular, if $\vec{i}_{ext}=\ket{s}-\ket{t}$, then we can estimate the effective resistance between $s$ and $t$ in the given complexity, even without assuming an input oracle for state preparation.
\end{corollary}
Our algorithm in the sparse-access input model compares favourably with Wang's algorithm that is also based on inverting $C_G$, which has complexity $\bigOt{\dfrac{w_{\mathrm{max}}d^2}{\lambda\epsilon}\mathrm{polylog}(N)}.$ However, Wang presents a second algorithm for estimating the dissipated power that uses quantum-walk-based techniques. Our result also improves on this second algorithm in some parameter regimes. We discuss this further at the end of this section. 

Our block-encoding result can also be applied to the case when the input is given as a quantum data structure, in which case, the value $\nrm{\vec{i}_{ext}}$ can be easily read off the root of the tree that stores the entries of $\vec{i}_{ext}$ in a quantum data structure:
\begin{corollary}[Estimating dissipated power in quantum data structure model]\label{cor:dissipated-qram1}
Fix parameters as in Theorem \ref{thm:dissipated}, and assume $\vec{i}_{ext}$ is stored in a quantum data structure and either (1) $C_G$ is stored in a quantum data structure, in which case, let $\mu(C_G)=\nrm{C_G}_F$; or (2) $C_G^{(p)}$ and $C_G^{(1-p)}$ are stored in quantum data structures, in which case, let $\mu(C_G)=\mu_p(C_G)$. Then there is a quantum algorithm that estimates the dissipated power of $\vec{i}_{ext}$ to multiplicative accuracy $\epsilon$ with bounded error in complexity 
$$\bigOt{\frac{\mu(C_G)}{\epsilon}\sqrt{\frac{dw_{\max}}{\lambda}}\mathrm{polylog}(N)}.$$
In particular, if $\vec{i}_{ext}=\ket{s}-\ket{t}$, then we can estimate the effective resistance between $s$ and $t$ in the given complexity, even without assuming an input oracle for state preparation.
\end{corollary}

In the quantum data structure model, it may be more natural to assume that the weights $W_G$ and the incidence matrix $B_G$ are stored separately. In that case, we get the following.
\begin{corollary}[Estimating dissipated power in quantum data structure model, alternative input]\label{cor:dissipated-qram2}
Fix parameters as in Theorem \ref{thm:dissipated}, and assume $\vec{i}_{ext}$ is stored in a quantum data structure, $w$ is stored in QROM so that we can compute $\ket{e}\mapsto \ket{e}\ket{w_e}$ in polylog$(M)$ complexity, and either (1) $B_G$ is stored in a quantum data structure, in which case, let $\mu(B_G)=\nrm{B_G}_F$; or (2) $B_G^{(p)}$ and $B_G^{(1-p)}$ are stored in quantum data structures, in which case, let $\mu(B_G)=\mu_p(B_G)$. Then there is a quantum algorithm that estimates the dissipated power of $\vec{i}_{ext}$ to multiplicative accuracy $\epsilon$ with bounded error in complexity 
$$\bigOt{\frac{\mu(B_G)w_{\max}}{\epsilon}\sqrt{\frac{d}{\lambda}}\mathrm{polylog}(N)}.$$
In particular, if $\vec{i}_{ext}=\ket{s}-\ket{t}$, then we can estimate the effective resistance between $s$ and $t$ in the given complexity, even without assuming an input oracle for state preparation.
\end{corollary}
\begin{proof}
Similar to the proof of Lemma \ref{lem:kp},
we can implement a $(\sqrt{w_{\max}}\mu(B_G),\mathrm{polylog}(N),\delta)$-block-encoding of $C_G=B_G\sqrt{W_G}$ in complexity $\mathrm{polylog}(N/\delta)$. Then the result follows from Theorem~\ref{thm:dissipated}. 
\end{proof}
We note that due to the specific structure of $B_G$, $\mu(B_G)$ can likely be bounded in some cases, but we leave this for future work.

\paragraph{Comparison with previous work.} In the sparse-access model (Corollary~\ref{cor:dissipated-sparse}) our complexity is
\begin{equation}\label{eq:sparseComp}
\bigOt{\frac{d^{3/2}}{\epsilon}\sqrt{\frac{w_{\max}}{\lambda}}\mathrm{polylog}(N)}.
\end{equation}
This improves upon Wang's QLS based algorithm for estimating the dissipated power, which has complexity
\begin{equation}\label{eq:WangQLS}
\bigOt{\dfrac{w_{\mathrm{max}}d^2}{\lambda\epsilon}\mathrm{polylog}(N)}.
\end{equation} Wang also gives an alternative algorithm for estimating the dissipated power based on quantum walks, which has complexity:\anote{For future work: In a way quantum walks are based on block-encodings too. Can we improve his quantum walk as well?}
%$$\widetilde{\mathcal{O}}\left(\min\left\{\dfrac{\sqrt{w_{\mathrm{max}}}d^{3/2}}{\lambda\epsilon}\right.\right.,
%    \left.\left.\dfrac{w_{\mathrm{max}}\sqrt{d}}{\lambda^{3/2}\epsilon}\right\}\mathrm{polylog}(N)\right).$$

\begin{equation}\label{eq:WangQW}
\bigOt{\frac{\sqrt{dw_{\mathrm{max}}}}{\lambda\epsilon}\min\left\{d,\sqrt{\frac{w_{\mathrm{max}}}{\lambda}}\right\}\mathrm{polylog}(N)}.
\end{equation}

We also compare this complexity with our algorithm's complexity~\eqref{eq:sparseComp} by case separation:

\begin{itemize}
\item[(i)] When $d<\sqrt{w_{\max}/\lambda}$,  the complexity of Wang's algorithm~\eqref{eq:WangQW} is $\widetilde{\mathcal{O}}\left(\sqrt{w_{\max}}d^{3/2}\epsilon^{-1}\lambda^{-1}\right)$. Our complexity~\eqref{eq:sparseComp} is better by a factor of $\bigOt{1/\sqrt{\lambda}}$.

\item[(ii)] When $d>\sqrt{w_{\mathrm{max}}/\lambda}$, the complexity of Wang's algorithm~\eqref{eq:WangQW} is $\widetilde{\mathcal{O}}\left(w_{\mathrm{max}}\sqrt{d}\lambda^{-3/2}\epsilon^{-1}\right)$. Our complexity~\eqref{eq:sparseComp} has a worse dependence on $d$, but a better dependence on $w_{\max}$ and $\lambda$. We get a speedup as long as $\sqrt{w_{\max}/\lambda} < d \ll \sqrt{w_{\max}}/\lambda$, e.g., if $d=\Oo(\polylog(N))$, we get a speedup of $\bigOt{\sqrt{w_{\max}}/\lambda}$. 
\end{itemize}

We now consider our algorithm in the quantum data structure access model (Corollary \ref{cor:dissipated-qram1}) and compare it to Wang's algorithms.  Note that as Wang's algorithms are in the sparse-access input model, these are not directly comparable. Assume that we are in Case (1), in which case $\mu(C_G)=\nrm{C_G}_F\leq \nrm{C_G}\sqrt{N}$. The complexity of our algorithm in this model is
\begin{equation}\label{eq:qram1Comp}
\bigOt{\dfrac{1}{\epsilon}\sqrt{\dfrac{d N w_{\max}}{\lambda}}\polylog(N)}.
\end{equation}

As compared to Wang's algorithm based on linear systems~\eqref{eq:WangQLS}, our complexity~\eqref{eq:qram1Comp} is better for graphs with maximum degree $d\gg \sqrt[3]{N \lambda /w_{\max}}$.
With respect to Wang's quantum walk-based algorithm~\eqref{eq:WangQW} our complexity~\eqref{eq:qram1Comp} is better only in certain regimes. 

\begin{itemize}
\item[(i)] When $d<\sqrt{w_{\mathrm{max}}/\lambda}$, the complexity of Wang's algorithm~\eqref{eq:WangQW} is $\widetilde{\mathcal{O}}\left(\sqrt{w_{\mathrm{max}}}d^{3/2}\epsilon^{-1}\lambda^{-1}\right)$. Our complexity \eqref{eq:qram1Comp} is better as long as $\lambda \ll d^2/N$.

\item[(ii)] When $d>\sqrt{w_{\mathrm{max}}/\lambda}$, the complexity of Wang's algorithm~\eqref{eq:WangQW} is $\widetilde{\mathcal{O}}\left(w_{\mathrm{max}}\sqrt{d}\lambda^{-3/2}\epsilon^{-1}\right)$. Our complexity \eqref{eq:qram1Comp} is better as long as $\lambda\ll\sqrt{w_{\mathrm{max}}/N}.$ 
\end{itemize}

In Ref.~\cite{IJ15}, the authors developed a quantum algorithm for estimating effective resistance between $s$ and $t$, $R_{s,t}$, in the adjacency query model. Moreover, the weights of each edge are assumed to be in $\{0,1\}$. The algorithm estimates $R_{s,t}$ up to a multiplicative error $\epsilon$ in complexity

%$$\bigOt{\min\left\{\frac{N\sqrt{R_{s,t}}}{\epsilon^{3/2}},\frac{N\sqrt{R_{s,t}}}{\epsilon\sqrt{d\lambda}}\right\}}.$$
$$\bigOt{\frac{N\sqrt{R_{s,t}}}{\epsilon}\min\left\{\frac{1}{\sqrt{\epsilon}},\frac{1}{\sqrt{d\lambda}}\right\}}.$$
Although our models are not directly comparable to that of \cite{IJ15}, the complexity~\eqref{eq:qram1Comp} in the quantum data structure input model is better whenever $\lambda = \Omega(1)$ and $R_{s,t}\gg d^2w_{\mathrm{max}}/N$.  

%On the other hand, in the sparse-access input model~\eqref{eq:sparseComp}, our algorithm has a speedup as long as $d\ll\sqrt{N}R_{s,t}^{1/4}$.\anote{I do not think it is fair to compare something with adjacency model to our sparse-access results, I would just skip this last one. Adjacency query access has a data structure feel, I find that comparison more appropriate.}
%%%%%%%%%%%%%%%%%%%%%%%%%%%%%%%%%%%%%%%%%%%%%%%%%%%%%%%%%%%%%%%%%%%%%%%%%%%%%%%%%%%%%%%%%%%%%%%%%%%%

\section*{Acknowledgments}
The authors are grateful for Iordanis Kerendis, Anupam Prakash and Michael Walter for useful discussions.
This work was initiated when SC was a member of the Physics of Information and Quantum Technologies group at IT Lisbon and was visiting QuSoft, CWI Amsterdam. SC was then supported by DP-PMI and FCT (Portugal) through scholarship SFRH/BD/52246/2013 and by FCT (Portugal) through programmes PTDC/POPH/POCH and projects UID/EEA/50008/2013, IT/QuNet, ProQuNet, partially funded by EU FEDER, and from the JTF project NQuN (ID 60478).
SC is supported by the Belgian Fonds de la Recherche Scientifique - FNRS under grants no F.4515.16 (QUICTIME) and R.50.05.18.F (QuantAlgo). AG is supported by ERC Consolidator Grant QPROGRESS.  SJ is supported by an NWO WISE Grant and NWO Veni Innovational Research Grant under project number 639.021.752.

\bibliographystyle{alphaUrlePrint}
\bibliography{qc_gily}

\appendix

\section{Technical results about block-encodings}\label{app:proofs}

In this appendix we first prove some results about products of block-encodings, then we turn to smooth-functions of Hermitian matrices. 

In order to improve the complexity of multiplication of block-encoded matrices, we invoke a result about efficiently amplifying a subnormalized block-encoding, as proposed by Low and Chuang~\cite{LowChuangHamSpectraAmp2017}. The following result is proven in~\cite{gilyenBlockMatrices}.

\begin{lemma}[Uniform block-amplification]\label{lem:uniformBlockAmp}
	Let $A\in \mathbb{R}^{M\times N}$ such that $\nrm{A}\leq 1$. If $\alpha\geq 1$ and $U$ is a $(\alpha,a,\delta)$-block-encoding of $A$ that can be implemented in time $T_U$, then there is a $(\sqrt{2},a+1,\delta+\gamma)$-block-encoding of $A$ that can be implemented in time $\bigO{\alpha(T_U+a)\log(1/\gamma)}$.  
\end{lemma}

\nonSquareProdPreAmp*
\begin{proof}
	Using Lemma~\ref{lem:uniformBlockAmp} we can implement a unitary $\widetilde{U}$ that is a $(\sqrt{2},a+1,\delta+\gamma/2)$ block-encoding of $A$ in time $\bigO{\alpha\log(1/\gamma)(T_U+a)}$. Similarly we can implement a unitary $\widetilde{V}$ that is a $(\sqrt{2},b+1,\eps+\gamma/2)$ block-encoding of $B$ in time $\bigO{\beta\log(1/\gamma)(T_V+b)}$. Using Lemma~\ref{lemma:disjointAncillaProduct} we get a unitary $W$ that is a $(2,a+b+2,\sqrt{2}(\delta+\eps+\gamma))$ block-encoding of $AB$, that can be implemented in time $\bigO{\left(\alpha(T_U+a)+\beta(T_V+b)\right)\log(1/\gamma)}$.
\end{proof}

\subsection{Error propagation of block-encodings under various operations}	\label{app:error}
	In this subsection we present bounds on how the error propagates in block-encoded matrices when we perform multiplication or Hamiltonian simulation. 
	
	First we present some results about the error propagation when multiplying block-encodings in the special case when the encoded matrices are unitaries and their block-encoding does not use any extra scaling factor. In this case one might reuse the ancilla qubits, however it introduces an extra error term, which can be bounded by the geometrical mean of the two input error bounds. The following two lemmas can be found in the work of Gilyén et al.~\cite{gilyenBlockMatrices}. 
	%For completeness we repeat their proof.

\begin{lemma}
	If $U$ is an $(1,a,\delta)$-block-encoding of an $s$-qubit unitary operator $A$, and $V$ is an $(1,a,\eps)$-block-encoding of an $s$-qubit unitary operator $B$ then $UV$ is a $(1,a,\delta+\eps+2\sqrt{\delta\eps})$-block-encoding of the unitary operator $AB$.
\end{lemma}

The above lemma suggests that if we multiply together multiple block-encoded unitaries, the error may grow super-linearly. By analysing the spreading of errors following a binary tree structure, one can show~\cite{gilyenBlockMatrices} that the error increases at most quadratically with the number of factors in the product, as stated in the following corollary.

\begin{corollary}\label{cor:blockProductPrecision}
	Suppose that $U_j$ is an $(1,a,\eps)$-block-encoding of an $s$-qubit unitary operator $W_j$ for all $j\in [K]$.
	Then $\prod_{j=1}^K U_j$ is an $(1,a,4K^2\eps)$-block-encoding of $\prod_{j=1}^K W_j$.
\end{corollary}
\begin{comment}
\begin{proof}
	First observe that for the product of two matrices we get the precision bound $3\eps$ by the above lemma. 
	If $K=2^k$ for some $k\in \mathbb{N}$. Then we can apply the above observation in a recursive fashion in a binary tree structure, to get the upper bound $3^k \eps$ on the precision, and observe that $3^k\leq 4^k=K^2$.
	
	If $K\leq 2^k$ we can just add identity operators so that we have $2^k$ matrices to multiply, which gives the precision bound $3^{\lceil\log_2 K\rceil}\eps$. We are almost done since $\forall K\in \mathbb{N}\setminus\{5\}$ we have that $3^{\lceil\log_2 K\rceil}\leq K^2$. This can be seen by case checking for $K\in [15]$, 
	% Checked by Mathematica: {Table[{k, k^2}, {k, 1, 15}], Table[{k, 3^Ceiling[Log[2, k]]} // Abs, {k, 1, 15}]} // ListPlot[#, PlotLegends -> {"Square", "Ceiling"}, ImageSize -> 600] &
	and observing that for $K\geq 16$ we have that $3^{\lceil\log_2 K\rceil}\leq 3^{\frac{5}{4}\lceil\log_2 K\rceil}=K^{\frac{5}{4}\log_2(3)}\leq K^2$. 
	% Checked by Mathematica: Log[3]/Log[2]*5/4 // N
	The $K=5$ case is also easy to check by hand.
\end{proof}
\end{comment}

The following lemma helps us to understand error accumulation in Hamiltonian simulation, which enables us to present a more generic claim in Theorem~\ref{thm:blockHamSim}.
\begin{lemma}\label{lemma:hamSimDiff}
	Suppose that $H,H'\in\mathbb{C}^s$ are Hermitian operators, then 
	$$\nrm{e^{it H}-e^{it H'}}\leq |t|\nrm{H-H'}.$$
\end{lemma}
\begin{proof}
	We recall a formula introduced by \cite{KarplusMicriWave,FeynmanOpCalc}, see also \cite[Page 181]{BellmanIntroMatrixAnal}:
	\begin{equation}\label{eq:expDerInt}
	\frac{d}{dx}e^{A(x)}=\int_{0}^{1}e^{y A(x)}\frac{d A(x)}{dx}e^{(1-y) A(x)} dy.
	\end{equation}
	Now observe that 
	\begin{align*}
	e^{it H'}-e^{it H}
	&=\int_{x=0}^{1} \frac{d}{dx}\left(e^{it(H + x(H'-H) )}\right)dx\\
	&=\int_{0}^{1} \int_{0}^{1} e^{y it(H + x(H'-H) )}it(H'-H)e^{(1-y) it(H + x(H'-H) )} dy dx. \tag{by \eqref{eq:expDerInt}}
	\end{align*}
	Finally using the triangle inequality we get that
	\begin{align*}
	\nrm{e^{it H'}-e^{it H}}
	&\leq\int_{0}^{1} \int_{0}^{1} \nrm{e^{y it(H + x(H'-H) )}it(H'-H)e^{(1-y) it(H + x(H'-H) )}} dy dx\\
	&=\int_{0}^{1} \int_{0}^{1} |t|\nrm{H'-H} dy dx\\
	&=|t|\nrm{H'-H}.\qedhere
	\end{align*}
\end{proof}

Now we restate the following result in order to better place its proof in context.

\optHamSim*
\begin{proof}
	Let $H'=\alpha(I\otimes \bra{0}^{\otimes a})U(I\otimes \bra{0}^{\otimes a})$, then $\nrm{H'-H}\leq \eps/|2t|$.
	By \cite[Theorem 1]{LowChuangQubitization2016} we can implement $V$ an $(1,a+2,\eps/2)$-block-encoding of $e^{itH'}$, with $\bigO{|\alpha t|+\log(1/\eps)}$ uses of controlled-$U$ or its inverse and with $\bigO{a|\alpha t|+a\log(1/\eps)}$ two-qubit gates. 
	By Lemma~\ref{lemma:hamSimDiff} we get that $V$ is an $(1,a+2,\eps)$-block-encoding of $e^{itH}$.
\end{proof}

Note that in order to get the optimal block-Hamiltonian simulation result, one can replace the $\log(1/\eps)$ term with the term $\frac{\log(1/\eps)}{\log(e+\log(1/\eps)/|\alpha t|)}$ in the above result and its proof. For more details see~\cite{gilyenBlockMatrices}.

\subsection{Implementing smooth functions of Block-Hamiltonians}\label{app:smooth}

Apeldoorn et al. developed some general techniques~\cite[Appendix B]{AGGW:SDP} that make it possible to implement smooth-functions of a Hamiltonian $H$, based on Fourier series decompositions and using the Linear Combinations of Unitaries (LCU) Lemma~\cite{BerryChilds:hamsimFOCS}. The techniques developed in~\cite[Appendix B]{AGGW:SDP} access $H$ only through controlled-Hamiltonian simulation, which we define in the following:

\begin{definition}\label{def:controlledSim}
	Let $M=2^J$ for some $J\in \mathbb{N}$, $\gamma\in\mathbb{R}$ and $\epsilon\geq0$. We say that the unitary 
	$$
	W:=\sum_{m=-M}^{M-1}\ketbra{m}{m}\otimes e^{im\gamma H}
	$$ 
	implements controlled $(M,\gamma)$-simulation of the Hamiltonian $H$, where $\ket{m}$ denotes a (signed) bitstring $\ket{b_Jb_{J-1}\ldots b_0}$ such that $m=-b_J2^J+\sum_{j=0}^{J-1}b_j2^j$. 
	%The unitary $\tilde{W}$ implements controlled $(M,\gamma,\eps)$-simulation of the Hamiltonian $H$, if
	%$$	\nrm{\tilde{W}-W}\leq \eps.	$$
\end{definition}

The following lemma shows what is the cost of implementing controlled Hamiltonian simulation, provided a block-encoding of $H$.

\begin{lemma}\label{lemma:controlledHamsin}
	Let $M=2^J$ for some $J\in \mathbb{N}$, $\gamma\in\mathbb{R}$ and $\epsilon\geq0$. Suppose that $U$ is an $(\alpha,a,\eps/|8 (J+1)^2 M \gamma|)$-block-encoding of the Hamiltonian $H$. Then we can implement a $(1,a+2,\eps)$-block-encoding of a controlled $(M,\gamma)$-simulation of the Hamiltonian $H$, with $\bigO{|\alpha  M\gamma|+J\log(J/\eps)}$ uses of controlled-$U$ or its inverse and with $\bigO{a|\alpha M\gamma|+aJ\log(J/\eps)}$ two-qubit gates.	
\end{lemma}
\begin{proof}
	We use the result of Theorem~\ref{thm:blockHamSim}, which tells us that we can implement Hamiltonian simulation of $H$ for time $t\leq M\gamma$ with $\eps/(J+1)^2$ precision using
	\begin{equation}\label{eq:hamSimQueries}
	\bigO{|\alpha M\gamma|+\log(J/\eps)}
	\end{equation}
	uses of controlled-$U$ or its inverse and with
	\begin{equation}\label{eq:hamSimGates}
	\bigO{a|\alpha M\gamma|+a\log(J/\eps)}
	\end{equation}
	two-qubit gates.
	
	Now we write the sought unitary $W$ as the product of controlled Hamiltonian simulation unitaries.
	For $b\in \{0,1\}$ let us introduce the projector $\ketbra{b}{b}_j:=I_{2^j}\otimes\ketbra{b}{b}\otimes I_{2^{J-j}}$,
	where $J=\log(M)$. Observe that 
	\begin{equation}\label{eq:cleverW}
	W=\left(\ketbra{1}{1}_{J}\otimes e^{-i2^{J}\gamma H}+\ketbra{0}{0}_{J}\otimes I\right)\prod_{j=0}^{J-1}\left(\ketbra{1}{1}_{j}\otimes e^{i2^{j}\gamma H}+\ketbra{0}{0}_{j}\otimes I\right).
	\end{equation}		
	
	We can implement an $(1,a+2,\eps/(4(J+1)^2))$-block-encoding of the $j$-th operator $e^{\pm i2^{j}\gamma H}$ in the product~\eqref{eq:cleverW} with using 
	$\bigO{\alpha 2^j\gamma + \log\left(\frac{J}{\eps}\right)}$ queries \eqref{eq:hamSimQueries} and using $\bigO{a|\alpha 2^j\gamma|+a\log(J/\eps)}$ two-qubit gates by~\eqref{eq:hamSimGates}. By Corollary~\ref{cor:blockProductPrecision} we get the sought error bound. The complexity statement easily follows by adding up the complexities.
\end{proof}

Now we invoke~\cite[Theorem~40]{AGGW:SDP} about implementing smooth functions of Hamiltonians.
The theorem is stated slightly differently in order to adapt it to the terminology used here, but the the same proof applies as for \cite[Theorem~40]{AGGW:SDP}.

\begin{theorem}[Implementing a smooth function of a Hamiltonian]\label{thm:Taylor}
	Let $x_0\in\mathbb{R}$ and $r>0$ be such that $f(x_0+x)=\sum_{\ell=0}^{\infty} a_\ell x^\ell$ for all $x\in\![-r,r]$. 
	Suppose $B>0$ and $\delta\in(0,r]$ are such that $\sum_{\ell=0}^{\infty}(r+\delta)^\ell|a_\ell|\leq B$. 
	If $\nrm{H-x_0I}\leq r$ and $\eps'\in\!\left(0,\frac{1}{2}\right]$, then we can implement a unitary $\tilde{U}$ that is a $(B,a+\bigO{\log(r\log(1/\eps')/\delta)},B \eps')$-block-encoding of $f(H)$, with a single use of a circuit $V$ which is a $(1,a,\eps'/2)$-block-encoding of controlled $\left(\bigO{r\log(1/\eps')/\delta},\bigO{1/r}\right)$-simulation of $H$, and with using $\bigO{r/\delta\log\left(r/(\delta\eps')\right)\log\left(1/\eps'\right)}$ two-qubit gates. 
\end{theorem}

Now we are ready to prove our result about implementing power functions of both negative and positive exponents.

\begin{corollary}\label{cor:NegativePowerCost}
	Let $\kappa\geq 2$, $c\in(0,\infty)$ and $H$ be an $s$-qubit Hamiltonian such that $I/\kappa\preceq H \preceq I$.\\
	Then we can implement a unitary $\tilde{U}$ that is a $(2\kappa^{c},a+\bigO{\log(\kappa^c\max\left(1,c\right)\log(\kappa^c/\eps))}, \eps)$-block-encoding of $H^{-c}$, with a single use of a circuit $V$ which is a $(1,a,\eps/(4\kappa^c))$-block-encoding of controlled $\left(\bigO{\kappa\max\left(1,c\right)\log(\kappa^c/\eps)},\bigO{1}\right)$-simulation of $H$, and with using $\bigO{\kappa\max\left(1,c\right)\log^2\left(\kappa^{1+c}/\eps\right)}$ two-qubit gates. 
\end{corollary}
\begin{proof}
	Let $f(y):=y^{-c}$ and observe that $f(1+x)=(1+x)^{-c}=\sum_{k=0}^{\infty}\binom{-c}{k}x^k$ for all $x\in(-1,1)$.
	We choose $x_0:=1$, $r:=1-1/\kappa$, $\delta:=1/(2\kappa\max\left(1,c\right))$, and observe that
	\begin{align*}
	\sum_{k=0}^{\infty}\left|\binom{-c}{k}\right|(r+\delta)^k
	&=\sum_{k=0}^{\infty}\left|\binom{-c}{k}\right|\left(1-\frac{1}{\kappa}+\frac{1}{2\kappa\max\left(1,c\right)}\right)^{\!\!k}\\
	&=\sum_{k=0}^{\infty}\binom{-c}{k}\left(\frac{1}{\kappa}\left(1-\frac{1}{2\max\left(1,c\right)}\right)-1\right)^{\!\!k}\\
	&=\kappa^c\left(1-\frac{1}{2\max\left(1,c\right)}\right)^{\!\!-c}\\	
	&\leq \underset{B:=}{\underbrace{2 \kappa^c}}.	
	%Checked by Mathemtica: (1 - 1/(2 Max[1, c]))^(-c) // Plot[{2, #, Sqrt[E]}, {c, 0, 5}] &
	\end{align*} By choosing $\eps':=\eps/(2\kappa^c)$ we get the results by invoking Theorem~\ref{thm:Taylor}.
\end{proof}

\negPower*
\begin{proof}
	By Lemma \ref{lemma:controlledHamsin}, we can implement a $(1,a+2,\frac{\eps}{4\kappa^c})$-block-encoding $V$ of $(t,\gamma)$-controlled simulation of $H$, for $t=\bigO{\kappa\mathrm{max}(1,c)\log\frac{\kappa^c}{\eps}}$ and $\gamma=\bigO{1}$, in cost
	$$T_V = \bigO{\left(\alpha t + \log t \log\frac{\kappa^c \log t}{\eps}\right)(a+T_U)}=\bigO{\left(\alpha \kappa (1+c) \log^2\frac{\kappa^{1+c} }{\eps}\right)(a+T_U)}.$$
	
	Then by Corollary \ref{cor:NegativePowerCost}, we can implement a $(2\kappa^c,a+\bigO{\log(\kappa^c\mathrm{max}(1,c)\log\frac{\kappa^c}{\eps}},\eps)$-block-encoding of $H^{-c}$ in gate complexity $T_V+\bigO{\kappa\mathrm{max}(1,c)\log^2\frac{\kappa^{1+c}\mathrm{max}(1,c)}{\eps}}$, which gives total cost:
	\begin{align*}
	\bigO{\left(\alpha \kappa (1+c) \log^2\frac{\kappa^{1+c} }{\eps}\right)(a+T_U)}.
	\end{align*}
\end{proof}

Similarly we prove a result about implementing power functions of positive exponents.
\begin{corollary}\label{cor:PositivePowerCost}
	Let $\kappa\geq 2$, $c\in(0,1]$ and $H$ be an $s$-qubit Hamiltonian such that $I/\kappa\preceq H \preceq I$.\\
	Then we can implement a unitary $\tilde{U}$ that is a $(2,a+\bigO{\log\log(1/\eps)}, \eps)$-block-encoding of $H^{c}$, with a single use of a circuit $V$ which is a $(1,a,\eps/4)$-block-encoding of controlled $\left(\bigO{\kappa\log(1/\eps)},\bigO{1}\right)$-simulation of $H$, and with using $\bigO{\kappa\log\left(\kappa/\eps\right)\log\left(1/\eps\right)}$ two-qubit gates. 
\end{corollary}
\begin{proof}
	Let $f(y):=y^{c}$ and observe that $f(1+x)=(1+x)^{c}=\sum_{k=0}^{\infty}\binom{c}{k}x^k$ for all $x\in[-1,1]$.
	We choose $x_0:=1$, $r:=1-1/\kappa$, $\delta:=1/\kappa$, and observe that
	\begin{align*}
	\sum_{k=0}^{\infty}\left|\binom{c}{k}\right|(r+\delta)^k
	&=	\sum_{k=0}^{\infty}\left|\binom{c}{k}\right|\\
	&=1- \sum_{k=1}^{\infty}\binom{c}{k}(-1)^k\\
	&=2- \sum_{k=0}^{\infty}\binom{c}{k}(-1)^k\\
	&=2- f(1-1)\\	
	&= \underset{B:=}{\underbrace{2}}.	
	%Checked by Mathemtica: (1 + 1/(2 Max[1, c]))^(c) // Plot[{2, #, Sqrt[E]}, {c, 0, 10}] &
	\end{align*} By choosing $\eps':=\eps/2$ we get the results by invoking Theorem~\ref{thm:Taylor}.
\end{proof}   

\posPower*
\begin{proof}
By Lemma \ref{lemma:controlledHamsin}, we can implement a $(1,a+2,\frac{\eps}{4})$-block-encoding $V$ of $(t,\gamma)$-controlled simulation of $H$, for $t=\bigO{\kappa\log(1/\eps)}$ and $\gamma=\bigO{1}$, in cost
$$T_V = \bigO{\left(\alpha t + \log t \log\frac{\log t}{\eps}\right)(a+T_U)}.$$

Then by Corollary \ref{cor:PositivePowerCost}, we can implement a  $(2,a+\bigO{\log\log(1/\eps)}, \eps)$-block-encoding of $H^{c}$ in gate complexity $T_V+\bigO{\kappa\log\left(\kappa/\eps\right)\log\left(1/\eps\right)}$. The result follows.
\end{proof}
 
%\anote{It should be possbile to improve the time of the Hamiltonian simulation to $c\kappa^{1-c}$. -- I am not sure any longer, $H$^0 seems to require $\Omega{\kappa}$}
\subsection{Variable time quantum algorithm for implementing negative powers of Hermitian matrices}
\label{app:negpow-vtaa}
\negpowvtaa*
\begin{proof}
We follow the same argument as Theorem~\ref{thm:qls-vtaa}, except that $\epsilon'=\eps/(m\alpha_{\mathrm{max}})$ where $\alpha_{max}=\mathcal{O}(\kappa^c)$. This gives us that 
\begin{align*}
%\label{eq:t-max-neg-pow-vtaa}
T_{\mathrm{max}}&=\bigO{\alpha\kappa q\log^2\left(\frac{q\kappa^{q}}{\epsilon'}\right)(a+T_U)}
                =\bigO{\alpha q \kappa\log^2\left(\frac{q\kappa^{q}}{\eps}\right)(a+T_U)},
\end{align*}
and $T'_{\mathrm{max}}=\bigO{\kappa}$.
We can calculate the $l_2$-averaged stopping time of $\mathcal{A}$, $\nrm{T}_2$ as
\begin{align*}
\nrm{T}_2^2&=\sum_j p_j t_j^2\\
           &=\sum_k |c_k|^2\sum_j \left(\nrm{\Pi_{C_j}\mathcal{A}_j\dots \mathcal{A}_1\ket{\lambda_k}_I\ket{0}_{CFPQ}}^2t_j^2\right)\\
           &=\mathcal{O}\left(\alpha^2 q^2 (a+T_U)^2 \sum_k \dfrac{|c_k|^2}{\lambda_k^2}\log^4\dfrac{q\kappa^c\log\kappa}{\eps\lambda^{q}_k} \right)\\
%\label{eq:neg-pow-l2-stopping-time}
\implies \nrm{T}_2&\leq \alpha q (a+T_U)\log^2\left(\dfrac{q\kappa^{q}}{\eps}\right)\sqrt{\sum_k\dfrac{|c_k|^2}{\lambda_k^2}}.
\end{align*} 
Also the success probability, $p_{\mathrm{succ}}$ can be written as 
\begin{align*}
\sqrt{p_{\mathrm{succ}}}&=\nrm{\Pi_F \dfrac{H^{-c}}{\alpha_{\mathrm{max}}}\ket{\psi}_I\ket{\Phi}_{CFPQ}}+\mathcal{O}\left(m\epsilon'\right)\\
                        &=\dfrac{1}{\alpha_{\mathrm{max}}}\left(\sum_k\dfrac{|c_k|^2}{\lambda^{2c}_k}\right)^{1/2}+\mathcal{O}\left(\dfrac{\eps}{\alpha_{\mathrm{max}}}\right)\\
%\label{eq:neg-pow-succ-prob-vtaa-1}
                        &\geq\Omega\left(\dfrac{1}{\kappa^c}\right)\left(\sum_k\dfrac{|c_k|^2}{\lambda^{2c}_k}\right)^{1/2}.
\end{align*}
When $c\geq 1$ we have:
\begin{equation}
\label{eq:ratio-c-greater-one}
\sqrt{\dfrac{\sum_k |c_k|^2}{\lambda^{2c}_k}}\geq \sqrt{\dfrac{\sum_k |c_k|^2}{\lambda^2_k}}.
\end{equation} 
Thus, the success probability satisfies:
\begin{align*}
%\label{eq:neg-pow-succ-prob-vtaa-c-greater-one}                        
  \sqrt{p_{\mathrm{succ}}}\geq \Omega\left(\dfrac{1}{\kappa^c}\right)\sqrt{\sum_k\dfrac{|c_k|^2}{\lambda^{2}_k}}.
\end{align*}
On the other hand, using that $|\kappa\lambda_k|\geq1$, term-by-term comparison reveals that for all $c\in[0,1]$
\begin{align}
\sum_k\dfrac{ |c_k|^2}{(\kappa\lambda_k)^{2c}}&\underset{\Downarrow}{\geq} \sum_k\dfrac{ |c_k|^2}{(\kappa\lambda_k)^2}\nonumber\\
\sqrt{\dfrac{\sum_k |c_k|^2}{\lambda^{2c}_k}}&\geq \kappa^{-1+c} \sqrt{\dfrac{\sum_k |c_k|^2}{\lambda^2_k}}.\label{eq:ratio-c-lesser-one} 
\end{align}
So for this case, the success probability is bounded as
\begin{align*}
%\label{eq:neg-pow-succ-prob-vtaa-c-lesser-one}                        
  \sqrt{p_{\mathrm{succ}}}\geq \Omega\left(\dfrac{1}{\kappa}\right)\sqrt{\sum_k\dfrac{|c_k|^2}{\lambda^{2}_k}}.
\end{align*}
By combining \eqref{eq:ratio-c-greater-one} and \eqref{eq:ratio-c-lesser-one}, we have that for $c\in (0,\infty)$ 
\begin{align*}
%\label{eq:neg-pow-succ-prob-vtaa}                        
  \sqrt{p_{\mathrm{succ}}}\geq \Omega\left(\dfrac{1}{\kappa^{q}}\right)\sqrt{\sum_k\dfrac{|c_k|^2}{\lambda^{2}_k}}.
\end{align*}
The final complexity of applying VTAA is given by Theorem~\ref{thm:efficient-vtaa-vtae} as (neglecting constants):
\begin{align*}
&T_{\max}+T_{\psi}+\frac{\left(\nrm{T}_2+T_\psi\right)\log(T'_{\max})}{\sqrt{p_\mathrm{succ}}}\nonumber
\\
=&\alpha q\kappa\log^2\left(\frac{q\kappa^{q}}{\eps}\right)(a+T_U)+\left(\alpha  q \kappa^{q}  (a+T_U)\log^2\left(\dfrac{q\kappa^{q}}{\eps}\right)+\kappa^{c}T_\psi\right)\log\left(\kappa\right)\nonumber
\\
%\label{eq:neg-pow-complexity-vtaa}
=&\mathcal{O}\left(\left(\alpha q \kappa^{q} (a+T_U)\log^2\left(\dfrac{q\kappa^{q}}{\eps}\right)+\kappa^{c}T_\psi\right)\log\left(\kappa\right)\right).
\end{align*}

The second part follows from Corollary \ref{cor:qls-vtae}.  
%\vskip-20pt % to move qed symbol up
\end{proof}

\end{document}	
